\documentclass[preprint]{elsarticle}

\usepackage{amsmath}
\usepackage{amsthm}
\usepackage{amssymb}
\usepackage{bussproofs}
\usepackage{xfrac}
\usepackage{mathrsfs}
\usepackage{verbatim}
\usepackage{color}
\usepackage[dvipsnames]{xcolor}

\usepackage{hyperref}

\theoremstyle{definition}
\newtheorem{definition}{Definition}

\theoremstyle{plain}
\newtheorem{theorem}{Theorem}
\newtheorem{lemma}{Lemma}
\newtheorem{corollary}{Corollary}
\newtheorem{proposition}[theorem]{Proposition}

\theoremstyle{remark}
\newtheorem{remark}{Remark}
\newtheorem{example}{Example}

\newcommand{\LK}{\textbf{LK}}
\newcommand{\seq}{\vdash}
\newcommand{\Gplus}{\textbf{G3c}$^+$}
\newcommand{\all}{\ensuremath{\forall}}
\newcommand{\alll}{\ensuremath{\forall\colon \text{l}}}
\newcommand{\allr}{\ensuremath{\forall\colon \text{r}}}
\newcommand{\ex}{\ensuremath{\exists}}
\newcommand{\exl}{\ensuremath{\exists\colon \text{l}}}
\newcommand{\exr}{\ensuremath{\exists\colon \text{r}}}
\newcommand{\Gcal}{{\cal G}}
\newcommand{\Ccal}{\mathcal{C}}
\newcommand{\Ehat}{\hat{E}}
\newcommand{\DNF}{\textit{DNF}}
\newcommand{\UC}{\textit{UC}}
\newcommand{\AllCl}[1]{\ensuremath{\mathcal{I}(#1)}}
\newcommand{\seHs}{SEHS }
\newcommand{\EH}[1]{\textit{EH}_{#1}}
\newcommand{\AllMul}{$\forall$-multiplicity }
\newcommand{\ExMul}{$\exists$-multiplicity }

\newcommand{\NA}{\textit{NA}}
\newcommand{\NB}{\textit{NB}}
\newcommand{\RI}[1]{\ensuremath{\mathcal{J}(#1)}}

\newcommand{\tuples}[1]{\ensuremath{\vec{#1}}}
\newcommand{\lengthOfTuple}[1]{\ensuremath{\textit{ar}(#1)}}

\journal{Theoretical Computer Science}

\begin{document}

\begin{frontmatter}



\title{\tnoteref{ProjektMichael}The problem of $\Pi_2$-cut-introduction}
\tnotetext[ProjektMichael]{Funded by FWF project W1255-N23.}

\author[Alex]{Alexander Leitsch}
\ead{leitsch@logic.at}

\author[Michael]{Michael Lettmann}
\ead{michael.lettmann@tuwien.ac.at}

\address[Alex]{Institute of Computer Languages, Technische Universit\"at Wien, Vienna, 
Austria}
\address[Michael]{Institute of Information Systems, Technische Universit\"at Wien, 
Vienna, Austria}

\begin{abstract}
We describe an algorithmic method of proof compression based on the introduction of 
$\Pi_2$-cuts into a cut-free LK-proof. The current approach is based on an inversion of Gentzen's cut-elimination method and extends former methods for introducing $\Pi_1$-cuts. The Herbrand instances of a cut-free proof $\pi$ of a sequent $S$ are described by a grammar $G$ which encodes  substitutions defined in the elimination of quantified cuts. We present an algorithm which, given a grammar $G$, constructs a $\Pi_2$-cut formula $A$ and a proof $\pi'$ of $S$ with one cut on $A$. It is shown that, by this algorithm, we can achieve an exponential proof compression. 
\end{abstract}

\begin{keyword}
Proof Theory \sep Automated Deduction \sep Proof Grammar \sep Sequent Calculus \sep 
Lemma Generation
\end{keyword}

\end{frontmatter}

\section{Introduction}\label{sec.intro}

The backbone of mathematical theories are theorems and their proofs.\! Proofs are typically structured in a way that they are using lemmas, auxiliary results proven in advance. The interdependence of proofs and lemmas define a structure which makes up the true theory of a mathematical subject. In proof theory the {\em elimination} of lemmas (appearing in the form of cut-elimination) is a key technique for proving consistency~\cite{Gentzen.1935} and proof mining~\cite{Leitsch.2015}. One of the most prominent results in proof theory is Gentzen's cut-elimination theorem~\cite{gentzenG1934-35aa}; its proof is based on a procedure of stepwise simplification of cut formulas and permutation of inferences, ending up in a final elimination of ``simple'' cut-derivations. This procedure yields cut-free proofs (proofs without cut rules) which are somehow artificial objects in the sense that they hardly appear in real mathematics. However, there exists a rich source of cut-free proofs - automated theorem provers. Theorem proving programs based on the tableau method~\cite{Haehnle.2001} produce proofs which can be interpreted as cut-free proofs in sequent calculus. On the other hand, resolution theorem provers are capable of producing lemmas; but these are simple universally closed disjunctions of literals. In particular resolution provers {\em cannot} produce lemmas with alternating quantifiers, e.g. of $\Pi_2$-type. In fact also resolution proofs allow an easy extraction of Herbrand conjunctions  from which cut-free proofs can be easily constructed. Proof search profits from the use of analytic calculi as there is no need to construct formulas which are not subformulas of the problem (there is, however, the possibility to use analytic cuts in combination with structural clause form transformation~\cite{BEL.2001}). The price to be payed for efficient proof search are long and unstructured proofs which are hard to interpret. This suggests a postprocessing of mechanically generated proofs by introducing additional structure. One method of structuring these proofs consists in the introduction of lemmas. 

\smallskip

Work on cut-introduction can be found at various places in the literature.
Closest to our work are other approaches which abbreviate or structure
{\em given input proofs}: in~\cite{WoltzenlogelPaleo10Atomic} an
algorithm for the introduction of atomic cuts is developed that is capable of exponential proof compression.
There exist several contributions to proof compression by cut-introduction in propositional logic: a method defined in~\cite{Finger07Equal} is
shown to never increase the size of proofs more than polynomially, in~\cite{dagostinoM2008aa} compression by cut-introduction is described in the more general context of cut-based abduction; the paper~\cite{dagostinoM2013aa} presents a general framework for theorem proving with analytic and bounded cut. 
Another approach to the compression of first-order proofs is based on 
introduction of definitions for abbreviating terms and can be found in~\cite{Vyskocil10Automated}.

\smallskip

This paper should be also considered part of a large
body of work on the generation of non-analytic formulas that has been carried out
by numerous researchers in various communities. Methods for lemma generation are of crucial importance
in inductive theorem proving where, frequently, generalization is needed~\cite{Bundy01Automation};
see e.g.~\cite{Ireland96Productive} for a method in the context of rippling~\cite{Bundy05Rippling}
which is based on failed proof attempts.
In automated theory formation~\cite{Colton01Automated,Colton02Automated}, an eager
approach to lemma generation is adopted. 

\smallskip

A method of lemma generation based on proof theoretic techniques (algorithmic introduction of $\Pi_1$-cuts) has been defined in~\cite{hetzlS2014aa}. This method evolved from an investigation of the sets of terms generated by substitutions defined by the elimination of $\Pi_1$-cuts via Gentzen's procedure. These sets of terms, which characterize the instantiations needed for a Herbrand sequent, are specified by so-called $\Pi_1$-grammars (these are totally rigid acyclic tree grammars). The method of cut-introduction roughly works as follows: given the set of Herbrand instantiation terms $H$ a (compressing) $\Pi_1$-grammar $G$ is defined which generates $H$. $G$ contains potential information about the substitutions generated by cut-elimination. In the next step a solution of a second-order problem (defined by the concept of  schematic extended Herbrand sequent) yields the corresponding universal cut formulas $A_1,\ldots,A_n$. From these formulas and the quantifier instantiations defined by the grammar a proof with the cut formulas $A_1,\ldots,A_n$ is eventually constructed. In~\cite{hetzlS2014aa} it is also proven that the method is capable of compressing proofs exponentially (which is the maximal compression possible via $\Pi_1$-cuts). In a recent publication~\cite{afshariB2015aa} a new grammar type ($\Pi_2$-grammar) is defined which characterizes the elimination of $\Pi_2$-cuts. $\Pi_2$-lemmas are more expressive than $\Pi_1$-lemmas: consider e.g. the property that a one-place predicate $I$ holds for infinitely many arguments which can be formalized by $\forall x \exists y.I(y) \land x<y$, $<$ being a transitive, irreflexive relation. Moreover, the complexity of cut-elimination in proofs with $\Pi_2$-cuts is superexponential. Therefore an algorithmic method for introducing $\Pi_2$-cuts via $\Pi_2$-grammars would be of major mathematical interest and could yield stronger means of proof compression.

\smallskip

In this paper we generalize the method for introducing $\Pi_1$-cuts described above to the introduction of a (single) $\Pi_2$-cut. Our starting point is a $\Pi_2$-grammar $G$ generating the Herbrand instantiation terms $H$ defined by a cut-free proof $\varphi$. From $G$ and $\varphi$ we define a unification problem which (under conditions to be defined below) is solvable and yields a $\Pi_2$-cut formula $A$ and a proof with one cut on $A$  as a solution. We prove that the method is capable of compressing proofs exponentially and discuss some experiments with the implementation of the method.

\smallskip

The paper is organized as follows: in Section~\ref{sec.prooftinfra} we describe the proof-theoretic infrastructure and recall the most important results from~\cite{hetzlS2014aa}. In a next step we extend the terminology of the $\Pi_1$ case to $\Pi_2$-cut introduction. In particular we adapt the concepts of extended Herbrand sequents and schematic extended Herbrand sequents accordingly. In Section~\ref{sec.grammars} we define schematic $\Pi_2$-grammars, a simplified version of the grammars defined in~\cite{afshariB2015aa}. A characterization of the solvability of the schematic extended Herbrand sequent, which is the key step in our method of cut-introduction, is given in Section~\ref{sec.cutintro}. This characterization admits an algorithmic generation of  $\Pi_2$-cuts (in case that a $\Pi_2$-cut corresponding to the given grammar can be introduced at all). We also show that there are schematic $\Pi_2$-grammars which do not yield $\Pi_2$-cuts at all, in contrast to the $\Pi_1$ case where {\em every} $\Pi_1$-grammar leads to a solution. Given a so called \emph{starting set} of atoms, the solvability of the $\Pi_2$-cut-introduction problem is shown to be decidable. However, the decidability of the general problem remains an open problem. The characterization of solvability described in  Section~\ref{sec.cutintro} yields a rather inefficient algorithm for cut-introduction. In Section~\ref{sec.Gunify} we define a method of constructing $\Pi_2$-cut formulas via a unification procedure based on a grammar $\Gcal$ (we call it $\Gcal^*$ unification); this method is more efficient and works whenever so called \emph{balanced solutions} of the problem exist. The methods defined so far work for the introduction of cut-formulas of the form $\forall x \exists y.A(x,y)$ ($A$ quantifier-free), $x$ and $y$ being (single) variables. In Section~\ref{sec.tuples} we generalize the methods to the construction of general $\Pi_2$-cuts defined by quantifier blocks. In Section~\ref{sec:app example} we construct an infinite sequence of proofs for which the 
method of $\Gcal^*$-unifiability can achieve an exponential proof compression by the introduction of $\Pi_2$-cuts. Finally we present some experiments with an implementation of the method in Section~\ref{sec.experiments}.

\smallskip

This work can be seen as a first step in the algorithmic introduction of cuts beyond 
$\Pi_1$. A full characterization of $\Pi_2$-cut introduction for a single cut and the 
generalization to the introduction of several $\Pi_2$-cuts are left to future work. Moreover, extending the method of $\Pi_2$-cut-introduction to proofs with analytic and bounded cuts would be desirable and would make cut-introduction more interesting for mathematical applications and in the context of proof theory.
The development of an algorithm computing maximally compressing schematic $\Pi_2$-grammars for a given set of Herbrand terms is planned for the near future.

\section{Preliminaries and Notation}
\label{sec:notation}

Most of the symbols in the following paper are explained when they are introduced the first time. But due to the high number of different symbols used in this paper we describe below those which are most frequently used. 

\smallskip

The capital letters $F,G$ are quantifier-free formulas which occur in the 
end-sequents. While $L$ is denoting a literal, $P$ denotes an atom and $Q$ 
might be a literal or an atom (if a second literal or atom is needed). Small letters, 
such as $f$ and $g$, are function symbols, $r,s,$ and $t$ are usually representations 
of terms, and the characters $r_1,\ldots ,r_m$ and $t_1,\ldots ,t_p$ are the designated terms of the observed schematic $\Pi_2$-grammar $\mathcal{G}$ (see Definition \ref{def.schem-pi2-grammar}) which are substitution instances of the designated variables (if not defined otherwise) $x,y$ of the binary cut formula $C$. The small letters $m,p$ are the corresponding natural numbers. For running indices we use $i,j,k,l,$ or $q$ and for fixed natural numbers we tend to use $n$. By an overline $\overline{\cdot}$ we denote the negation of a literal $\overline{L}$, atom $\overline{P}$, or sequent $\overline{S}$ (see Definition \ref{def:dual sequent operator}). By \tuples{t} we denote a tuple of terms. The Greek characters $\alpha, \beta_1,\ldots ,\beta_m$ denote the eigenvariables of the binary cut formula in a proof $\pi$. An arbitrary sequent is denoted by $S$, while we are denoting special sequents by $H$ (Herbrand sequent), $\textit{EH}_A$ (extended Herbrand sequent with the cut-formula $A$), $S(X)$ (schematic extended Herbrand sequent with the formula variable $X$), 
and $R$ (reduced representation). If we consider a non-tautological leaf of a proof, we use the symbols $S,S'$ or $J$. 

\smallskip

Furthermore, we will abbreviate the sets $\{ 1,\ldots ,n\}$ by $\mathbb{N}_n$. For 
terms $t$ or formulas $F$ we denote by V$(t)$ or V$(F)$ the set of free variables in $t$ and $F$ respectively. For a tuple of terms $\tuples{t}=(t_1,\ldots ,t_n)$, V$(\tuples{t})$ denotes $\bigcup_{i=1}^n\text{V}(t_i)$. For a set of tuples $T$, V$(T)$ denotes the union over all sets of free variables of the elements of $T$. If $f$ is an $n$-ary function symbol and $t_1,\ldots,t_n$ are terms we frequently write $ft_1\ldots t_n$ for $f(t_1,\ldots,t_n)$. So if $f$ is binary and g is unary $fgxa$ stands for $f(g(x),a)$. The transformation of sets of clauses $\mathcal{C}$ into formulas in disjunctive normal form will be denoted by $\DNF (\mathcal{C})$. For a given tuple of variables $\tuples{x}=(x_1,\ldots 
,x_n)$ and a tuple of terms $\tuples{t}=(t_1,\ldots ,t_n)$ we denote by $F[\tuples{x}
\backslash\tuples{t}]$ the substitution of all occurrences of $x_i$ with $i\in\mathbb{N}_n$ in a given formula (or term) $F$ by $t_i$. For abbreviation we write instead of $F[\tuples{x}\backslash\tuples{t}_1],\ldots,F[\tuples{x}\backslash\tuples{t}_k]$ simply $F[\tuples{x}\backslash T]$ with $T=\{\tuples{t}_1,\ldots ,\tuples{t}_k\}$ and $k$ being an arbitrary natural number. Production rules of grammars are of the form $\tau\to t$ and correspond to a substitution $[\tau\backslash t]$ that replaces $\tau$ with $t$ at a single position. For a list of production rules $\tau\to t_1,\ldots ,\tau\to t_n$ with the same non-terminal $\tau$ we write $\tau\to t_1\; |\;\ldots\; |\; t_n$. For sequents $S_1\colon \Gamma_1\seq\Delta_1$ and $S_2\colon\Gamma_2\seq\Delta_2$, $S_1\circ S_2$ denotes the concatenation $\Gamma_1,\Gamma_2\seq\Delta_1,\Delta_2$ of $S_1$ and $S_2$.

\section{A Motivating Example}
\label{sec:Example}

To illustrate the effect of $\Pi_2$-cuts, we present the following sequence of sequents. 
Let $n\ge 2$ be a natural number and 
\begin{align*}
A_n:= &\; \forall x.\bigl( P(x,f_1x)\lor\ldots\lor P(x,f_nx)\bigr) , \\
B:= &\; \forall x,y.\bigl( P(x,y)\to P(x,fy)\bigr) , \\
Z_n:= &\; P(x_1,fx_2)\land P(fx_2,fx_3)\land\ldots\land P(fx_{n-1},fx_n), \\ 
C_n:= &\; \forall x_1,\ldots ,x_n.\bigl( Z_n\to P(x_1,gx_n)\bigr) \text{, and} \\ 
D:= &\; \exists u,v. P(u,gv).
\end{align*}
Then the sequents $S_n\colon A_n,B,C_n\seq D $
are provable. All cut-free {\LK}-proofs of the sequents $S_n$ require more than 
$n^{n}$ quantifier inferences of $A_n,B,C_n$ and $D$. But note that there are proofs 
of the sequents 
\begin{align*}
A_n,B\seq\forall x\exists y.P(x,fy) \;\text{ and }\;
\forall x\exists y.P(x,fy),C_n\seq D
\end{align*}
using only $O(n)$ quantifier inferences. As a consequence we can prove the sequents $S_n$ 
by a linear number of instances; this compression can be achieved by the 
$\Pi_2$-cut-formula $\forall x\exists y.P(x,fy)$ (see Section \ref{sec:app example}).

\section{Proof-Theoretic Infrastructure}\label{sec.prooftinfra}

A {\em sequent} $S$ is an ordered pair of sets of formulas, written as 
$\Gamma \seq \Delta$; we call $\Gamma$ the {\em antecedent} and $\Delta$ the 
{\em succedent} of $S$. The sequent calculus we use is \textbf{G3c} together with the 
cut-rule, which we call {\Gplus};  note that \textbf{G3c} is an invertible version 
of cut-free {\LK} (see \cite{troelstraA1996aa}). In this paper we 
only consider proofs of prenex skolemized end-sequents 
(the antecedents are only universal, the succedents existential); note that this 
restriction is not essential as every sequent is provability-equivalent to such a 
normal form. Every cut-free proof $\varphi$ of a prenex end-sequent $S$ can be 
transformed into a cut-free proof $\psi$ of $S$ (without increase of proof length) 
s.t.\ $\psi$ contains a {\em midsequent} $S^*$, i.e. a sequent in $\psi$ such that all quantifier inferences in $\psi$ are below $S^*$ and all propositional ones above~\cite{gentzenG1934-35aa}. Let $S'$ be the sequent that contains only the quantifier-free formulas of $S^*$. We call $S'$ a {\em Herbrand sequent} 
of $S$. Herbrand sequents of cut-free proofs will play a crucial role in our 
cut-introduction method.
\begin{example}
Let $S\colon P(a)\lor P(b)\lor P(c)\seq \exists x.P(x)$ be a sequent where $P$ is a unary predicate symbol and $a,b,c$ are terms. A cut-free proof $\psi$ is 
\begin{center}
\begin{fCenter}
\Axiom$ \fCenter $
\LeftLabel{\textit{Ax}}
\UnaryInf$P(a) \;\fCenter\seq P(a),P(b),P(c),\Delta$
\Axiom$ \;\fCenter\pi $
\UnaryInf$P(b)\lor P(c) \;\fCenter\seq P(a),P(b),P(c),\Delta$
\LeftLabel{$\lor\colon$l}
\BinaryInf$P(a)\lor P(b)\lor P(c) \;\fCenter\seq P(a),P(b),P(c),\Delta$
\LeftLabel{\exr}
\UnaryInf$P(a)\lor P(b)\lor P(c) \;\fCenter\seq P(a),P(b),\Delta$
\LeftLabel{\exr}
\UnaryInf$P(a)\lor P(b)\lor P(c) \;\fCenter\seq P(a),\Delta$
\LeftLabel{\exr}
\UnaryInf$P(a)\lor P(b)\lor P(c) \;\fCenter\seq \Delta$
\DisplayProof
\end{fCenter}
\end{center}
where $\pi=$
\begin{center}
\begin{fCenter}
\Axiom$ \fCenter $
\LeftLabel{\textit{Ax}}
\UnaryInf$P(b) \;\fCenter\seq P(a),P(b),P(c),\Delta$
\Axiom$ \fCenter $
\RightLabel{\textit{Ax}}
\UnaryInf$P(c) \;\fCenter\seq P(a),P(b),P(c),\Delta$
\LeftLabel{$\lor\colon$l}
\BinaryInf$P(b)\lor P(c) \;\fCenter\seq P(a),P(b),P(c),\Delta$
\DisplayProof
\end{fCenter}
\end{center}
and $\Delta :=\exists x.P(x)$. The midsequent $M$ of $\psi$ is 
\[P(a)\lor P(b)\lor P(c)\seq P(a),P(b),P(c),\Delta\]
and therefore, 
\[P(a)\lor P(b)\lor P(c)\seq P(a),P(b),P(c)\]
is a Herbrand sequent of $S$.
\end{example}
Let $\tuples{x}=(x_1,\ldots ,x_n)$ and $\tuples{t}=(t_1,\ldots ,t_n)$ are tuples of terms;  
then $[\tuples{x}\backslash\tuples{t}]$ denotes the substitution $[x_1\backslash t_1, 
\ldots ,x_n\backslash t_n]$. Furthermore, we denote by $|_i$ for $i\in\mathbb{N}_n$ 
the projection that gives the $i$-th element of a tuple $\tuples{t}|_i:=t_i$. For the 
complexity measurement of Herbrand sequents we have to count the minimal number of instantiations needed to introduce all terms.
\begin{definition}[Instantiation complexity]
Let $\{\vec{t}_1,\ldots,\vec{t}_k\}$ be a set of $n$-tuples of ground terms and let $P$ be an $n$-ary predicate symbol. Then we define the instantiation complexity of $T$, denoted by $\sharp T$, as the minimal number of quantifier inferences in a cut-free proof of the sequent 
$$\forall x_1 \cdots x_n.P(x_1,\ldots,x_n) \seq P(\vec{t}_1) \land \cdots \land P(\vec{t}_k).$$
\end{definition}
\begin{example}
Let $\vec{a}=(r,s,t)$ and $\vec{b}=(r,s,s)$ be tuples of terms $r,s,t$ and $T=\{\vec{a},\vec{b}\}$. Then 
\[\sharp T = 4\]
because of 
\begin{center}
\begin{fCenter}
\Axiom$ \fCenter $
\RightLabel{\textit{Ax}}
\UnaryInf$P(r,s,t),P(r,s,s),\Gamma_3 \;\fCenter\seq P(\vec{a})$
\Axiom$ \fCenter $
\RightLabel{\textit{Ax}}
\UnaryInf$P(r,s,t),P(r,s,s),\Gamma_3 \;\fCenter\seq P(\vec{b})$
\RightLabel{$\land\colon$r}
\BinaryInf$P(r,s,t),P(r,s,s),\Gamma_3 \;\fCenter\seq P(\vec{a})\land P(\vec{b})$
\RightLabel{\alll}
\UnaryInf$P(r,s,t),\forall x_3.P(r,s,x_3),\Gamma_2 \;\fCenter\seq P(\vec{a})\land P(\vec{b})$
\RightLabel{\alll}
\UnaryInf$\forall x_3.P(r,s,x_3),\Gamma_2 \;\fCenter\seq P(\vec{a})\land P(\vec{b})$
\RightLabel{\alll}
\UnaryInf$\forall x_2,x_3.P(r,x_2,x_3),\Gamma_1 \;\fCenter\seq P(\vec{a})\land P(\vec{b})$
\RightLabel{\alll}
\UnaryInf$\forall x_1,x_2,x_3.P(x_1,x_2,x_3) \;\fCenter\seq P(\vec{a})\land P(\vec{b})$
\DisplayProof
\end{fCenter}
\end{center}
\begin{align*}
\Gamma_1 := &\; \{\forall x_1,x_2,x_3.P(x_1,x_2,x_3)\} \\
\Gamma_2 := &\; \{\forall x_2,x_3.P(r,x_2,x_3)\}\cup\Gamma_1 \\
\Gamma_3 := &\; \{\forall x_3.P(r,s,x_3)\}\cup\Gamma_2
\end{align*}
being a proof with the minimal number of quantifier inferences.
\end{example}
The capital letters $F$ and $G$ are quantifier free formulas used to define the end sequent $\forall\vec{x}.F\seq\exists\vec{y}.G$ for the rest of this paper.
\begin{definition}[Herbrand sequent]\label{def.Herbrandsequent}
Let $S\colon \forall \tuples{x}.F\seq\exists\tuples{y}.G$ be a given sequent, where 
$\tuples{x} = (x_1,\ldots,x_m)$, $\tuples{y} = (y_1,\ldots,y_l)$ and let
\begin{align*}
H:=F[\tuples{x}\backslash\tuples{t}_1],\ldots ,F[\tuples{x}\backslash\tuples{t}_k]\seq 
G[\tuples{y}\backslash\tuples{t}_{k+1}],\ldots ,G[\tuples{y}\backslash\tuples{t}_n]
\end{align*}
be a valid sequent 
where $F[\tuples{x}\backslash\tuples{t}_i]$ for $i\in\mathbb{N}_k$ are instances of $F$ 
and $G[\tuples{y}\backslash\tuples{t}_j]$ for $j\in\{ k+1,\ldots ,n\}$ are instances of $G$. 
Then we call $H$ a {\em Herbrand sequent} of $S$. The {\em complexity} of $H$ is 
defined as 
\begin{align*}
|H|:=\sharp\{\tuples{t}_1,\ldots ,\tuples{t}_k\} +\sharp\{\tuples{t}_{k+1},\ldots ,\tuples{t}_n\} .
\end{align*}
\end{definition}

We could further simplify the sequents in Definition \ref{def.Herbrandsequent} to 
sequents with an empty succedent, since every sequent is provability-equivalent to 
such a normal form as well. But for the sake of readability of the definitions in the 
following sections we will always consider sequents of this form. Otherwise, we would 
have to include a discussion about when two subformulas are \emph{separable} within a 
sequent. To give an notion of separable subformulas assume a formula $A$ with 
subformulas $B$ and $C$ within a sequent $S$. We would call $B$ and $C$ separable 
if there are formulas $A_1$ and $A_2$ such that $B$ is a subformula of $A_1$, 
$C$ is a subformula of $A_2$ and there is a unary rule in \textbf{G3c} where $A_1$ and 
$A_2$ are the active parts of the premise and $A$ is the active part of the conclusion. 
In the sequent $\seq B\lor (C\land D)$ are $B$ and $C$ separable but not $C$ and $D$. 
\begin{definition}\label{de.wsquant}
$\forall$-left inferences $\alll$ and $\exists$-right inferences $\exr$ in a proof are 
called {\em weak} quantifier inferences, $\forall$-right inferences 
$\allr$ and $\ex$-left inferences $\exl$ are called {\em strong}.
\end{definition}

We measure the quantifier-complexity of a proof $\varphi$ by the number of weak 
quantifier-inferences in $\varphi$. Note that all quantifier inferences on ancestors 
of the end-sequent are weak, and multiple uses of quantified formulas in cuts is 
necessary only for formulas starting with weak quantifiers. Hence in any such proof 
the number of strong quantifier-inferences is less or equal to the number of weak 
quantifier-inferences (see \cite{takeutiG1987aa}). 

\begin{definition}
\label{def.quantifiercomp}
Let $\pi$ be a proof in {\Gplus}; then the {\em quantifier-complexity} of $\pi$ 
is defined as the number of weak quantifier inferences in $\pi$. We write $|\pi |_q=n$ 
if $\pi$ has quantifier-complexity $n$. 
\end{definition}
The following theorem is a variant of Theorem 1 in \cite{hetzlS2014aa}.
\begin{theorem}
Assume a sequent $S:\forall\tuples{x}.F\seq\exists\tuples{y}.G$. There is a Herbrand-sequent $H$ of $S$ with $|H|=n$ iff there exists a cut-free proof $\pi$ of $S$ such that $|\pi |_q = n$.
\end{theorem}
\begin{proof}
A Herbrand sequent describes exactly the terms we have to introduce by weak 
quantifier inferences. Let $H$ be a Herbrand sequent of $S$ with $|H|=n$. Then a 
cut-free proof $\pi$ with $|\pi |_q=n$ can be constructed in the following way: 
apply first all propositional inferences and afterwards all quantifier rules. 

\smallskip

Let $\pi$ be a cut-free proof of $S$. Then different terms for a given position 
of an atom can only be produced by weak quantifier inferences. Hence, the number of 
weak quantifier inferences in $\pi$ is equal to the number of different terms obtained 
by substitution, and therefore $|\pi |_q=|H|$ for $H$ being the Herbrand sequent 
obtained from $\pi$. 
\end{proof}

We define the notion of an extended Herbrand sequent as in~\cite{hetzlS2014aa}; for 
simplicity we do not consider blocks of quantifiers in the cuts, but only formulas 
of the form $\all x \ex y.A$ where $A$ is quantifier-free, $V(A) \subseteq \{x,y\}$, and $V(A)$ denotes the set of variables in $A$. As in the case of $\all$-cuts, extended Herbrand sequents represent proofs with cuts by encoding the cuts by implication formulas. As we consider only the introduction of a single $\Pi_2$ cut, we need only one formula for coding the cut.

If $U$ is a set of term tuples $\{\tuples{t}_1,\ldots , 
\tuples{t}_m\}$ and $F$ is a formula $F[\tuples{x}\backslash U]$ stands for the set of instances $\{ F[\tuples{x}\backslash\tuples{t}_1],\ldots ,F[\tuples{x}\backslash\tuples{t}_m]\}$. 
\begin{definition}[Extended Herbrand-sequent]\label{def.extHseq}
Let $S$ be a sequent of the form $\forall\tuples{x}.F\seq\exists\tuples{y}.G$ (with $\tuples{x}= 
(x_1,\ldots ,x_k)$ and $\tuples{y}=(y_1,\ldots ,y_l)$) and $A$ be a 
quantifier-free formula with $V(A) \subseteq \{x,y\}$. Let $U_1:=\{\tuples{u}_1,\ldots 
,\tuples{u}_{c}\}$ be a set of term tuples of the length $k$, $U_2:=\{\tuples{v}_1 
,\ldots ,\tuples{v}_d\}$ be a set of term tuples of the length $l$. Let 
$\beta_1,\ldots ,\beta_m,\alpha$ be variables and $t_i$ for 
$i\in\mathbb{N}_p$, $r_j$ for $j\in\mathbb{N}_m$ be terms s.t.\ 
\begin{align*}
& V(U_1) \subseteq \{\alpha\},\ V(U_2) \subseteq \{\beta_1,\ldots,\beta_m\},\\
& V(t_i)\subseteq\{\alpha\}\text{ for all }i, \ V(r_j)\subseteq\{\beta_1,\ldots 
,\beta_{j-1}\}\text{ for }j\ge 2,\text{ and } \ V(r_1)=\emptyset .
\end{align*}
Then the sequent 
\begin{align*}
& \EH{A} := F[\tuples{x}\backslash U_1],\bigvee\limits_{i=1}^pA[x\backslash\alpha 
,y\backslash t_i]\to 
\bigwedge\limits_{j=1}^mA[x\backslash r_j,y\backslash \beta_j]\seq 
G[\tuples{y}\backslash U_2]
\end{align*}
is called an \emph{extended Herbrand-sequent} of $S$ if $\EH{A}$ is a tautology. \\ 
The complexity of an extended Herbrand sequent $\EH{A}$ is 
defined as $|\EH{A}|=\sharp U_1+\sharp U_2+p+m$.
\end{definition}
For simplicity we will assume that the formulas $F$ and $G$ are constructed such that 
\begin{align*}
& F[\tuples{x}\backslash U_1],\bigvee\limits_{i=1}^pA[x\backslash\alpha 
,y\backslash t_i]\to 
\bigwedge\limits_{j=1}^mA[x\backslash r_j,y\backslash \beta_j]\seq 
G[\tuples{y}\backslash U_2]
\end{align*}
being a tautology implies the provability of the sequents 
\begin{align*}
F[\tuples{x}\backslash U_1]\seq\bigvee\limits_{i=1}^pA[x\backslash\alpha 
,y\backslash t_i]
\end{align*}
and
\begin{align*}
\bigwedge\limits_{j=1}^mA[x\backslash r_j,y\backslash \beta_j]\seq G[\tuples{y}\backslash U_2].
\end{align*}
This can be done by extending the formulas $F$ and $G$. This is only done for simpler reasoning. All statements and definitions also hold for the original case.
\begin{example}
\label{exa:extended Herbrand}
Let $u,v,w,x,y$ be variables, $\alpha ,\beta_1,\beta_2$ be eigenvariables, $c$ be a constant, $f,g,h$ be unary function symbols, and $P$ be a binary predicate symbol. Consider the following proof with a single $\Pi_2$ cut. 
\begin{center}
\begin{fCenter}
\Axiom$ \;\fCenter \pi_l $
\Axiom$ \;\fCenter \pi_r $
\LeftLabel{\textit{Cut}}
\BinaryInf$\forall u.P(u,fu)\lor P(u,gu) 
\;\fCenter\seq 
\exists v,w.P(c,v)\land P(hv,w)$
\DisplayProof
\end{fCenter}
\end{center}
with $\pi_l:=$ 
\begin{center}
\begin{fCenter}
\Axiom$ \fCenter $
\LeftLabel{\textit{Ax}}
\UnaryInf$P(\alpha ,f\alpha ),\Gamma \;\fCenter\seq 
P(\alpha ,f\alpha ),\Delta_4$
\Axiom$ \fCenter $
\RightLabel{\textit{Ax}}
\UnaryInf$P(\alpha ,g\alpha ),\Gamma \;\fCenter\seq 
P(\alpha ,g\alpha ),\Delta_5$
\LeftLabel{$\lor\colon$l}
\BinaryInf$P(\alpha ,f\alpha )\lor P(\alpha ,g\alpha ) ,\Gamma
\;\fCenter\seq 
P(\alpha ,f\alpha ),P(\alpha ,g\alpha ),\Delta_3$
\LeftLabel{$\alll$}
\UnaryInf$\forall u.P(u,fu)\lor P(u,gu) 
\;\fCenter\seq 
P(\alpha ,f\alpha ),P(\alpha ,g\alpha ),\Delta_3$
\LeftLabel{$\exr$}
\UnaryInf$\forall u.P(u,fu)\lor P(u,gu) 
\;\fCenter\seq 
P(\alpha ,f\alpha ),\exists y.P(\alpha ,y),\Delta_2$
\LeftLabel{$\exr$}
\UnaryInf$\forall u.P(u,fu)\lor P(u,gu) \;\fCenter\seq \exists y.P(\alpha ,y),\Delta_2$
\LeftLabel{$\allr$}
\UnaryInf$\forall u.P(u,fu)\lor P(u,gu) \;\fCenter\seq \forall x\exists y.P(x,y),\Delta_1$
\DisplayProof
\end{fCenter}
\end{center}
$\pi_r:=$ 
\begin{center}
\begin{fCenter}
\Axiom$ \fCenter $
\LeftLabel{\textit{Ax}}
\UnaryInf$\Lambda_5,P(c,\beta_1) 
\,\fCenter\seq\! 
P(c,\beta_1),\Xi_2$
\Axiom$ \fCenter $
\RightLabel{\textit{Ax}}
\UnaryInf$\Lambda_6,P(h\beta_1,\beta_2) 
\,\fCenter\seq\! 
P(h\beta_1,\beta_2),\Xi_2$
\LeftLabel{$\land\colon$r}
\BinaryInf$\Lambda_4,P(c,\beta_1),P(h\beta_1,\beta_2) 
\;\fCenter\seq 
P(c,\beta_1)\land P(h\beta_1,\beta_2),\Xi_2$
\LeftLabel{$\exr$}
\UnaryInf$\Lambda_4,P(c,\beta_1),P(h\beta_1,\beta_2) 
\;\fCenter\seq 
\exists w.P(c,\beta_1)\land P(h\beta_1,w),\Xi_1$
\LeftLabel{$\exr$}
\UnaryInf$\Lambda_4,P(c,\beta_1),P(h\beta_1,\beta_2) 
\;\fCenter\seq 
\exists v,w.P(c,v)\land P(hv,w)$
\LeftLabel{$\exl$}
\UnaryInf$\Lambda_3,P(c,\beta_1),\exists y.P(h\beta_1,y) 
\;\fCenter\seq 
\exists v,w.P(c,v)\land P(hv,w)$
\LeftLabel{$\alll$}
\UnaryInf$\Lambda_3,P(c,\beta_1) 
\;\fCenter\seq 
\exists v,w.P(c,v)\land P(hv,w)$
\LeftLabel{$\exl$}
\UnaryInf$\Lambda_2,\exists y.P(c,y) 
\;\fCenter\seq 
\exists v,w.P(c,v)\land P(hv,w)$
\LeftLabel{$\alll$}
\UnaryInf$\Lambda_1,\forall x\exists y.P(x,y) 
\;\fCenter\seq 
\exists v,w.P(c,v)\land P(hv,w)$
\DisplayProof
\end{fCenter}
\end{center}
and  
\begin{center}
\begin{minipage}{0.49\textwidth}
\begin{align*}
\Gamma := &\; \{\forall u.P(u,fu)\lor P(u,gu)\} \\
\Delta_1 := &\; \{\exists v,w.P(c,v)\land P(hv,w)\} \\
\Delta_2 := &\; \{\forall x\exists y.P(x,y)\}\cup\Delta_1 \\
\Delta_3 := &\; \{\exists y.P(\alpha ,y)\}\cup\Delta_2 \\
\Delta_4 := &\; \{P(\alpha ,g\alpha )\}\cup\Delta_3 \\
\Delta_5 := &\; \{P(\alpha ,f\alpha )\}\cup\Delta_3 
\end{align*}
\end{minipage}
\begin{minipage}{0.49\textwidth}
\begin{align*}
\Lambda_1 := &\; \Gamma \\
\Lambda_2 := &\; \{\forall x\exists y.P(x,y)\}\cup\Lambda_1 \\
\Lambda_3 := &\; \{\exists y.P(c,y)\}\cup\Lambda_2 \\
\Lambda_4 := &\; \{\exists y.P(h\beta_1,y)\}\cup\Lambda_3 \\
\Lambda_5 := &\; \{P(h\beta_1,\beta_2)\}\cup\Lambda_4 \\
\Lambda_6 := &\; \{P(c,\beta_1)\}\cup\Lambda_4 \\
\Xi_1 := &\; \Delta_1 \\
\Xi_2 := &\; \{\exists w.P(c,\beta_1)\land P(h\beta_1,w)\}\cup\Xi_1.
\end{align*}
\end{minipage}
\end{center}
Then the extended Herbrand sequent $\EH{P(x,y)}$ is given by 
\begin{align*}
&\; P(\alpha ,f\alpha )\lor P(\alpha ,g\alpha ), &\; \\
&\; (P(\alpha ,f\alpha )\lor P(\alpha ,g\alpha ))
\to (P(c,\beta_1)\land P(h\beta_1,\beta_2)) & \seq P(c,\beta_1)\land 
P(h\beta_1,\beta_2)
\end{align*}
where $U_1=\{\alpha\}$ and $U_2=\{ (\beta_1,\beta_2)\}$. Considering Definition \ref{def.extHseq}, the formula $P(\alpha ,f\alpha )\lor P(\alpha ,g\alpha )$ corresponds to $F[\vec{x}\backslash U_1]$, the formula $P(c,\beta_1)\land P(h\beta_1,\beta_2)$ corresponds to $G[\vec{y}\backslash U_2]$, and $P(x,y)$ corresponds to $A$. Therefore, $c$ and $h\beta_1$ are the analogue of $r_1$ and $r_2$, $m=2$, $p=2$, and $f$ and $g$ are the analogue of $t_1$ and $t_2$.

We will not give a cut-free proof because of its  size. Instead, we 
define a Herbrand sequent $H$, which provides the quantifier information of a corresponding cut-free proof. 
\begin{align*}
& P(c,fc)\lor P(c,gc), &\; \\ 
& P(hfc,fhfc)\lor P(hfc,ghfc), &\; \\ 
& P(hgc,fhgc)\lor P(hgc,ghgc) & \seq P(c,fc)\land P(hfc,fhfc), \\
& &\; P(c,fc)\land P(hfc,ghfc), \\
& &\; P(c,gc)\land P(hgc,fhgc), \\
& &\; P(c,gc)\land P(hgc,ghgc) . 
\end{align*}
Both sequents, $\EH{P(x,y)}$ and $H$, are tautological as one can easily verify. 
\end{example}
We obtain a result analogous to that in $\Pi_1$-cut-introduction~\cite{hetzlS2014aa}:
\begin{theorem}
\label{the:Pi2-cut extended Herbrand sequent}
The sequent $\forall\tuples{x}F\seq\exists\tuples{y}G$ has a proof $\pi$ with a single 
$\Pi_2$-cut $\forall x\exists y.A$ such that $|\pi |_q=n$ iff it has an extended 
Herbrand-sequent $\EH{A}$ with $|\EH{A}|=n$. 
\end{theorem}
\begin{proof}
For the left-to-right direction we pass through the proof $\pi$ and read off the instances of quantified formulas (of both the end-formula and the cut). We obtain an extended Herbrand sequent $\EH{A}$ with $|\EH{A}|\le |\pi |_q$ (which can be padded with dummy instances if necessary in order to obtain $|\EH{A}|=|\pi |_q$). \\ 
Let 
\begin{center}
\begin{minipage}{.39\textwidth}
\begin{align*}
\Gamma:=&\; F[\tuples{x}\backslash U_1], \\
\Delta:=&\; G[\tuples{y}\backslash U_2], \\
\Lambda:=&\; \forall\tuples{x}.F, \\
\Xi:=&\; \exists\tuples{y}.G,
\end{align*}
\end{minipage}
\begin{minipage}{.59\textwidth}
\begin{align*}
\Xi^l_1:=&\; \{\forall x\exists y.A\}\cup\Xi, \\
\Xi^l_2:=&\; \{\exists y.A[x\backslash\alpha ]\}\cup\Xi^l_1, \\
\Xi^l_{i+1}:=&\; \{A[x\backslash\alpha ,y\backslash t_{i-1}]\}\cup\Xi^l_i\text{ for }i\ge 2, \\
\Lambda^r_1:=&\; \{\forall x\exists y.A\}\cup\Lambda, \\
\Lambda^r_{2i}:=&\; \{\exists y.A[x\backslash r_i]\}\cup\Lambda^r_{2i-1}\text{ for }i\ge 1,\\
\Lambda^r_{2i+1}:=&\; \{A[x\backslash r_i,y\backslash\beta_i]\}\cup\Lambda^r_{2i}\text{ for }i\ge 1.
\end{align*}
\end{minipage}
\end{center}
For the right-to-left direction we conclude from 
\begin{center}
\begin{fCenter}
\Axiom$\Gamma \;\fCenter\seq A[x\backslash\alpha ,y\backslash t_i] 
_{i=1}^p,\Delta$
\LeftLabel{$\lor\colon$r}
\UnaryInf$ \fCenter \vdots $
\LeftLabel{$\lor\colon$r}
\UnaryInf$\Gamma \;\fCenter\seq \bigvee\limits_{i=1}^pA[x\backslash 
\alpha ,y\backslash t_i],\Delta$
\Axiom$\Gamma,A[x\backslash r_j,y\backslash \beta_j]_{j=1}^m 
\;\fCenter\seq \Delta$
\RightLabel{$\land\colon$l}
\UnaryInf$ \fCenter \vdots $
\RightLabel{$\land\colon$l}
\UnaryInf$\Gamma,\bigwedge\limits_{j=1}^mA[x\backslash r_j, 
y\backslash \beta_j] \;\fCenter\seq \Delta$
\LeftLabel{$\to\colon$l}
\BinaryInf$\Gamma,\bigvee\limits_{i=1}^pA[x\backslash\alpha 
,y\backslash t_i]\to 
\bigwedge\limits_{j=1}^mA[x\backslash r_j,y\backslash \beta_j] \;\fCenter\seq 
\Delta$
\DisplayProof
\end{fCenter}
\end{center}
that 
\begin{center}
\begin{fCenter}
\Axiom$\fCenter\vdots $
\UnaryInf$\Lambda^l,\Gamma \;\fCenter\seq A[x\backslash\alpha ,y\backslash t_i]_{i=1}^p, 
\Xi^l_{p+1}$
\LeftLabel{$\alll$}
\UnaryInf$ \fCenter \vdots $
\LeftLabel{$\alll$}
\UnaryInf$\Lambda \;\fCenter\seq A[x\backslash\alpha ,y\backslash t_i] 
_{i=1}^p,\Xi^l_{p+1}$
\LeftLabel{$\exr$}
\UnaryInf$ \fCenter \vdots $
\LeftLabel{$\exr$}
\UnaryInf$\Lambda \;\fCenter\seq A[x\backslash\alpha ,y\backslash t_1],\Xi^l_2$
\LeftLabel{$\exr$}
\UnaryInf$\Lambda \;\fCenter\seq \exists y.A[x\backslash\alpha ],\Xi^l_1$
\LeftLabel{$\allr$}
\UnaryInf$\Lambda \;\fCenter\seq \forall x\exists y.A,\Xi$
\Axiom$\fCenter\vdots $
\UnaryInf$\Lambda^r_{2m},A[x\backslash r_j,y\backslash \beta_j]_{j=1}^m 
\;\fCenter\seq \Delta,\Xi^r$
\RightLabel{$\exr$}
\UnaryInf$ \fCenter \vdots $ 
\RightLabel{$\exr$}
\UnaryInf$\Lambda^r_{2m},A[x\backslash r_j,y\backslash \beta_j]_{j=1}^m 
\;\fCenter\seq \Xi$
\RightLabel{$\exl$}
\UnaryInf$ \fCenter \vdots $
\RightLabel{$\alll$}
\UnaryInf$\Lambda^r_2,A[x\backslash r_1,y\backslash\beta_1] 
\;\fCenter\seq \Xi$
\RightLabel{$\exl$}
\UnaryInf$\Lambda^r_1,\exists y.A[x\backslash r_1] 
\;\fCenter\seq \Xi$
\RightLabel{$\alll$}
\UnaryInf$\Lambda,\forall x\exists y.A \;\fCenter\seq \Xi$
\LeftLabel{\textit{Cut}}
\BinaryInf$\Lambda \;\fCenter\seq \Xi$
\DisplayProof
\end{fCenter}
\end{center}
is a valid proof with a single $\Pi_2$-cut. The provability of 
\[\Gamma \seq A[x\backslash\alpha ,y\backslash t_i] _{i=1}^p,\Delta\] 
is given by the extended Herbrand sequent being a tautology and it implies the provability of 
\[\Lambda^l,\Gamma \seq A[x\backslash\alpha ,y\backslash t_i]_{i=1}^p, \Xi^l_{p+1}.\] 
Note that we defined $\Gamma$ and $\Delta$ such that 
\[\Gamma \seq A[x\backslash\alpha ,y\backslash t_i] _{i=1}^p\] 
is provable. The reasoning for the right branch is analogous. The notation $A[x\backslash \alpha ,y\backslash t_i]_{i=1}^p$ is an abbreviation for $A[x\backslash\alpha ,y \backslash t_1],$ $\ldots ,A[x\backslash\alpha ,y\backslash t_p]$. The labels $\lor\colon$r, $\to\colon$l, and $\land\colon$l denote the used rule, i.e.\ the {\em right-disjunction rule}, the {\em left-implication rule}, and the {\em left-conjunction rule}. In the latter variant the dots represent a multiple application of $\allr$ and $\exr$. In the particular case between the sequents $\Lambda^r_2,A[x\backslash r_1,y\backslash\beta_1] \seq \Xi$ and $\Lambda^r_{2m},A[x\backslash r_j,y\backslash \beta_j]_{j=1}^m \seq \Xi$ the dots denote an alternating application of $\alll$ and $\exl$ ($m-1$ times). 

Given that every term of $\EH{A}$ is used exactly once in a quantifier rule the quantifier complexity is equal to $|\EH{A}|$. 
\end{proof}

\section{Grammars}\label{sec.grammars}

The way variables are replaced in the procedure of cut-elimination can be defined by 
grammars modeling substitutions of terms. A characterization of the substitutions 
defining the Herbrand instances of a proof after cut-elimination of $\Pi_1$-cuts can 
be found in~\cite{hetzlS2012ab}. Below we give some necessary definitions.
\begin{definition}[Regular tree grammar]\label{def.reg-tree-grammar}
A {\em regular tree grammar} $\Gcal$ is a tuple $\langle\tau ,N,\Sigma ,\Pr\rangle$ where $N$ is a 
finite set of non-terminal symbols with arity $0$ such that $\tau\in N$. Furthermore, 
$\Sigma$ is a finite set of function symbols of arbitrary arities, i.e.\ a term 
signature, satisfying $N\cap\Sigma =\emptyset$. The productions $\Pr$ are a finite set 
of rules of the form $\gamma\to t$ where $\gamma\in N$ and $t\in T(\Sigma\cup N)$, where
$T(\Sigma\cup N)$ denotes the set of all terms definable from symbols in $\Sigma\cup 
N$. As usual $L(\Gcal)$, the {\em language} defined by $\Gcal$, is the set of all 
terminal strings (ground terms) derivable in $\Gcal$.
\end{definition}
The languages of grammars specifying Herbrand instances are finite (see~\cite{hetzlS2014aa}) and therefore their productions are w.l.o.g.\ acyclic.
\begin{definition}[Acyclic tree grammar]\label{def.acyclic-grammar}
We call a regular tree grammar {\em acyclic} if there is a total order $<$ on the 
non-terminals $N$ such that for each rule $\gamma\to t$ in $\Pr$, only non-terminals 
smaller than $\gamma$ occur in $t$. 
\end{definition}
We are interested in grammars specifying substitutions. As substitutions are 
homomorphic mappings on terms, variables have to be replaced only by single 
terms within a derivation. Therefore we need a restriction of derivations, \emph{rigid derivations}.
\begin{definition}[Rigid derivation]\label{def.der-rigid}
We call a derivation {\em rigid} with respect to a non-terminal $\gamma$ if only a 
single rule for $\gamma$ is allowed to occur in the derivation. 
\end{definition}
\begin{remark}
This allows us to consider production rules as substitutions.
\end{remark}
The following type of grammar describes the substitutions generated in the 
elimination of a $\Pi_2$ cut; this type grammar is a special case of more general 
grammars defined in~\cite{afshariB2015aa}. Assume a subproof of the form
\begin{center}
\begin{fCenter}
\Axiom$\Gamma \;\fCenter\seq \Delta'' ,A[x\backslash\alpha ,y\backslash t]$
\LeftLabel{\exr}
\UnaryInf$\;\fCenter\vdots $
\UnaryInf$\Gamma \;\fCenter\seq \Delta' ,\exists y.A[x\backslash \alpha ]$
\LeftLabel{\allr}
\UnaryInf$\Gamma \;\fCenter\seq \Delta ,\forall x\exists y.A$
\Axiom$A[x\backslash r,y\backslash\beta],\Gamma'' \;\fCenter\seq \Delta$
\RightLabel{\exl}
\UnaryInf$\exists y.A[x\backslash r],\Gamma' \;\fCenter\seq \Delta$
\RightLabel{\alll}
\UnaryInf$\;\fCenter\vdots $
\UnaryInf$\forall x\exists y.A,\Gamma \;\fCenter\seq \Delta$
\LeftLabel{\textit{Cut}}
\BinaryInf$\Gamma \;\fCenter\seq \Delta$
\DisplayProof
\end{fCenter}
\end{center}
Then we extract the production rules $\alpha\to r$ and $\beta\to t[\alpha\backslash r]$. For a better understanding we extract a \emph{schematic $\Pi_2$-grammar} corresponding to the substitutions in the elimination of the $\Pi_2$ cut of Example \ref{exa:extended Herbrand} and afterwards give the formal definition.
\begin{example}
\label{exa:schematic Pi2-grammar}
Let $\EH{P(x,y)}$ be as in Example \ref{exa:extended Herbrand}. Assume a unary term $h_F$ and a binary term $h_G$. While $h_F$ is a representation of the formula $P(u,fu)\lor P(u,gu)$ with the argument $u$, the term $h_G$ represents $P(c,v)\land P(hv,w)$ with the arguments $v,w$. As long as there is an unique identification between the term representation and the formula, the choice of the term symbols $h_F,h_G$ is not restricted. It follows that we start with the production rules $\tau\to h_F(\alpha )$ and $\tau\to h_G(\beta_1,\beta_2)$ where $\tau$ is the designated starting symbol to get a representation of the context 
\begin{align*}
P(\alpha ,f\alpha )\lor P(\alpha ,g\alpha )\seq P(c,\beta_1)\land P(h\beta_1,\beta_2).
\end{align*} 
The substitutions generated in the cut-elimination procedure correspond to the substitutions of the eigenvariables $\alpha ,\beta_1,$ and $\beta_2$. So we consider them as non-terminals in our grammar and examine the cut encoded by the implication 
\begin{align*}
(P(\alpha ,f\alpha )\lor P(\alpha ,g\alpha ))\to (P(c,\beta_1)\land P(h\beta_1,\beta_2)). 
\end{align*}
If we compare the substitutions of $x$ and $y$ of $P(x,y)$ in $P(\alpha ,f\alpha ), P(\alpha ,g\alpha )$ with the substitutions in $P(c,\beta_1),P(h\beta_1,\beta_2)$, we can derive the production rules $\alpha\to c|h\beta_1,\beta_2\to fh\beta_1|gh\beta_1,$ and $\beta_1\to fc|gc$. This is due to the cut-elimination procedure in which at first $\exists y.P(\alpha ,y)$ would be gradually replaced by $\exists y.P(c,y)$, $\exists y.P(hfc,y)$, and $\exists y.P(ghc,y)$, such that there is instead of a single subproof with $\exists y.P(\alpha ,y)$ tree structurally equivalent subproofs. One might assume that $\beta_i$ has to map to $f\alpha$ and $g\alpha$ for $i\in\mathbb{N}_2$ instead. Note that not only the terms substituted for $x$ belong together, but the pairs substituted for $x$ and $y$. For instance, if we substitute $(\alpha ,f\alpha )$ and $(c,\beta_1)$ for $x$ and $y$, $\beta_1$ might map to $f\alpha$, but also $\alpha$ maps to $c$ at the same time. \\
Altogether, we can define a grammar $\Gcal= \langle\tau ,N,\Sigma ,\Pr\rangle$ where $N=\{\tau ,\alpha ,\beta_1,\beta_2\}$, $\Sigma =\{ h_F,h_G,c,h,f,g\}$, and $\Pr=$
\begin{align*}
\{\;\;\tau \to &\; h_F(\alpha )|h_G(\beta_1,\beta_2), \\
\alpha \to &\; c|h\beta_1, \\
\beta_2 \to &\; fh\beta_1|gh\beta_1, \\
\beta_1 \to &\; fc|gc \;\}.
\end{align*}
If we allow only rigid derivations the language $L(\Gcal )$ of $\Gcal$ is 
\begin{align*}
\{ &\; h_F(c ),h_F(hfc ),h_F(hgc ), \\
&\; h_G(fc ,fhfc ),h_G(gc ,fhgc ), \\
&\; h_G(fc ,ghfc ),h_G(gc ,ghgc ) \;\} .
\end{align*}
By reconsidering the formulas that are represented by $h_F$ and $h_G$, we obtain that the language represents the Herbrand sequent of Example \ref{exa:extended Herbrand}. In general, we consider grammars whose language corresponds to a superset of a Herbrand sequent.  
\end{example}
\begin{definition}[Schematic $\Pi_2$-grammar]\label{def.schem-pi2-grammar}
Let $\Gcal = \langle\tau ,N,\Sigma ,\Pr\rangle$ be an acyclic tree grammar, $N = 
\{\tau ,\alpha ,\beta_1,\ldots ,\beta_m\}$, and $s_i,w_j,r_k,t_l$ for $i\in\mathbb{N}_c,j\in\mathbb{N}_d,k\in\mathbb{N}_m,l\in\mathbb{N}_p$ be terms. 
Let the non-terminals be ordered according to $\beta_1<\ldots <\beta_m<\alpha <\tau$. We call $\Gcal$ a {\em schematic $\Pi_2$-grammar} if the production rules are of the following form: 
\begin{align*}
\tau \to &\; s_1|\ldots |s_c | w_1|\ldots |w_d\  \text{ with }V(s_i)\subseteq\{\alpha\} 
\text{ for }1\le i\le c
\text{ and }\\
&\; V(w_j)\subseteq\{\beta_1,\ldots ,\beta_m\}\text{ for } 1\le j\le d \\ 
\alpha \to &\;  r_1|\ldots |r_m \text{ with }V(r_j)\subseteq\{\beta_1,\ldots 
,\beta_{j-1}\} \text{ for }2\le j\le m\text{ and }V(r_1)=\emptyset \\
\beta_j \to &\;  t_1r_j|\ldots |t_pr_j \text{ for }1\le j\le m
\end{align*}
where $ts$ stands for $t[\alpha\backslash s]$ with $t$ being a term (possibly) containing the variable $\alpha$.
%the $\beta$-normal form of $(\lambda \alpha .t)s$. 
We call $m$ the \AllMul and $p$ the \ExMul\!\!.
\end{definition}
Let $\EH{A}$ be an extended Herbrand sequent as defined in 
Definition~\ref{def.extHseq}; 
every such $\EH{A}$ defines a schematic $\Pi_2$-grammar. As the sequent $S\colon 
\forall\tuples{x}F\seq\exists\tuples{y}G$ contains blocks of quantifiers and we want to 
use an ordinary term grammar, we generate function symbols $h_F,h_G$ 
where $h_F$ is of the arity of the length of $\tuples{x}$, $h_G$ of the arity of the 
length of $\tuples{y}$. So every term tuple $\tuples{u}_i \in U_1$ is represented by 
$h_F(\tuples{u}_i)$, and every term tuple $\tuples{v}_j \in U_2$ by $h_G(\tuples{v}_j)$.  
\begin{definition}\label{def.extHseq-grammar}
Let $\EH{A}$ as in Definition~\ref{def.extHseq}. We define 
$\Gcal(\EH{A}) = \langle\tau 
,N,\Sigma ,\Pr\rangle$, the {\em schematic $\Pi_2$-grammar corresponding to 
$\EH{A}$}, 
where $N = \{\tau ,\alpha ,\beta_1,\ldots ,\beta_m\}$ and the variables are ordered 
as in Definition~\ref{def.schem-pi2-grammar}; the production rules are as in 
Definition~\ref{def.schem-pi2-grammar}, except for the variable $\tau$ where we have 
\begin{align*}
\tau \to &\; h_F(\tuples{u}_1)|\ldots |h_F(\tuples{u}_N)\ |  \  h_G(\tuples{v}_1)
|\ldots |h_G(\tuples{v}_M).
\end{align*}
We call the production rules $\tau\to h_F(\tuples{u}_1)|\ldots |h_F(\tuples{u}_N)$ {\em $F$-productions} and the production rules $\tau\to h_G(\tuples{v}_1)|\ldots |
h_G(\tuples{v}_M)$ {\em $G$-productions}. 
\end{definition}
At this point it becomes apparent why we have chosen the form 
$\forall\tuples{x}.F\seq\exists\tuples{y}.G$ as end sequent. In a schematic $\Pi_2$-grammar 
we have terms depending on $\alpha$ and terms depending on some $\beta_i$ with $i\in 
\mathbb{N}_m$. These terms correspond to the function symbols $h_F$ and $h_G$, i.e.\ 
we implicitly ask for formulas that can be separated within one sequent (by a comma on 
the right side, a comma on the left side, or the sequent symbol $\seq$). This separated 
formulas depend either on $\alpha$ or on some $\beta_i$ with $i\in\mathbb{N}_m$. 
Hence, there are no atoms that depend on both, $\alpha$ and $\beta_i$ for $i\in 
\mathbb{N}_m$.

\section{Cut-Introduction}\label{sec.cutintro}

We have shown that from any proof with a $\Pi_2$-cut we can extract a schematic $\Pi_2$-grammar.
The language of this grammar covers the so-called 
\emph{Herbrand term set}, a representation of the instantiations defining a Herbrand 
sequent. Now the question arises, whether we can invert this step, i.e. to construct 
a proof with a $\Pi_2$-cut from a cut-free proof $\varphi$ and a given schematic 
$\Pi_2$-grammar specifying the Herbrand instances of $\varphi$.  

\begin{definition}[Herbrand term set]\label{def.Htset}
Let $S\colon \all\tuples{x}F \seq \ex \tuples{y}G$ be  a sequent and $H$ be a Herbrand 
sequent of $S$ of the form  
\begin{align*}
H:=F[\tuples{x}\backslash\tuples{t}_1],\ldots ,F[\tuples{x}\backslash\tuples{t}_k]\seq 
G[\tuples{y}\backslash\tuples{t}_{k+1}],\ldots ,G[\tuples{y}\backslash\tuples{t}_n]
\end{align*}
as in Definition~\ref{def.Herbrandsequent}, and $h_F,h_G$ function symbols as defined 
above. Then the set 
\begin{align*}
H_s(S)\colon \{h_F(\tuples{t}_1),\ldots,h_F(\tuples{t}_k),h_G(\tuples{t}_{k+1}),\ldots, 
h_G(\tuples{t}_n)\}
\end{align*}
is called a {\em Herbrand term set} of $S$.
\end{definition}
While Herbrand sequents represent cut-free proofs, extended Herbrand sequents represent 
proofs with cuts. To introduce (yet unknown) cut-formulas 
we consider the Herbrand sequent of a cut-free proof and specify the Herbrand term 
set by a schematic $\Pi_2$-grammar. The unknown cut formula is represented by a 
second-order variable $X$.
 
\begin{definition}[Schematic extended Herbrand sequent] 
\label{def:Schematic extended Herbrand sequent}
Let $S\colon \forall\tuples{x}F\seq\exists\tuples{y}G$ be a provable sequent and $H_s(S)$ 
be a Herbrand term set of $S$.
Let  $\Gcal\colon \langle\tau ,N,\Sigma ,\Pr\rangle$ be a schematic $\Pi_2$-grammar 
with $N=\{\tau ,\alpha ,\beta_1, 
\ldots ,\beta_m\},$ $\beta_1<\ldots <\beta_m<\alpha <\tau$, and the production rules 
\begin{align*}
\tau \to &\; h_F(\tuples{u}_1)|\ldots |h_F(\tuples{u}_N)\ |  \  h_G(\tuples{v}_1)|\ldots |h_G(\tuples{v}_M) \\ 
&\; \text{ with }V(\tuples{u}_i)\subseteq\{\alpha\}\text{ for }1\le i\le N \\ 
&\; \text{ and }V(\tuples{v}_j)\subseteq\{\beta_1,\ldots ,\beta_m\}\text{ for } 
1\le j\le M \\ 
\alpha \to &\; r_1|\ldots |r_m \text{ with }V(r_j)\subseteq\{\beta_1,\ldots 
,\beta_{j-1}\} \text{ for }2\le j\le m\text{ and }V(r_1)=\emptyset \\
\beta_j \to &\;  t_1r_j|\ldots |t_pr_j \text{ for }1\le j\le m.
\end{align*}
Let $L(\Gcal)$ be the language of $\Gcal$ generated only by rigid derivations 
with respect to all non-terminals, and 
$H_s(S) \subseteq L(\Gcal)$. Let $U_1:=\{\tuples{u}_1,\ldots ,\tuples{u}_N\}$ and 
$U_2:=\{\tuples{v}_1,\ldots ,\tuples{v}_M\}$. We call the sequent 
$$S(X)\colon F[\tuples{x}\backslash U_1],\bigvee\limits_{i=1}^pX\alpha t_i\to 
\bigwedge\limits_{j=1}^mXr_j\beta_j\seq 
G[\tuples{y}\backslash U_2],$$
where $X$ is a two-place predicate variable, a {\em schematic extended Herbrand 
sequent} corresponding to $\Gcal$ and $S$ (in the following abbreviated by {\em SEHS}). 

Furthermore, we call $F[\tuples{x}\backslash U_1]\seq G[\tuples{y}\backslash U_2]$ the 
{\em reduced representation} of $S(X)$. 
\end{definition}
The schematic extended Herbrand sequent is an abstraction of an extended Herbrand sequent as well as the extended Herbrand sequent is a schematic extended Herbrand sequent $S(X)$ for which a \emph{solution} (see Definition \ref{def:solution of an sehs}) has been found, i.e.\ a formula $C$ such that $S(C)$ is a tautology. Further explanation can be found in Example \ref{exa:schematic extended Herbrand sequent}.

Note that we did not require $L(\Gcal) = H_s(S)$; indeed if we generate a proper 
superset of $H_s(S)$ we still obtain a Herbrand sequent of $S$ (but not a minimal one). 
Generating supersets can be beneficial to the construction of cut-formulas.

A solution of an \seHs gives us a cut formula for a proof with a $\Pi_2$-cut. 
\begin{definition}
\label{def:solution of an sehs}
Let $S$ be a provable sequent, $\mathcal{G}$ a schematic $\Pi_2$-grammar with the non-terminals $\{\tau ,\alpha ,\beta_1,\ldots ,\beta_m\}$ as in Definition \ref{def.schem-pi2-grammar}, and $S(X)$ the corresponding \seHs\!\!. Let $S(X)[X \setminus \lambda x y.A]$ be a tautology where $A$ may not contain $\alpha$ and $\beta_j$ with $j\in\mathbb{N}_m$. Then we call $A$ a \emph{solution} of the \seHs\!\! $S(X)$.
\end{definition}
\begin{example}
\label{exa:schematic extended Herbrand sequent}
In Example \ref{exa:extended Herbrand}, we extracted an extended Herbrand sequent $\EH{P(x,y,)}$ from a proof with a single $\Pi_2$ cut: 
\begin{align*}
&\; P(\alpha ,f\alpha )\lor P(\alpha ,g\alpha ), &\; \\
&\; (P(\alpha ,f\alpha )\lor P(\alpha ,g\alpha ))
\to (P(c,\beta_1)\land P(h\beta_1,\beta_2)) & \seq P(c,\beta_1)\land 
P(h\beta_1,\beta_2).
\end{align*}
Later, we generated the schematic $\Pi_2$-grammar $\Gcal= \langle\tau ,N,\Sigma ,\Pr\rangle$ for $\EH{P(x,y)}$ (see Example \ref{exa:schematic Pi2-grammar}). $\Gcal$ is defined by $N=\{\tau ,\alpha ,\beta_1,\beta_2\}$, $\Sigma =\{ h_F,h_G,c,h,f,g\}$, and $\Pr=$
\begin{align*}
\{\;\;\tau \to &\; h_F(\alpha )|h_G(\beta_1,\beta_2), \\
\alpha \to &\; c|h\beta_1, \\
\beta_2 \to &\; fh\beta_1|gh\beta_1, \\
\beta_1 \to &\; fc|gc \;\}.
\end{align*}
The order of the non-terminals was not discussed, is already intrinsic, and reads $\beta_1<\beta_2<\alpha <\tau$. Then the corresponding \seHs\!\! can intuitively be described as $\EH{X}$ plus the grammar $\Gcal$ where $X$ is a two-place predicate variable. More precisely, $S(X)=$
\begin{align*}
&\; P(\alpha ,f\alpha )\lor P(\alpha ,g\alpha ), \\
&\; (X\alpha f\alpha \lor X\alpha g\alpha )
\to (Xc\beta_1\land Xh\beta_1\beta_2) \seq P(c,\beta_1)\land 
P(h\beta_1,\beta_2)
\end{align*}
is an \seHs\!\! corresponding to $\Gcal$ and $\forall u.P(u,fu)\lor P(u,gu)\seq\exists v,w.P(c,v)\land P(hv,w)$. The solution of $S(X)$ is $\lambda xy.P(x,y)$.
\end{example}
In the following we will think of proofs as trees. This will facilitate the description 
of our approaches to find a solution for the \seHs\!\!. Hence, the leaves of a proof 
represent tautological or non-tautological axioms. 
\begin{definition}
Let $S$ be a sequent. We call an arbitrary tree a \textbf{G3c}-derivation 
if it has only sequents as nodes, has $S$ as lowest element such that each edge corresponds to a rule of \textbf{G3c}, i.e.\ each node is a (tautological or non-tautological) axiom or a conclusion of a rule of \textbf{G3c}, and the immediate successors are the premise of that rule. 
\end{definition}
\begin{definition}
Let $S$ be a quantifier-free sequent. We call a \textbf{G3c}-derivation of $S$ \emph{maximal} if the leaves of the tree cannot be conclusions of rules. 
\end{definition}
Before we discuss a characterization of the solvability of the \seHs we show that in general it is not solvable. 
\begin{lemma}
\label{lem:counterExample}
Let $F:=P(x,fx)\land Q(x,gx)$ and $G:=P(c,x)\land Q(d,y)$. Assume the sequent $\forall x.F\seq \exists x,y.G$, the \seHs $S(X)$ 
\begin{align*}
&\; P(\alpha , f\alpha )\land Q(\alpha , g\alpha ), \\ 
&\; (X\alpha f\alpha \lor X\alpha g\alpha )\to (Xc\beta_1\land 
Xd\beta_2) \seq P(c,\beta_1) \land Q(d,\beta_2) , 
\end{align*}
and the schematic $\Pi_2$-grammar $\Gcal=\langle \tau ,N,\Sigma ,\Pr\rangle$ 
where $N=\{\tau ,\alpha ,\beta_1,\beta_2\}$ and 
\begin{align*}
\Pr = \{ &\tau\to h_{F}\alpha ,\tau\to 
h_{G}\beta_1\beta_2, \\ 
& \alpha\to c\; |\; d,\beta_2\to fd\; |\; gd, 
\beta_1\to fc\; |\; gc\} .
\end{align*}
Then the \seHs does not have a solution, i.e.\ there is no $C$ such that $S(C)$ 
is a tautology. 
\end{lemma}
To prove the given lemma we have to make a detour. First we prove a simpler case. 
\begin{lemma}
\label{lem:counterExamplesimple}
Let $F:=P(x,fx)\land Q(x,gx)$ and $G:=P(c,x)\land Q(c,y)$. Assume the sequent $\forall x.F\seq \exists x,y.G$, the \seHs $S(X)$ 
\begin{align*}
&\; P(\alpha , f\alpha )\land Q(\alpha , g\alpha ), \\ 
&\; (X\alpha f\alpha \lor X\alpha g\alpha )\to (Xc\beta_1\land 
Xc\beta_2) \seq P(c,\beta_1) \land Q(c,\beta_2) 
\end{align*}
and the schematic $\Pi_2$-grammar $\Gcal=\langle \tau ,N,\Sigma ,\Pr\rangle$ where $N=\{\tau ,\alpha ,\beta_1,\beta_2\}$ and 
\begin{align*}
\Pr = \{ &\tau\to h_{F}\alpha ,\tau\to 
h_{G}\beta_1\beta_2, \\ 
& \alpha\to c\; |\; c,\beta_2\to fc\; |\; gc, 
\beta_1\to fc\; |\; gc\} .
\end{align*}
Then the \seHs does not have a solution. 
\end{lemma}
\begin{proof}
We prove the lemma by contradiction. Let us assume a valid cut-formula $E$ that corresponds to the grammar $\mathcal{G}$. 
A maximal \textbf{G3c}-derivation $\psi$ of the reduced representation produces the 
leaves 
\begin{align*}
\{ & P(\alpha ,f\alpha ),Q(\alpha ,g\alpha )\seq P(c,\beta_1); \\ 
& P(\alpha ,f\alpha ),Q(\alpha ,g\alpha )\seq Q(c,\beta_2)\} .
\end{align*}
Given that $E$ is a valid cut-formula the following sequents have to be tautologies 
\begin{align*}
\mathcal{B}_1:=\;\;\{ & P(\alpha ,f\alpha ),Q(\alpha ,g\alpha )\seq P(c,\beta_1)
,E(\alpha ,f\alpha ),E(\alpha ,g\alpha ); \\ 
& P(\alpha ,f\alpha ),Q(\alpha ,g\alpha )\seq Q(c,\beta_2)
,E(\alpha ,f\alpha ),E(\alpha ,g\alpha ); \\ 
& E(c,\beta_1),E(c,\beta_2),P(\alpha ,f\alpha ),Q(\alpha ,g\alpha )\seq 
P(c,\beta_1); \\ 
& E(c,\beta_1),E(c,\beta_2),P(\alpha ,f\alpha ),Q(\alpha ,g\alpha )\seq 
Q(c,\beta_2)\} 
\end{align*}
and, hence, also the following sequents 
\begin{align*}
\mathcal{B}_2:=\;\; \{ & P(\alpha ,f\alpha ),Q(\alpha ,g\alpha )\seq 
E(\alpha ,f\alpha ),E(\alpha ,g\alpha ); \\ 
& E(c,\beta_1)\seq P(c,\beta_1); \\ 
& E(c,\beta_2)\seq Q(c,\beta_2)\} .
\end{align*}
That we can drop $P(c,\beta_1)$ and $Q(c,\beta_2)$ in the first two lines of $\mathcal{B}_1$ and $P(\alpha ,f\alpha ),$ $Q(\alpha ,g\alpha )$ in the last two lines of $\mathcal{B}_1$ to obtain the sequents in $\mathcal{B}_2$ is obvious (Neither $E(\alpha ,f\alpha ),E(\alpha ,g\alpha )$ can contain an atom depending on $\beta_1$ or $\beta_2$ nor $E(c,\beta_1)$ and $E(c,\beta_2)$ can contain an atom depending on $\alpha$). To prove that we can also ignore $E(c,\beta_2)$ in the third line we assume that $T:=E(c,\beta_1)\seq P(c,\beta_1)$ is not provable. Hence, there is a non-tautological branch $\Lambda_1\seq \Theta_1,P(c,\beta_1)$ in every maximal \textbf{G3c}-derivation $\psi$ of $T$. Given that $E(c,\beta_2)$ has the same logical structure as $E(c,\beta_1)$ we can apply the same \textbf{G3c}-rules of $\psi$ to $E(c,\beta_2)$ and get the sequent $\Lambda_2\seq \Theta_2$. The atoms of the sets $\Lambda_1$ and $\Theta_1$ are the same as the atoms in $\Lambda_2$ and $\Theta_2$ except for those which depend on the second argument of $E$, i.e.\ they contain $\beta_1$ or $\beta_2$. Thus, the sequent $\Lambda_1,\Lambda_2\seq \Theta_1, \Theta_2$ is not a tautology and also the atom $P(c,\beta_1)$ is not an element of $\Lambda_1\cup \Lambda_2$. Then also $S:=\Lambda_1,\Lambda_2\seq P(c,\beta_1), \Theta_1,\Theta_2$ is not a tautology. But $S$ is a leaf of every proof tree of $E(c,\beta_1),E(c,\beta_2)\seq P(c,\beta_1)$. This is a contradiction and, therefore, $T$ has to be a tautology. Analogously we can prove that $E(c,\beta_2)\seq Q(c,\beta_2)$ has to be a tautology if $E(c,\beta_1), E(c,\beta_2)\seq P(c,\beta_1)$ is a tautology. 

If the sequents in $\mathcal{B}_2$ are provable then we can replace in their proofs $\alpha$ with $c$, $\beta_1$ with $gc$, and $\beta_2$ with $fc$ to get the provable sequents 
\begin{align*}
\{ & P(c ,fc ),Q(f ,gc )\seq E(c ,fc ),E(c ,gc ); \\ 
& E(c,gc)\seq P(c,gc); \\ 
& E(c,fc)\seq Q(c,fc)\} .
\end{align*}
Now we can apply two times the cut-rule 
\begin{center}
\begin{fCenter}
\Axiom$ \;\fCenter \pi_l $
\Axiom$ E(c,gc) \;\fCenter\seq P(c,gc) $ 
\LeftLabel{\textit{Cut}}
\BinaryInf$ P(c,fc),Q(c,gc) \;\fCenter\seq Q(c,fc),P(c,gc) $ 
\DisplayProof
\end{fCenter}
\end{center}
with $\pi_l:=$ 
\begin{center}
\begin{fCenter}
\Axiom$ P(c,fc),Q(c,gc) \;\fCenter\seq E(c,fc),E(c,gc) $ 
\Axiom$ E(c,fc) \;\fCenter\seq Q(c,fc) $
\LeftLabel{\textit{Cut}}
\BinaryInf$ P(c,fc),Q(c,gc) \;\fCenter\seq Q(c,fc),E(c,gc) $ 
\DisplayProof
\end{fCenter}
\end{center}
and derive the sequent $P(c,fc),Q(c,gc)\seq Q(c,fc),P(c,gc)$. But this sequent is not valid and, by contradiction, there is no cut-formula. 
\end{proof}
In general this example suffices to show that there is not always a solution for an 
\seHs\!\!. But at this point one can argue that we have to refine the definition of 
schematic $\Pi_2$-grammars. If production rules have to be unique then the given example 
would be inappropriate (the production rules with $\beta_1$ and $\beta_2$ on the left map on the same terms and are, therefore, not unique). The \seHs of Lemma \ref{lem:counterExample} contains only unique production rules and now we are able to prove this lemma.
\begin{proof}[Proof of Lemma \ref{lem:counterExample}]
To prove the lemma we give a model in which $c$ and $d$ are equal because for this case Lemma \ref{lem:counterExamplesimple} shows the non-existence of a cut-formula. Assume the natural numbers modulo $2$. We interpret $c$ as $0$, $d$ as $2$, $\lambda x.f$ as the successor function $\lambda x.sx$, and $\lambda x.g$ as $\lambda x.ssx$. In this model $c$ is equal to $d$ and, hence, there cannot be a cut-formula. 
\end{proof}
\begin{remark}
If we take a sequent calculus with equality and add the formula $\neg c=d$ to the left of the end-sequent, i.e.\ an additional assumption, then $\forall x\exists y.(x=c\to P(x,y))\land (\neg x=c\to Q(x,y))$ is a valid cut-formula that corresponds to the given schematic $\Pi_2$-grammar. Let 
\begin{align*}
\forall x\exists y.C(x,y) := &\; \forall x\exists y.(x=c\to P(x,y))\land (\neg x=c\to Q(x,y)) 
\end{align*}
\vspace{-11mm}
\begin{center}
\begin{minipage}{.49\textwidth}
\begin{align*}
\Gamma_1 := &\; \{\forall x.P(x,fx)\land Q(x,gx)\} \\ 
\Gamma_2 := &\; \{\exists y.C(c,y)\}\cup\Gamma_1 \\ 
\Gamma_3 := &\; \{\forall x\exists y.C(x,y)\}\cup\Gamma_2 \\ 
\Gamma_4 := &\; \{\exists y.C(d,y)\}\cup\Gamma_3  
\end{align*}
\end{minipage}
\begin{minipage}{.49\textwidth}
\begin{align*}
\Delta_1 := &\; \{\exists x,y.P(c,x)\land Q(d,y)\} \\
\Delta_2 := &\; \{\forall x\exists y.C(x,y)\}\cup\Delta_1 \\
\Delta_3 := &\; \{\exists y.C(\alpha ,y)\}\cup\Delta_2 \\
\Delta_4 := &\; \{\exists y.P(c,\beta_1)\land Q(d,y)\}\cup\Delta_1 
\end{align*}
\end{minipage}
\end{center}
Then 
\begin{center}
\begin{fCenter}
\Axiom$ \;\fCenter \pi_1$
\Axiom$ \;\fCenter \pi_2$
\BinaryInf$\neg c=d,\forall x.P(x,fx)\land Q(x,gx) \;\fCenter\seq \exists x,y.P(c,x)\land Q(d,y)$
\DisplayProof
\end{fCenter}
\end{center}
where $\pi_1=$ 
\begin{center}
\begin{fCenter}
\Axiom$\;\fCenter\vdots$
\UnaryInf$\neg c=d,P(\alpha ,f\alpha)\land Q(\alpha ,g\alpha ),\Gamma_1  \;\fCenter\seq \Delta_3,C(\alpha ,g\alpha ),C(\alpha ,f\alpha )$
\RightLabel{\alll}
\UnaryInf$\neg c=d,\forall x.P(x,fx)\land Q(x,gx)  \;\fCenter\seq \Delta_3,C(\alpha ,g\alpha ),C(\alpha ,f\alpha )$
\RightLabel{\exr}
\UnaryInf$\neg c=d,\forall x.P(x,fx)\land Q(x,gx)  \;\fCenter\seq \Delta_2,\exists y.C(\alpha ,y),C(\alpha ,f\alpha )$
\RightLabel{\exr}
\UnaryInf$\neg c=d,\forall x.P(x,fx)\land Q(x,gx)  \;\fCenter\seq \Delta_2,\exists y.C(\alpha ,y)$
\RightLabel{\allr}
\UnaryInf$\neg c=d,\forall x.P(x,fx)\land Q(x,gx) \;\fCenter\seq \Delta_1,\forall x\exists y.C(x,y)$
\DisplayProof
\end{fCenter}
\end{center}
with the axiomatic leaves 
\begin{align*}
&\alpha =c,P(\alpha ,f\alpha ), Q(\alpha ,g\alpha ) \seq \Delta_3,c=d,P(\alpha ,f\alpha ),C(\alpha ,g\alpha ) , \\
&\alpha =c,P(\alpha ,f\alpha ), Q(\alpha ,g\alpha ) \seq \Delta_3,c=d, \alpha =c, Q(\alpha ,f\alpha ),P(\alpha ,g\alpha ) ,\text{ and} \\
&P(\alpha ,f\alpha ), Q(\alpha ,g\alpha ) \seq \Delta_3,c=d, \alpha =c, Q(\alpha ,f\alpha ), \alpha =c, Q(\alpha ,g\alpha )
\end{align*}
and $\pi_2=$
\begin{center}
\begin{fCenter}
\Axiom$\;\fCenter\vdots$
\UnaryInf$\neg c=d,\Gamma_4,C(c,\beta_1),C(d,\beta_2) \;\fCenter\seq \Delta_4,P(c,\beta_1)\land Q(d,\beta_2)$
\RightLabel{\exr}
\UnaryInf$\neg c=d,\Gamma_4,C(c,\beta_1),C(d,\beta_2) \;\fCenter\seq \Delta_1,\exists y.P(c,\beta_1)\land Q(d,y)$
\RightLabel{\exr}
\UnaryInf$\neg c=d,\Gamma_4,C(c,\beta_1),C(d,\beta_2) \;\fCenter\seq \exists x,y.P(c,x)\land Q(d,y)$
\RightLabel{\exl}
\UnaryInf$\neg c=d,\Gamma_3,C(c,\beta_1),\exists y.C(d,y) \;\fCenter\seq \exists x,y.P(c,x)\land Q(d,y)$
\RightLabel{\alll}
\UnaryInf$\neg c=d,\Gamma_2,C(c,\beta_1),\forall x\exists y.C(x,y) \;\fCenter\seq \exists x,y.P(c,x)\land Q(d,y)$
\RightLabel{\exl}
\UnaryInf$\neg c=d,\Gamma_1,\exists y.C(c,y),\forall x\exists y.C(x,y) \;\fCenter\seq \exists x,y.P(c,x)\land Q(d,y)$
\RightLabel{\alll}
\UnaryInf$\neg c=d,\Gamma_1,\forall x\exists y.C(x,y) \;\fCenter\seq \exists x,y.P(c,x)\land Q(d,y)$
\DisplayProof
\end{fCenter}
\end{center}
with the axiomatic leaves 
\begin{align}
&\Gamma_4,(\neg c=c\to Q(c,\beta_1)),C(d,\beta_2) \seq \Delta_4,c=c,c=d, P(c,\beta_1),\label{equ:cutproof with equality axiom 1} \\
&\Gamma_4,P(c,\beta_1),(\neg c=c\to Q(c,\beta_1)),C(d,\beta_2) \seq \Delta_4,c=d, P(c,\beta_1), \\
&\Gamma_4,C(c,\beta_1),d=c,(d=c\to P(d,\beta_2)) \seq \Delta_4,c=d, Q(d,\beta_2),\text{ and}\label{equ:cutproof with equality axiom 3} \\
&\Gamma_4,C(c,\beta_1),Q(d,\beta_2),(d=c\to P(d,\beta_2)) \seq \Delta_4,c=d, Q(d,\beta_2) 
\end{align}
is a valid proof with $\Pi_2$ cut corresponding to the schematic $\Pi_2$-grammar of Lemma \ref{lem:counterExample}. Note that the sequents \eqref{equ:cutproof with equality axiom 1} and \eqref{equ:cutproof with equality axiom 3} are only provable in a sequent calculus with equality.
\end{remark}

Both examples show that, in general, we cannot expect to find a solution for an 
\seHs\!\!. Moreover, it is difficult to give an easy restriction to the grammar 
such that the solvability is guaranteed. 
\\
We start now to characterize some conditions for the introduction of $\Pi_2$ cuts. We begin with a, so called, {\em starting set}. It may contain a set of clauses that is interpreted as a formula in DNF a solution for the \seHs\!\!, i.e.\ the \seHs where $X$ is replaced with this formula is a tautology. Later, we will define starting sets that always contain a solution as a subset for certain classes of solutions.
\begin{definition}[Starting set]\label{def:starting set}
Let $\mathcal{O}$ be a set of variables. We call a finite set of finite sets of literals $\Ccal^{\mathcal{O}}$ s.t. $V(\Ccal^{\mathcal{O}})\subseteq \{x,y\}\cup\mathcal{O}$ for designated variables $x,y$ a {\em starting set}. The variables $\beta_1,\ldots ,\beta_m$, and $\alpha$ may not occur in $\mathcal{O}$. If $\mathcal{O}=\emptyset$, we abbreviate $\Ccal^{\emptyset}$ by $\Ccal$.
\end{definition}
In general, we assume that the set of variables in the reduced representation contains only the eigenvariables $\alpha,\beta_1,\ldots ,\beta_m$. This is not a restriction because all other variables can be treated as constants. Hence, we can treat the variables in $\mathcal{O}$ as constants such that $\mathcal{O}$ can be considered empty. Thus, we will always consider $\mathcal{O}$ to be empty. \\
Now we define a normal form for the representation of the leaves of a reduced representation. Therefore we need the following proposition. 
\begin{proposition}\label{pro:non-tautological axioms}
Let $R$ be a reduced representation of an \seHs as in Definition~\ref{def:Schematic extended Herbrand sequent} and $\psi$ be a maximal \textbf{G3c}-derivation of $R$. Let NTA$(\psi )$ be the set of non-tautological axioms of $\psi$. Let $S\in\text{NTA}(\psi )$. Then $S$ is of the form $ A(S)\circ B(S)\circ N(S)$ where $A(S)$ is the sequent of all atoms in $S$ containing $\alpha$, $B(S)$ the sequent of all atoms in $S$ containing a non-empty subset of the variables $\{\beta_1,\ldots ,\beta_m\}$, and $N(S)$ ($N$ stands for ``neutral'') the sequent of all atoms in $S$ neither containing $\alpha$ nor $\beta_i$-s. 
\end{proposition}
\begin{proof}
Assume an arbitrary atom $P$ of $S$. We know that $P$ is a subformula of the reduced representation $R$. The reduced representation $R=F[\tuples{x}\backslash U_1]\seq G[\tuples{y}\backslash U_2]$ can be divided into two parts: $F[\tuples{x}\backslash U_1]\seq$ and $\seq G[\tuples{y}\backslash U_2]$. In the first part, neither of the variables $\beta_1,\ldots ,\beta_m$ appear; in the second part, the variable $\alpha$ does not appear. $P$ is either a subformula occurring in the first or second part, i.e.\ it cannot contain both, variables of the set $\{\beta_1,\ldots ,\beta_m\}$ and the variable $\alpha$. 
\end{proof}
The proposition gives us a representation of the leaves, but in this form 
we are not able to distinguish between atoms occurring on the left hand-side of a 
sequent and atoms occurring on the right hand-side of the sequent. 
\begin{definition}
\label{def:dual sequent operator}
Let $S:=P_1,\ldots ,P_i\seq Q_1,\ldots ,Q_j$ be a sequent containing only atoms. Then 
we define the literal normal form $D(S)$ of the sequent $S$ as $\neg Q_1, 
\ldots ,\neg Q_j,P_1,\ldots ,P_i\seq$. 
\end{definition}
Now each literal carries the information on which side of the sequent it occurs. 
If it is an atom it occurs on the left hand-side. If it is a negated atom it occurs 
on the right hand-side. Hence, we can define a normal form of the sequents. 
\begin{definition}
\label{def:DNTA}
Let NTA$(\psi )$ be the set of non-tautological axioms of a maximal \textbf{G3c}-derivation $\psi$ of a reduced representation $R$. We define the set of non-tautological axioms in literal normal form 
\begin{align*}
& \text{DNTA}(\psi ):=\{ D(S)\; |\; S\in\text{NTA}(\psi )\} .
\end{align*}
Let $S\in\text{DNTA}(\psi )$. Then $S$ is also of the form $ A(S)\circ B(S)\circ N(S) $ where 
\begin{itemize} 
\item $A(S)$ is the sequent of all literals in $S$ containing $\alpha$, 
\item $B(S)$ the sequent of all literals in $S$ containing a non-empty subset of the variables $\{\beta_1,\ldots ,\beta_m\}$, and 
\item $N(S)$ the sequent of all literals in $S$ neither containing $\alpha$ nor $\beta_i$-s. 
\end{itemize}

Let ${\rm LIT}$ be the set of all literals. For all literals $L\in A(S)$ let 
\[ \xi (L) :=  \{ Q\; |\; Q \in {\rm LIT} \textbf{ and }V(Q)\subseteq\{ x,y\}\textbf{ and }\exists i\in\mathbb{N}_p.Q[x\backslash\alpha ,y\backslash t_i]=L\}.\] 
then
\begin{align*}
A'(S):= &\; \bigcup\limits_{L\in A(S)}\xi (L)
\end{align*}
denotes the set of all literals that can be mapped  to an element of $A(S)$. 
\end{definition}
Now we reconsider the main problem of $\Pi_2$-cut introduction and reformulate the 
necessary conditions. Instead of finding a substitution for 
$X$ such that the \seHs
\begin{align*}
& F[\tuples{x}\backslash U_1],\bigvee\limits_{i=1}^pX\alpha t_i\to 
\bigwedge\limits_{j=1}^mXr_j\beta_j\seq 
G[\tuples{y}\backslash U_2]
\end{align*}
is a tautology we have to find a substitution $\theta\colon [X \setminus \hat{A}]$ such that, for all leaves $S\in \text{DNTA}(\psi )$ of the reduced 
representation $R$, the sequent 
\begin{align*}
S \circ (\bigvee\limits_{i=1}^p\hat{A}\alpha t_i\to \bigwedge\limits_{j=1}^m \hat{A}r_j\beta_j\seq ) 
\end{align*}
is a tautology. Hence, we can divide it into two problems: 
\begin{itemize}
\item the {\em $\beta$-problem}  of $S$, $S \circ (Xr_1\beta_1,\ldots ,Xr_m\beta_m \seq )$ and 
\item the {\em $\alpha$-problem} of $S$, $S \circ (\; \seq  X\alpha t_1,\ldots ,X\alpha t_p )$.
\end{itemize}
and say that $V$ is a solution of the $\beta$-problem and the $\alpha$-problem if there is a substitution $\theta$ for $X$ such that $\theta$ is of the form $\lambda xy.V$ where in $V$ may not occur $\beta_1,\ldots ,\beta_m$, or $\alpha$ and the sequents of the $\beta$-problem and $\alpha$-problem for all $S\in\text{DNTA}(\psi)$ become tautologies. A shared solution for the $\beta$-problem and the $\alpha$-problem is also a solution for the \seHs\!\!.

Now we want to find formulas in disjunctive normal form that are solutions. Therefore, we assume an arbitrary starting set $\mathcal{A}^\mathcal{O}$ that is a collection of literals not containing $\beta_1,\ldots ,\beta_m$, or $\alpha$ (see Definition \ref{def:starting set}). Again, we can consider $\mathcal{O}$ as being the empty set $\emptyset$. The characterization we give in this section finds for a given starting set all possible solutions of the $\beta$- and $\alpha$-problem and, therefore, of the \seHs\!\!. Assuming a finite starting set, we can implement a terminating algorithm to find all solutions that can be built by the literals in the starting set based on this characterization. Hence, after defining the characterization, we have to construct starting sets containing solutions. In Section \ref{sec.Gunify}, we define a method constructing finite starting sets, such that this method finds always a solution if there is a \emph{balanced solution} (see Definition \ref{def:Balanced solution}). However, the concept of balanced solution is not needed in the characterization below.

A solution of the \seHs has to solve the $\beta$-problem as well as the $\alpha$-problem. Therefore, we formulate the restrictions given by them and eliminate gradually all subsets of $\mathcal{A}$ that are not solutions. First we consider the $\beta$-problem. In Definition \ref{def:Set of possible sets of clauses}, we eliminate all subsets of $\mathcal{A}^\emptyset =\mathcal{A}$ that do not turn the $\beta$-problem into a tautology. Consider the sequent of the $\beta$-problem: If we substitute a possible solution in DNF for $X$ then the sequent branches into all possible sequents with one clause for each $Xr_1\beta_1,\ldots ,Xr_m\beta_m$ on the left hand-side of the sequent. In Definition~\ref{def:Set of possible sets of clauses}, the choice of these $m$ arbitrary clauses is represented by the $m$-tuples $(C_1,\ldots ,C_m)$ where $C_i$ is instantiated with $r_i$ and $\beta_i$ for $i\in\mathbb{N}_m$. For each choice we guarantee the provability by demanding an axiomatic constant ($T_1$), an axiomatic literal ($T_2$), or an interactive literal ($T_3$). These literals cover every possible case in which there is a literal and its dual on the left hand-side of the sequent. Finally we can shift the negated literal to the right and receive a tautological axiom. 
\begin{definition}[Set of possible sets of clauses]
\label{def:Set of possible sets of clauses}
Let $R$ be a given reduced representation of an \seHs $S(X)$ and $\psi$ be a maximal 
\textbf{G3c}-derivation of $R$. Let $S\in\text{DNTA}(\psi )$, $m$ be the 
\AllMul (see Definition \ref{def.schem-pi2-grammar}), $\mathcal{C}$ be a set of clauses\footnote{A clause is a set of literals.}. Let $\vec{\mathcal{C}}_m$ be the set of all $m$-tuples $(C_1,\ldots ,C_m)$ where $C_i\in\mathcal{C}$ for $i\in\mathbb{N}_m$. If $\vec{C}\in\vec{\mathcal{C}}_m$, $\vec{C}=(C_1,\ldots ,C_m)$, and $i\in\mathbb{N}_m$ we write $\vec{C}(i)$ for $C_i$. Furthermore, let $\mathcal{A}$ be a starting set and $N(S),B(S)$ as in Definition \ref{def:DNTA}. \\[1ex]
We define the three conditions - ($T_1$) axiomatic constant, ($T_2$) axiomatic 
literal, ($T_3$) interactive literal - 
\begin{align*}
T_1(\vec{C},S):= &\; \exists i\in\mathbb{N}_m\exists L\in\vec{C}(i). 
L[x\backslash r_i]\in \overline{N(S)} \\ 
&\; \text{ where }\overline{N(S)}\text{ denotes the dualized set }N(S), \\
T_2(\vec{C},S):= &\; \exists i\in\mathbb{N}_m\exists L\in \vec{C}(i). 
L[x\backslash r_{i},y\backslash\beta_{i}]\in \overline{B(S)}, \\
T_3(\vec{C}):= &\; \exists i,j\in\mathbb{N}_m 
\exists L\in\vec{C}(i)\exists Q\in\vec{C}(j).L[x\backslash r_i,y\backslash\beta_i]= 
\overline{Q}[x\backslash r_j,y\backslash\beta_j], \\ 
&\; \text{and} \\ 
T(\vec{C},S):= &\; T_1(\vec{C},S)\,\textbf{or}\, 
T_2(\vec{C},S)\,\textbf{or}\, T_3(\vec{C}).
\end{align*}
Then 
\begin{align*}
Cl(\mathcal{A}):= &\; \{\mathcal{C}\subseteq\mathcal{A}\; |\;\forall \vec{C}\in 
\vec{\mathcal{C}}_m\forall S\in\text{DNTA}(\psi ).T(\vec{C},S)\} 
\end{align*}
is the set of possible sets of clauses. 
\end{definition}
The next step guarantees that the sequent of the $\alpha$-problem becomes a tautological axiom. Consider the following example. 

\begin{example}
\label{exa.explaining}
Let 
\begin{center}
\begin{fCenter}
\Axiom$ \fCenter $
\LeftLabel{\textit{Ax}}
\UnaryInf$P(c,fc), Q(c,gc),\Gamma \;\fCenter\seq \Delta ,P(c,fc),Q(c,fc)$
\LeftLabel{$\lor\colon$r}
\UnaryInf$P(c,fc), Q(c,gc),\Gamma \;\fCenter\seq \Delta ,P(c,fc)\lor Q(c,fc)$
\LeftLabel{$\exr$}
\UnaryInf$P(c,fc), Q(c,gc),\Gamma \;\fCenter\seq \exists x. P(c,x)\lor Q(c,x)$
\LeftLabel{$\land\colon$l}
\UnaryInf$P(c,fc)\land Q(c,gc),\Gamma \;\fCenter\seq \exists x. P(c,x)\lor Q(c,x)$
\Axiom$ \;\fCenter \pi_r $
\LeftLabel{$\lor\colon$l}
\BinaryInf$(P(c,fc)\land Q(c,gc))\lor (P(c,gc)\land Q(c,fc)),\Gamma
\,\fCenter\seq\! \exists x. P(c,x)\lor Q(c,x)$
\LeftLabel{$\alll$}
\UnaryInf$\forall x.(P(x,fx)\land Q(x,gx))\lor (P(x,gx)\land Q(x,fx)) 
\,\fCenter\seq\! \exists x. P(c,x)\lor Q(c,x)$
\DisplayProof
\end{fCenter}
\end{center}
with $\pi_r:=$
\begin{center}
\begin{fCenter}
\Axiom$ \fCenter $
\LeftLabel{\textit{Ax}}
\UnaryInf$P(c,gc), Q(c,fc),\Gamma \;\fCenter\seq \Delta ,P(c,gc), Q(c,gc)$
\LeftLabel{$\lor\colon$r}
\UnaryInf$P(c,gc), Q(c,fc),\Gamma \;\fCenter\seq \Delta ,P(c,gc)\lor Q(c,gc)$
\LeftLabel{$\exr$}
\UnaryInf$P(c,gc), Q(c,fc),\Gamma \;\fCenter\seq \exists x. P(c,x)\lor Q(c,x)$
\LeftLabel{$\land\colon$l}
\UnaryInf$P(c,gc)\land Q(c,fc),\Gamma \;\fCenter\seq \exists x. P(c,x)\lor Q(c,x)$
\DisplayProof
\end{fCenter}
\end{center}
and 
\begin{align*}
\Gamma := &\; \forall x.(P(x,fx)\land Q(x,gx))\lor (P(x,gx)\land Q(x,fx)) \\
\Delta := &\; \exists x. P(c,x)\lor Q(c,x)
\end{align*}
be a given proof of the sequent 
\begin{align*}
\forall x.(P(x,fx)\land Q(x,gx))\lor 
(P(x,gx)\land Q(x,fx))\seq\exists x. P(c,x)\lor Q(c,x).
\end{align*} 
Furthermore we assume the schematic $\Pi_2$-grammar $G=\langle \tau ,N,\Sigma ,\Pr\rangle$ with $N=\{\tau ,\alpha ,\beta\}$ and $\Pr =  \{ \tau\to h_\Gamma \alpha ,\tau\to h_\Delta \beta ,\alpha\to c,\beta\to fc\; |\; gc\} $ where $\tau\to h_\Gamma$ is the only $\Gamma$-production and $\tau\to h_\Delta \beta$ is the only $\Delta$-production according to Definition \ref{def.extHseq-grammar}. Hence, the reduced representation of the \seHs is given by 
\begin{align*}
(P(\alpha ,f\alpha )\land Q(\alpha ,g\alpha ))\lor (P(\alpha ,g\alpha ) 
\land Q(\alpha ,f\alpha )) &\; \seq P(c,\beta )\lor Q(c,\beta ).
\end{align*}
A maximal \textbf{G3c}-derivation $\psi$ gives us the set of non-tautological axioms 
\begin{align*}
\text{DNTA}(\psi ) =\{ S_1;S_2\}=\{ & P(\alpha ,f\alpha ),Q(\alpha ,g\alpha ),\neg P(c,\beta ), 
\neg Q(c,\beta )\seq ; \\ 
& P(\alpha ,g\alpha ),Q(\alpha ,f\alpha ),\neg P(c,\beta ),\neg Q(c,\beta ) 
\seq \; \} .
\end{align*}
Now we consider the starting set $\mathcal{A}=\{ \{ P(x,y),Q(x,y)\}\}$ and compute 
$Cl(\mathcal{A})$. The only subsets of $\mathcal{A}$ are the empty set and $\mathcal{A}$ itself. The empty set does not fulfil any of the conditions of a possible set of clauses. The only clause in $\mathcal{A}$ is $\{ P(x,y),Q(x,y)\}$ which contains for each $S\in\text{DNTA}(\psi )$ an axiomatic literal, i.e.\ $P(c,\beta )$ and $Q(c,\beta )$. Thus, $Cl(\mathcal{A})=\{\mathcal{A}\}$. But the \seHs where $X$ is replaced with $\lambda xy.P(x,y)\land Q(x,y)$ is not a tautology. A maximal \textbf{G3c}-derivation of 
\begin{align*}
(P(\alpha ,f\alpha )\land Q(\alpha ,g\alpha ))\lor (P(\alpha ,g\alpha ) 
\land Q(\alpha ,f\alpha )), &\; \\ 
(P(\alpha ,f\alpha )\land Q(\alpha ,f\alpha ))\lor(P(\alpha ,g\alpha )\land 
Q(\alpha ,g\alpha )) &\; \\ 
\to P(c,\beta )\land Q(c,\beta ) &\; 
\seq P(c,\beta ), Q(c,\beta )
\end{align*}
gives us the the non-tautological leaves 
\begin{align*}
\{ 
&\; P(\alpha ,f\alpha ),Q(\alpha ,g\alpha )\seq P(c,\beta ),Q(c,\beta ) 
,Q(\alpha ,f\alpha ),P(\alpha ,g\alpha ); \\ 
&\; P(\alpha ,g\alpha ),Q(\alpha ,f\alpha )\seq P(c,\beta ),Q(c,\beta ) 
,P(\alpha ,f\alpha ),Q(\alpha ,g\alpha ) \; \} .
\end{align*}
This is due to the existence of a leaf $S$ in DNTA$(\psi )$ that fulfils the following 
property: we find for each term $f\alpha$ and $g\alpha$ an atom $P(\alpha ,f\alpha )$ or $Q(\alpha ,g\alpha )$ that does not appear in the leaf $S$. 
\end{example}
In Definition \ref{def:allowed clauses} we generalize this property and define a 
set $I(S)$ for each leaf $S$ that contains only {\em allowed clauses}. Clauses as $\{P(x,y),Q(x,y)\}$ in the previous example are excluded. 

To understand the necessity of this property for all clauses, we have to examine to behaviour of a set of clauses on the right of a sequent, i.e.\ the $\alpha$-problem. Assume a single clause $\{ L\}\cup R$ and a single instantiation $p=1$ such that $L[x\backslash\alpha ,y\backslash t_1]\notin A(S)$ for the non-tautological leaf $S$. Then neither $S\circ (\,\seq L[x\backslash\alpha ,y\backslash t_1])$ nor $S\circ (\,\seq L[x\backslash\alpha ,y\backslash t_1]\land R[x\backslash\alpha ,y\backslash t_1])$ is provable. If we extend the number of instantiations $p$ without gaining an instantiation  $1\le j\le p$ such that for all literals $Q$ in $\{ L\}\cup R$ the substituted variant $Q[x\backslash\alpha ,y\backslash t_j]$ is not an element of $A(S)$ the sequent 
\begin{align*}
S\circ (\,\seq L[x\backslash\alpha ,y\backslash t_1]\land R[x\backslash\alpha ,y\backslash t_1])\circ\ldots\circ (\,\seq L[x\backslash\alpha ,y\backslash t_p]\land R[x\backslash\alpha ,y\backslash t_p])
\end{align*}
stays non-tautological. If we consider the case that there is more than a single clause and one clause does not fulfill the described property, we can eliminate this clause. Note that if you consider the clauses made of formulas that are solutions of the $\beta$-problem, those clauses are solutions of the $\beta$-problem themselves, i.e.\ we are allowed to eliminate all but one clause without making the solution invalid.

Definition \ref{def:allowed clauses} constructs the set of all clauses with the described property for a given leaf $S$.
\begin{definition}
\label{def:allowed clauses}
Let $R$ be a given reduced representation of an \seHs $S(X)$, $\psi$ be a maximal 
\textbf{G3c}-derivation of $R$, and $S\in\text{DNTA}(\psi )$. $A'(S)$ is defined 
as in Definition \ref{def:DNTA}. Let 
\begin{align*}
M(k)\subseteq &\; A'(S)\text{ such that }|M(k)|=k. \text{ Then}
\end{align*}
\begin{align*}
\AllCl{S}:= &\; \bigcup\limits_{k\le |A'(S)|}\{ M(k)\; |\;\exists i\in\mathbb{N}_p 
\forall L\in M(k).L[x\backslash\alpha ,y\backslash t_i]\in A(S) \} .
\end{align*}
is the set of {\em allowed clauses}. \\[1ex]
Let $\Ccal$ be a starting set as defined in Definition \ref{def:starting set}. We denote the set of \emph{refined allowed clauses} $\RI{S}$ as 
\begin{align*}
\RI{S}\colon =\AllCl{S}\cap \Ccal .
\end{align*} 
\end{definition}
A useful tool for the application of the set of allowed clauses in practice can be 
obtained from the following proposition.
\begin{proposition}
\label{pro.allowedclauses}
Let $R$ be a given reduced representation of an \seHs $S(X)$, $\psi$ be a maximal 
\textbf{G3c}-derivation of $R$, $S\in\text{DNTA}(\psi )$ and \AllCl{S} the set of 
allowed clauses. If $I$ is an element of \AllCl{S} and $J$ is a subset of $I$ then $J$ is an element of \AllCl{S}. 
\end{proposition}
\begin{proof}
The claim trivially holds.
\end{proof}
Now we can formulate the conditions that guarantee the provability of the sequent of the $\alpha$-problem. Again, we need for each non-tautological leaf an axiomatic constant, an axiomatic literal, or an interactive literal. The differences to Definition \ref{def:Set of possible sets of clauses} are due to the different behaviour of formulas in disjunctive normal form on different sides of a sequent in a proof in sequent calculus. We define two sets of solution candidates; one by using the allowed clauses and the other by using the refined allowed clauses (see Definition \ref{def:allowed clauses}).
\begin{definition}[Set of solution candidates]
\label{def:solution set}
Let $R$ be a given reduced representation of an \seHs $S(X)$ and $\psi$ be a maximal 
\textbf{G3c}-derivation of $R$. Let $S\in\text{DNTA}(\psi )$, $p$ be the 
\ExMul\!\!, and $\mathcal{C}$ be a set of clauses. 
Let $\vec{\mathcal{L}}_p(C)$ be the set of all $p$-tuples $(L_1,\ldots ,L_p)$ where 
$L_i\in C$ for $i\in\mathbb{N}_p$ and $C\in\mathcal{C}$. If $\vec{L}\in 
\vec{\mathcal{L}}_p(C)$, $\vec{L}=(L_1,\ldots ,L_p)$, and $i\in\mathbb{N}_p$ we write 
$\vec{L}(i)$ for $L_i$. Let $\vec{\mathcal{C}}=\prod_{C\in\mathcal{C}} 
\vec{\mathcal{L}}_p(C)$ be the Cartesian product of the subspaces 
$\vec{\mathcal{L}}_p(C)$ where $C\in\mathcal{C}$. If $\vec{C}\in\vec{\mathcal{C}}$ and 
$\vec{L}\in\vec{\mathcal{L}}_p(C)$ is the element of $\vec{C}$ that corresponds to 
the subspace $\vec{\mathcal{L}}_p(C)$ we write $L(C,i)$ for $\vec{L}(i)$. 
Furthermore, let $\mathcal{A}$ be a starting set and $\mathcal{D}$ be either the set of allowed clauses or the set of refined allowed clauses. \\[1ex]
We define the three 
conditions - ($T'_1$) axiomatic constant, ($T'_2$) axiomatic 
literal, ($T'_3$) interactive literal - 
\begin{align*}
T'_1(\mathcal{C},\vec{C},S):= &\; \exists C\in\mathcal{C}\exists i\in\mathbb{N}_p. 
L(C,i)[y\backslash t_i]\in N(S), \\
T'_2(\mathcal{C},\vec{C},S):= &\; \exists C\in\mathcal{C}\exists I\in \mathcal{D} \forall i\in 
\mathbb{N}_p. L(C,i) \in I , \\ 
T'_3(\mathcal{C},\vec{C}):= &\; \exists C,D\in\mathcal{C}\exists i,j\in\mathbb{N}_p. 
L(C,i)[x\backslash\alpha ,y\backslash t_i]=\overline{L(D,j)} 
[x\backslash\alpha ,y\backslash t_j], \\ 
&\; \text{and} \\ 
T'(\mathcal{C},\vec{C},S):= &\; T'_1(\mathcal{C},\vec{C},S)\,\textbf{or}\, 
T'_2(\mathcal{C},\vec{C},S)\,\textbf{or}\, T'_3(\mathcal{C},\vec{C}). 
\end{align*}
Then, for $\mathcal{D} =  \AllCl{S}$, the set
\begin{align*}
Sol(\mathcal{A}):= &\; \{\mathcal{C}\in Cl(\mathcal{A})\; |\;\forall 
\vec{C}\in\vec{\mathcal{C}}\forall S\in\text{DNTA}(\psi ). 
T'(\mathcal{C},\vec{C},S)\}
\end{align*}
is called the {\em set of solution candidates}, and for $\mathcal{D} =  \RI{S}$ the {\em set of refined solution candidates}  (for a given starting set and a given \seHs 
in DNF). 
\end{definition}
\begin{theorem}
\label{the:Sol=Sol}
The set of refined solution candidates coincides with the set of solution candidates.
\end{theorem}
\begin{proof}
If $\Ccal$ is a refined solution candidate then $\Ccal$ is a solution candidate by definition. 

\smallskip

Assume $\Ccal$ is a solution candidate. The only difference to a refined solution candidate is the axiomatic literal. $\Ccal$ being a solution candidate, there is a clause $C$ in $\Ccal$ and an allowed clause $I$ such that for all $i\in\mathbb{N}_p$ the literal $L(C,i)$ is an element of $I$. Furthermore, $\Ccal$ is an element of $Cl(\mathcal{A})$, i.e.\ $\Ccal\subseteq\mathcal{A}$. Altogether, $L(C,i)$ is a literal occurring in $\mathcal{A}$ and there is a subset $J$ of $I$ such that $I\supseteq J=\bigcup_{i\in\mathbb{N}_p}L(C,i)$. By Proposition \ref{pro.allowedclauses}, $J$ is an element of \AllCl{S} and, therefore, $J$ is a refined allowed clause. Since it is always possible to construct a refined allowed clause for a given axiomatic literal, $\Ccal$ is also an element of the set of refined solution candidates. 
\end{proof}
\begin{example}
If we consider Example \ref{exa.explaining} again and compute $Sol(\mathcal{A})$ we will get the empty set. $\RI{S}$ consists of all clauses $C$ that are an element of the starting set $\mathcal{A}$ such that there is an index $i\in\mathbb{N}_p$ for all literals $L$ in the clause $C$ with $L[x\backslash\alpha ,y\backslash t_i]\in A(S)$. For the two non-tautological leaves $S_1$ and $S_2$, we get the sequents 
\begin{align*}
A(S_1): &\; P(\alpha ,f\alpha ),Q(\alpha ,g\alpha )\seq \; ,\\
A(S_2): &\; P(\alpha ,g\alpha ),Q(\alpha ,f\alpha )\seq .
\end{align*}
The only element of the starting set is $\{ P(x,y),Q(x,y)\}$, and $p=2$. For $i=1$, the substituted literals are $\{ P(\alpha ,f\alpha ),Q(\alpha ,f\alpha )\}$ and for $i=2$ the substituted literals are $\{ P(\alpha ,g\alpha ),Q(\alpha ,g\alpha )\}$. In both cases and independent from the chosen leaf ($j\in\mathbb{N}_2$), one of the substituted literals is not an element of $A(S_j)$. For instance: Since $Q(\alpha ,f\alpha )$ of the substituted literals $\{ P(\alpha ,f\alpha ),Q(\alpha ,f\alpha )\}$ does not appear in 
\begin{align*}
A(P(\alpha ,f\alpha ),Q(\alpha ,g\alpha ),\neg P(c,\beta),\neg Q(c,\beta )\seq)= P(\alpha ,f\alpha ),Q(\alpha ,g\alpha )\seq 
\end{align*}
and $P(\alpha ,g\alpha )$ of the substituted literals $\{ P(\alpha ,g\alpha ),Q(\alpha ,g\alpha )\}$ does not appear in 
\begin{align*}
P(\alpha ,f\alpha ),Q(\alpha ,g\alpha )\seq 
\end{align*}
we conclude 
\begin{align*}
\RI{P(\alpha ,f\alpha ),Q(\alpha ,g\alpha ),\neg P(c,\beta),\neg Q(c,\beta )\seq}= \emptyset .
\end{align*}
Hence, $Sol(\mathcal{A})$ is empty.
\end{example}
\begin{remark}
The condition of Definition \ref{def:allowed clauses} of allowed clauses is necessary. 
\end{remark}
\begin{proof}
Assume a solution $\mathcal{S}$ of an \seHs $S(X)$ in disjunctive normal form such that no clause fulfils the condition of Definition \ref{def:allowed clauses}, i.e.\ there is a leaf $S$ for all clauses $C$ and all $i\in\mathbb{N}_p$ such that we find literals $L_{i,C}$ where $L_{i,C}[x\backslash\alpha ,y\backslash t_i]\notin A(S)$. Let $\mathcal{L}$ be the set of all $L_{i,C}[x\backslash\alpha ,y\backslash t_i]$ with $C\in\mathcal{S}$ and $i\in\mathbb{N}_p$. Then $A(S)\circ\{\seq\mathcal{L}\}$ is a non-tautological sequent whose initial sequent appears in every proof of $S(\mathcal{S})$. Therefore, $\mathcal{S}$ cannot be a solution.
\end{proof}
We can show that each solution candidate is actually a solution. 
\begin{theorem}\label{theo.soundness}
Let 
\begin{align*}
& F[\tuples{x}\backslash U_1],\bigvee\limits_{i=1}^pX\alpha t_i\to 
\bigwedge\limits_{j=1}^mXr_j\beta_j\seq 
G[\tuples{y}\backslash U_2]
\end{align*}
be an \seHs \!\!, $Sol(\mathcal{A})\neq\emptyset$ be 
defined as in Definition \ref{def:solution set} for a given starting set 
$\mathcal{A}$, and $\Ccal \in Sol(\mathcal{A})$. Let $E\colon \DNF(\Ccal)$ be the formula 
in DNF corresponding to $\Ccal$ and $\Ehat = \lambda x y.E$. Then 
\begin{align*}
F[\tuples{x}\backslash U_1],\bigvee\limits_{i=1}^p\Ehat\alpha t_i\to 
\bigwedge\limits_{j=1}^m\Ehat r_j\beta_j\seq 
G[\tuples{y}\backslash U_2]
\end{align*}
is a tautology, i.e.\ a solution candidate and a refined solution candidate are solutions. 
\end{theorem}
\begin{proof}
If we want to prove that 
\begin{align*}
& F[\tuples{x}\backslash U_1],\bigvee\limits_{i=1}^p\Ehat \alpha t_i\to 
\bigwedge\limits_{j=1}^m\Ehat r_j\beta_j\seq G[\tuples{y}\backslash U_2]
\end{align*}
is a tautology we have to prove the sequents 
\begin{align}
F[\tuples{x}\backslash U_1],\Ehat r_1\beta_1, \ldots ,\Ehat r_m\beta_m & \seq G[\tuples{y} 
\backslash U_2] \text{ and} \label{equ:case 1 proof soundness} \\ 
F[\tuples{x}\backslash U_1] & \seq \Ehat \alpha t_1,\ldots ,\Ehat \alpha t_p, 
G[\tuples{y}\backslash U_2] \label{equ:case 2 proof soundness}.
\end{align}
First we will show that the sequent \eqref{equ:case 1 proof soundness} is a tautology. 
Assume it is not provable. The formulas $\Ehat r_j\beta_j$ for $j\in\mathbb{N}_m$ 
are formulas in DNF which can be interpreted as sets of sets of literals. In 
\textbf{G3c} a disjunction on the left 
\begin{align*}
\Delta ,\bigvee\limits_{i\in I}A_i & \seq\Gamma
\end{align*}
is considered to be true if the sequents $\Delta ,A_i \seq \Gamma$
are true for all $i\in I$. Hence, if the sequent 
\eqref{equ:case 1 proof soundness} is 
not provable then there are clauses $C_1,\ldots ,C_m$ in $\mathcal{C}$ such that, for 
$E_i = \lambda x y.\DNF(\{ C_i\})$, 
\begin{align}
F[\tuples{x}\backslash U_1],E_1 r_1\beta_1, \ldots ,E_m r_m\beta_m & \seq G[\tuples{y} 
\backslash U_2] \label{equ:interStepSoundness1}
\end{align}
is not provable. 

Now we apply the rules of a maximal \textbf{G3c}-derivation $\psi$ of 
$R$ and let the instantiations $E_1r_1\beta_1,\ldots ,E_m r_m\beta_m$ be untouched. 
The non-tautological axioms of $R$ can be represented by DNTA$(\psi )$ where $A(S) 
\circ B(S)\circ N(S)$ for $S\in 
\text{DNTA}(\psi )$ is defined as in Definition \ref{def:DNTA}. 
Hence, we can add the literals of the clauses $E_1 r_1\beta_1,\ldots ,E_m r_m\beta_m$ 
to $B(S)$ and $N(S)$ to get a representation of the non-tautological axioms of a 
maximal \textbf{G3c}-derivation of the sequent \eqref{equ:interStepSoundness1}. The 
part of literals that has been added to $B(S)$ will be denoted by $B$ and the part 
that has been added to $N(S)$ will be denoted by $N$. If the sequent 
\eqref{equ:interStepSoundness1} is not provable there has 
to be a non-tautological axiom $S'$, i.e.\ 
\begin{align*}
\forall L,Q\in S'\circ B\circ N.L\neq \overline{Q}.
\end{align*}
But this implies that there is no axiomatic constant ($T_1$), axiomatic literal 
($T_2$), or interactive literal ($T_3$). Thus it contradicts 
Definition \ref{def:Set of possible sets of clauses} and the sequent 
\eqref{equ:case 1 proof soundness} is provable. 

\smallskip

Now we have to prove that the sequent \eqref{equ:case 2 proof soundness} is a tautology. 
We will again assume that it is not a tautology and derive a contradiction. 
Let us assume there are $k$ clauses $C_1,\ldots ,C_k$ in $\mathcal{C}$. Thus, 
the sequent 
\begin{align*}
F[\tuples{x}\backslash U_1]\seq  E_1\alpha t_1,\ldots ,E_k\alpha t_1, \ldots 
,E_1\alpha t_p,\ldots ,E_k\alpha t_p,G[\tuples{y}\backslash U_2]
\end{align*}
where $E_i=\lambda xy.\DNF (\{ C_i\})$ is also not a tautology. Now we apply 
again the rules of a maximal 
\textbf{G3c}-derivation $\psi$ of $R$ and let the clauses be untouched. 
Given that the sequent above is not a tautology, there is also a leaf $S'$ in 
the derivation that is not a tautology. We find in each $E_i\alpha t_j$ for $i\in\{ 
1,\ldots ,k\}$ and $j\in\{ 1,\ldots ,p\}$ a literal $L_l$ with $l=(i-1)\cdot p+j$ such 
that 
\begin{align*}
S':= & S\circ (\;\seq L_1,\ldots ,L_{k\cdot p}) 
\end{align*}
is not a tautology. But this implies that there is neither an axiomatic constant 
($T'_1$), nor an axiomatic literal ($T'_2$), nor an interactive literal ($T'_3$) and 
contradicts Definition \ref{def:solution set}. Hence, the sequent 
\eqref{equ:case 2 proof soundness} is a tautology. 
\end{proof}
Furthermore we can show that the Definitions 
\ref{def:Set of possible sets of clauses} and \ref{def:solution set} do not eliminate 
solutions, i.e.\ if there is a subset in the starting set $\mathcal{A}$ that is a 
solution then this set will also be an element of $Sol(\mathcal{A})$. 
\begin{theorem}[Partial completeness]
\label{the:Partial completeness}
Let
\begin{align*}
& F[\tuples{x}\backslash U_1],\bigvee\limits_{i=1}^pX\alpha t_i\to 
\bigwedge\limits_{j=1}^mXr_j\beta_j\seq G[\tuples{y}\backslash U_2]
\end{align*}
be an \seHs\!\!, $\mathcal{A}$ be a starting set, and $\mathcal{C}\subseteq\mathcal{A}$. Let $E = \DNF(\Ccal)$ be the formula in DNF corresponding to $\Ccal$ and $\Ehat = \lambda x y.E$. If 
\begin{align*}
& F[\tuples{x}\backslash U_1],\bigvee\limits_{i=1}^p\Ehat\alpha t_i\to 
\bigwedge\limits_{j=1}^m\Ehat r_j\beta_j\seq G[\tuples{y}\backslash U_2]
\end{align*}
is a tautology then $\mathcal{C}\in Sol(\mathcal{A})$ where $Sol(\mathcal{A})$ is as in Definition \ref{def:solution set}.
%$Sol(\mathcal{A})$ is not empty.
\end{theorem}

\begin{proof}
At first we assume that there is a solution $\mathcal{C}$ for the \seHs 
that is a subset of the starting set $\mathcal{A}$ but 
$\mathcal{C}$ is not an element of $Cl(\mathcal{A})$ of Definition 
\ref{def:Set of possible sets of clauses}. Let $\psi$ be a maximal 
\textbf{G3c}-derivation. If $\mathcal{C}$ is not an element of 
$Cl(\mathcal{A})$ but $\mathcal{C}\subseteq\mathcal{A}$ then 
\begin{align*}
&\; \exists\vec{C}\in\vec{\mathcal{C}}_m\exists S\in\text{DNTA}(\psi ). 
\neg T(\vec{C},S) \\ 
&\; \text{with } \\ 
\neg T(\vec{C},S):= &\; \neg T_1(\vec{C},S)\,\textbf{and}\, 
\neg T_2(\vec{C},S)\,\textbf{and}\,\neg T_3(\vec{C}), \\
\neg T_1(\vec{C},S):= &\; \forall i\in\mathbb{N}_m\forall L\in\vec{C}(i), 
L[x\backslash r_i]\notin \overline{N(S)} \\
&\; \text{ where }\overline{N(S)}\text{ denotes the dualized set }N(S), \\
\neg T_2(\vec{C},S):= &\; \forall i\in\mathbb{N}_m\forall L\in \vec{C}(i). 
L[x\backslash r_{i},y\backslash\beta_{i}]\notin \overline{B(S)}, \\ 
&\; \text{and} \\ 
\neg T_3(\vec{C}):= &\; \forall i,j\in\mathbb{N}_m 
\forall L\in\vec{C}(i)\forall Q\in\vec{C}(j).L[x\backslash r_i,y\backslash\beta_i]\neq 
\overline{Q}[x\backslash r_j,y\backslash\beta_j] 
\end{align*}
where $\vec{\mathcal{C}}_m$ is defined as in Definition 
\ref{def:Set of possible sets of clauses}.
Let $S$ be an element of $\text{DNTA}(\psi )$ of the form 
$A(S)\circ B(S)\circ N(S)$. There is a $m$-tuple of clauses 
$(C_1,\ldots ,C_m)$ with $C_i\in\mathcal{C}$ for 
$i\in\mathbb{N}_m$ fulfilling the following property. Let $E_k:=\DNF (\{ C_k\})$ 
and $\Ehat_k:=\lambda xy.E_k$ for $k\in\mathbb{N}_m$ then 
\begin{align*}
& S\circ (\Ehat_1r_1\beta_1,\ldots ,\Ehat_mr_m\beta_m\seq )
\end{align*}
is not a tautology. But then also 
\begin{align*}
& F[\tuples{x}\backslash U_1], 
\bigwedge\limits_{j=1}^m\Ehat r_j\beta_j\seq G[\tuples{y}\backslash U_2]
\end{align*}
is not a tautology, i.e.\ $\mathcal{C}$ is not a solution and by contradiction 
$\mathcal{C}\in Cl(\mathcal{A})$. \\ 
Now we assume $\mathcal{C}\notin Sol(\mathcal{A})$. Given that 
$\mathcal{C}\in Cl(\mathcal{A})$ we find an element $\vec{C}\in\vec{\mathcal{C}}$ and 
a leaf $S\in\text{DNTA}(\psi )$ such that 
\begin{align*}
&\; \forall C\in\mathcal{C}\ \forall i\in\mathbb{N}_p.\; L(C,i)[y\backslash t_i] 
\notin N(S) \\
\textbf{and}\, &\; \forall C\in\mathcal{C}\ \forall I\in \RI{S}\ \exists i\in\mathbb{N}_p. \;
L(C,i)\notin I \\ 
\textbf{and}\, &\; \forall C,D\in\mathcal{C}\ \forall i,j\in\mathbb{N}_p.\; 
L(C,i)[x\backslash\alpha ,y\backslash t_i]\neq\overline{L(D,j)} 
[x\backslash\alpha ,y\backslash t_j] 
\end{align*}
where $\vec{\mathcal{C}}$ is defined as in Definition \ref{def:solution set} and $\RI{S}$ is defined as in Definition \ref{def:allowed clauses}. Let $k$ be the number of clauses in 
$\mathcal{C}$ then we find for all of them $p$ literals 
\begin{align*}
L(C_1,1),\ldots ,L(C_1,p),\ldots ,L(C_k,1),\ldots ,L(C_k,p)
\end{align*}
where $C_1,\ldots ,C_k$ are the $k$ clauses such that the sequent 
\begin{align*}
S\circ (\; \seq &\; \hat{L}_1\alpha t_1,\ldots ,\hat{L}_p\alpha t_p,\ldots 
,\hat{L}_{(k-1)\cdot p+1}\alpha t_1,\ldots ,\hat{L}_{k\cdot p}\alpha t_p )
\end{align*}
with $\hat{L}_q:= \lambda xy.\DNF (\{ L(C_i,j)\} )$ for $q\in\mathbb{N}_{k\cdot p},$ $q=(i-1)\cdot p+j,$ $i\in\mathbb{N}_k,$ and $j\in\mathbb{N}_p$ does not contain an axiomatic constant ($T'_1$), an axiomatic literal ($T'_2$), or an interactive literal ($T'_3$). Furthermore, $S$ is not a tautology and the literals 
\begin{align*}
\hat{L}_1\alpha t_1,\ldots ,\hat{L}_p\alpha t_p,\ldots 
,\hat{L}_{(k-1)\cdot p+1}\alpha t_1,\ldots ,\hat{L}_{k\cdot p}\alpha t_p
\end{align*} 
do not contain the eigenvariables $\beta_1,\ldots ,\beta_m$. Hence, none of the literals occurs in $A(S),B(S),$ or $N(S)$ and the found sequent is not a tautology. 
This contradicts the assumption that $\mathcal{C}$ is a solution and is not an element 
of $Sol(\mathcal{A})$. Thus, $\mathcal{C}\in Sol(\mathcal{A})$. 
\end{proof}
To prove full completeness we need a starting set for every possible reduced representation. In Section \ref{sec.Gunify} we show that we can define starting sets, provided a {\em balanced solution} of the \seHs exists. The general case is not treated in this paper. But the characterization is complete in that it will always compute a solution if by the clauses of the starting set a solution can be constructed. So the problem reduces to find starting sets.

Finally we show that, whenever  $Sol(\mathcal{A}) \neq\emptyset$ for a given starting set $\mathcal{A}$, the problem of $\Pi_2$-cut introduction is solvable.

\begin{theorem}\label{the.main}
Let 
\begin{align*}
& F[\tuples{x}\backslash U_1],\bigvee\limits_{i=1}^pX\alpha t_i\to 
\bigwedge\limits_{j=1}^mXr_j\beta_j\seq 
G[\tuples{y}\backslash U_2]
\end{align*}
be an \seHs \!\! corresponding to a Herbrand sequent of a cut-free proof of $S$ and 
a grammar $\mathcal{G}$ covering the Herbrand term set of $S$.  Let $Sol(\mathcal{A}) 
\neq\emptyset$ be 
defined as in Definition \ref{def:solution set} for a given starting set 
$\mathcal{A}$, and $\Ccal \in Sol(\mathcal{A})$. Let $E = \DNF(\Ccal)$ be the formula 
in DNF corresponding to $\Ccal$ and $V(E) \subseteq \{x,y\}$.  Then there exists a proof of $S$ with one cut and the cut formula $\forall x \exists y.E$ 
\end{theorem}
\begin{proof}
If there is an element $\mathcal{C}$ in $Sol(\mathcal{A})$ for a given starting set 
$\mathcal{A}$ and a given \seHs\!\!, we are able to construct a proof with a 
$\Pi_2$-cut. 
\\
Let 
\begin{align*}
\mathbb{F}:=&\; F[\tuples{x}\backslash U_1], \\
\mathbb{G}:=&\; G[\tuples{y}\backslash U_2], \\
\mathbb{A}:=&\; \forall\tuples{x}.F,\text{ and} \\
\mathbb{B}:=&\; \exists\tuples{y}.G. 
\end{align*}
Assume an \seHs 
\begin{align*}
\mathbb{F},\bigvee\limits_{i=1}^pX[x\backslash\alpha 
,y\backslash t_i]\to 
\bigwedge\limits_{j=1}^mX[x\backslash r_j,y\backslash \beta_j]\seq\mathbb{G}
\end{align*}
and the clause set $\mathcal{C}\in Sol(\mathcal{A})$ for the starting set $\mathcal{A}$. 
Then there are maximal \textbf{G3c}-derivations $\pi_l$ and $\pi_r$ with axioms as 
leaves for the sequents 
\begin{align*}
\mathbb{F}\seq\bigvee\limits_{i=1}^p\lambda xy.(\DNF (\mathcal{C}))
\alpha t_i,\mathbb{G}
\end{align*}
and
\begin{align*}
\mathbb{F},\bigwedge\limits_{j=1}^m\lambda xy.(\DNF (\mathcal{C}))r_j\beta_j 
\seq\mathbb{G},
\end{align*}
respectively. Furthermore, the following proof is valid and contains a single 
$\Pi_2$-cut: 
\begin{center}
\begin{fCenter}
\Axiom$ \;\fCenter \pi_l $
\UnaryInf$\mathbb{A}',\mathbb{F} \;\fCenter\seq 
\bigvee\limits_{i=1}^p\lambda xy.(\DNF (\mathcal{C}))
\alpha t_i,\mathbb{G},\mathbb{B}'$
\UnaryInf$ \;\fCenter\vdots $
\UnaryInf$\mathbb{A} \;\fCenter\seq 
\forall x\exists y.\DNF (\mathcal{C}),\mathbb{B}$
\Axiom$ \;\fCenter \pi_r $
\UnaryInf$\mathbb{A}',\mathbb{F},\bigwedge\limits_{j=1}^m\lambda xy.(\DNF (\mathcal{C}))r_j\beta_j 
\;\fCenter\seq \mathbb{G},\mathbb{B}'$
\UnaryInf$ \;\fCenter\vdots $
\UnaryInf$\mathbb{A},\forall x\exists y.\DNF (\mathcal{C}) 
\;\fCenter\seq \mathbb{B}$
\LeftLabel{\textit{Cut}}
\BinaryInf$\mathbb{A} \;\fCenter\seq \mathbb{B}.$
\DisplayProof
\end{fCenter}
\end{center}
This is guaranteed by the Theorems \ref{theo.soundness} and 
\ref{the:Pi2-cut extended Herbrand sequent} and, hence, solves the main problem of 
our paper.
\end{proof}

\section{$G^*$-Unifiability}\label{sec.Gunify}

In the previous section we developed a method to check whether a given starting set 
contains a solution for an \seHs\!\!. However, we did not explain how such starting sets can be constructed. In this section we present a method that produces a starting set for a given reduced representation of an 
\seHs\!\!. This starting set will contain a solution if there is a so-called {\em balanced solution}. 

To understand the construction of the starting set we have a look on the leaves DNTA$(\psi )$ of a maximal \textbf{G3c}-derivation 
$\psi$ of a given reduced representation $R$. A solution of the corresponding 
\seHs contains for each leaf $S$ in DNTA$(\psi )$ 
at least one literal $L$ with V$(L)\subseteq\{ x,y\}$, that is an element of 
$A(S),B(S),$ or $N(S)$ with the correct substitutions for $x$ and $y$. Hence, the first approach is to collect all literals that can be substituted such that they become at least one element of $A(S),B(S),$ and $N(S)$. Then we consider all possible sets containing a subset of these literals (see the naive starting set in Definition 
\ref{def:naive starting set}).
\begin{definition}
\label{def.interact}
A literal $L$ with V$(L)\subseteq\{ x,y\}$ {\em interacts} with a literal in $A(S),B(S),$ or $N(S)$ if there are substitutions $[x\backslash u,y\backslash v]$ corresponding to the \seHs such that $L[x\backslash u,y\backslash v]$ is an element of $A(S),B(S),$ or $N(S)$. We say $[x\backslash u,y\backslash v]$ corresponds to the \seHs if $u=\alpha\land v=t_i$ for some $i\in\mathbb{N}_p$ or $u=r_j\land v=\beta_j$ for some $j\in\mathbb{N}_m$.
\end{definition}
Let us assume that a literal $L$ of the solution interacts twice with a literal 
in $A(S),B(S),$ or $N(S)$: with a literal $L_\alpha$ when we replace $x$ and $y$ with $\alpha$ and $t_i$, and with a literal $L_\beta$ when we replace $x$ and $y$ with $r_j$ and $\beta_j$ for $i\in\mathbb{N}_p$ and $j\in\mathbb{N}_m$. We call $L_\alpha$ and $L_\beta$ {\em interacting literals}. We will prove in this section that there are replacements making $L_\alpha$ ans $L_\beta$ equal in a way that the resulting literal does not contain the variables $\alpha,\beta_1,\ldots,\beta_m$. The basic idea of the unification method to be defined below is to find all interacting literals and use them for the constructing of the starting set.

In the first step we collect for each $S \in\text{DNTA}(\psi )$ all pairs of literals which are potential candidates for interacting literals.

\begin{definition}[Unification candidates]
\label{def:Unification candidates}
Assume an \seHs with the corresponding reduced 
representation $R$. Let $S,S'\in\text{DNTA}(\psi )$ for a maximal 
\textbf{G3c}-derivation $\psi$ of $R$. Then 
\begin{align*}
\UC (S, S')=&\{ (L,Q) \; |\; 
L\in A(S)\cup N(S)\textbf{ and }Q\in B(S')\cup N(S')\}
\end{align*}
is the set of {\em unification candidates} for the leaves $S$ and $S'$. 
\end{definition}

To be able to unify them we introduce a specific type-0-grammar \cite{chomskyN1956aa}. 

\begin{definition}
\label{def:type-0-grammar}
Let $\Gcal=\langle\tau ,N,\Sigma ,\Pr\rangle$ be a schematic $\Pi_2$-grammar with the 
non-terminals $\tau ,\alpha ,$ $\beta_1,\ldots ,\beta_m$. We define the 
type-0-grammar 
$\Gcal^*=\langle\tau ,N,\Sigma^*,\Pr^*\rangle$ by 
\begin{align*}
& \Sigma^*= \Sigma\cup\{ x,y\} \text{ and } 
\text{Pr}^*=  \Upsilon_1\cup\Upsilon_2\cup\Upsilon_3 \text{ with} \\
\Upsilon_1= & \{ q\; |\; q\text{ is a }F\text{-production or a } 
G\text{-production}\} \\ 
\Upsilon_2= & \{\alpha\to x,r_1\to x,\ldots ,r_m\to x\} ,\text{ and} \\ 
\Upsilon_3= & \{ t_1(\alpha )\to y,\ldots ,t_p(\alpha )\to y,\beta_1\to y,\ldots 
,\beta_m\to y\} . 
\end{align*}
\end{definition}

For the definition of the unification method we need a notion of a derivation 
applied to a literal. A derivation $d$ is a finite number of positions $p_1,\ldots ,p_n$ 
and production rules $\theta_1\to s_1,\ldots ,\theta_n\to s_n$. If we apply 
$d$ to a literal $L$, i.e.\ $L|d$ then we replace sequentially the terms $\theta_i$ with $s_i$ at the positions $p_i$ for $i=1$ until $i=n$.

\begin{definition}[$\Gcal^*$-unifiability]
\label{def:G*-unifiability}
Assume an \seHs with the corresponding reduced representation $R$ and schematic $\Pi_2$-grammar $\Gcal$. Let $S,S' \in\text{DNTA}(\psi )$ for a maximal \textbf{G3c}-derivation $\psi$ of $R$, $(L,Q)\in\UC (S,S')$, and $\Gcal^*=\langle\tau ,N,\Sigma^*,\Pr^*\rangle$ as in Definition \ref{def:type-0-grammar}. \\
We say $(L,Q)$ is $\Gcal^*$-{\em unifiable} if there are derivations $d$ and $b$ in 
$\Gcal^*$ such that $L|d= \overline{Q}|b$ and V$(L|d) \subseteq \{ x,y\}$. Furthermore 
we call $L|d$ the $\Gcal^*$-{\em unified literal} of $(L,Q)$. \\
We call $R$ $\Gcal^*$-unifiable if we find for every $S\in\text{DNTA}(\psi )$ a 
$S'\in\text{DNTA}(\psi )$ such that there is a 
$\Gcal^*$-unifiable unification candidate in $\UC (S,S')$. 

Formally we define the {\em maximal set of $\Gcal^*$-unified literals} as
$$\text{MGUL}(S,S'):= \{ L\; |\; L\text{ is a }\Gcal^*\text{-unified literal of } 
(L_1,L_2)\in\UC (S,S')\} .$$
\end{definition}
In the construction of a starting set for a unifiable reduced representation $R$ we use all possible clauses that consist of $\mathcal{G}^*$-unified literals.
\begin{definition}[Starting set for $\Gcal^*$-unifiable sequents]
\label{def:starting set for G*-unifiable sequents}
Let $R:=F[\tuples{x}\backslash U_1]\seq G[\tuples{x}\backslash U_2]$ be a $\Gcal^*$-unifiable reduced representation $R$ of an \seHs with a corresponding schematic $\Pi_2$-grammar $\Gcal$. Let $\psi$ be a fixed maximal \textbf{G3c}-derivation. For each pair of leaves $S,S'\in \text{DNTA}(\psi )$ we have the maximal set of $\Gcal^*$-unifiable terms $\text{MGUL}(S,S')$. Then the {\em starting set for the $\Gcal^*$-unifiable reduced  representation } $R$ is defined as 
\begin{align*}
& \mathcal{U}:=\{ U\; |\;U\subseteq\bigcup\limits_{S,S'\in 
\text{DNTA}(\psi )} 
\text{MGUL}(S,S') \} .
\end{align*}
\end{definition}
\begin{lemma}
Let $R$ be the reduced representation of a given \seHs with grammar $\Gcal$, $n$ be the number of atoms occurring in $R$, and $k$ be the length of an encoding or $R$. Let $m$ and $p$ be the numbers given by the to $\Gcal$ corresponding grammar $\Gcal^*$ (see Definition \ref{def:starting set for G*-unifiable sequents}). Then the starting set for $\Gcal^*$-unifiable sequents $\mathcal{U}$ can be constructed in polynomial time $\mathcal{O}(n^2\cdot k\cdot (p+m))$.
\end{lemma}
\begin{proof}
Note that the set of pairs we can build by concatenating two atoms of $R$ is a superset of the set of all unification candidates. The size of this set is $n^2$. For each pair we have to compare at most $k$ symbols in order to unify them. The unification itself compares two symbols with each other or checks whether the symbols can be replaced simultaneously with $x$ (there are $2\cdot (m+1)$ cases) or $y$ (there are $2\cdot (p+m)$ cases). Altogether, there exists a constant $c$ such that $c\cdot (n^2\cdot k\cdot (p+m))$ is an upper bound to the number of operations to construct the starting set for $\Gcal^*$-unifiable sequents $\mathcal{U}$.
\end{proof}
The starting set for $\Gcal^*$-unifiable sequents suffices to find 
{\em balanced solutions}. 

\begin{definition}[Balanced solution]
\label{def:Balanced solution}
Let 
\begin{align*}
& S(X):=F[\tuples{x}\backslash U_1],\bigvee\limits_{i=1}^pX\alpha t_i\to 
\bigwedge\limits_{j=1}^mXr_j\beta_j\seq G[\tuples{y}\backslash U_2]
\end{align*}
be an \seHs\!\!, $\mathcal{C}$ a finite set of sets of 
literals not containing the variables $\alpha,\beta_1,\ldots,\beta_m$, and $\hat{E}:=\lambda xy.\DNF (\mathcal{C})$ such that 
\begin{align*}
& S(\hat{E}):=F[\tuples{x}\backslash U_1],\bigvee\limits_{i=1}^p\hat{E} 
\alpha t_i\to\bigwedge\limits_{j=1}^m\hat{E}r_j\beta_j\seq 
G[\tuples{y}\backslash U_2]
\end{align*}
is a tautology. Let $\psi$ be a maximal \textbf{G3c}-derivation of $S(\hat{E})$. 
We say $\mathcal{C}$ is a {\em balanced solution} if in all axioms of 
$S(\hat{E})$ at least 
one of the active formulas is not an ancestor of $\hat{E}$ in $\psi$. 
\end{definition}

A balanced solution does not contain interactive literals (not to be confused with 
interacting literals) as described in Definition \ref{def:Set of possible sets of clauses} by $T_3$ and in Definition \ref{def:solution set} by $T_3'$. 

\begin{theorem}
\label{the:balanced solution implies G*-unifiable solution}
Let $S$ be $\forall\tuples{x}.F\seq\exists\tuples{y}.G$, $\mathcal{G}$ be a schematic 
$\Pi_2$-grammar, and 
\begin{align*}
& S(X):=F[\tuples{x}\backslash U_1],\bigvee\limits_{i=1}^pX\alpha t_i\to 
\bigwedge\limits_{j=1}^mXr_j\beta_j\seq G[\tuples{y}\backslash U_2]
\end{align*}
be an \seHs for $S$ and $\mathcal{G}$. Assume that $S(X)$ has a balanced solution 
$\mathcal{C}$. Then the set of solution candidates $Sol(\mathcal{U})$ (defined as in Definition \ref{def:solution set}) is not empty where $\mathcal{U}$ is the 
starting set for the $\Gcal^*$-unifiable sequent $R$ as in Definition \ref{def:starting set for G*-unifiable sequents}. 
\end{theorem}

To prove the theorem we show the same result for the naive starting set instead 
of the starting set for $\Gcal^*$-unifiable sequents $\mathcal{U}$ and conclude that 
$Sol(\mathcal{U})$ is also not empty. 

\begin{definition}[Naive starting set]
\label{def:naive starting set}
Let $R$ be a reduced representation and $\psi$ a maximal \textbf{G3c}-derivation of 
$R$. We define for each leaf $S\in\text{DNTA}(\psi )$ of the form $A(S)\circ B(S)\circ 
N(S)$ the sets 
\begin{align*}
\NA (S) = &\; \{L \mid \ex j \in \mathbb{N}_p. (\lambda x y.L)\alpha t_j \in A(S) 
\cup N(S),V(L) \subseteq \{x,y\} \}.\\
\NB (S) = &\; \{L \mid \ex j \in \mathbb{N}_m. (\lambda x y.L)r_j \beta_j \in 
\overline{B(S)}\cup\overline{N(S)} 
,V(L) \subseteq \{x,y\} \}.
\end{align*}
\begin{align*}
\mathcal{N}:= &\; \{ N\; |\; N\subseteq\bigcup\limits_{S\in\text{DNTA}(\psi )} 
\NA (S)\cup \NB (S)\}
\end{align*}
is then called the {\em naive starting set}. 
\end{definition}

\begin{corollary}
\label{cor:naive solution}
Let $S$ be $\forall\tuples{x}.F\seq\exists\tuples{y}.G$, $\mathcal{G}$ be a schematic 
$\Pi_2$-grammar, and 
\begin{align*}
S(X):= &\; F[\tuples{x}\backslash U_1],\bigvee\limits_{i=1}^pX\alpha t_i\to 
\bigwedge\limits_{j=1}^mXr_j\beta_j\seq G[\tuples{y}\backslash U_2]
\end{align*}
be an \seHs for $S$ and $\mathcal{G}$. Assume there is a balanced solution 
$\mathcal{C}$. Then $\mathcal{C}\in Sol(\mathcal{N})$ where $\mathcal{N}$ is the naive 
starting set and $Sol()$ is defined as in Definition \ref{def:solution set}. 
\end{corollary}

\begin{proof}
The Definition \ref{def:Balanced solution} of a balanced solution implies that every 
literal $L$ of the balanced solution $\mathcal{C}$ is either an element of 
$N(S)\cup\overline{N(S)}$ for a leaf $S\in\text{DNTA}(\psi )$ of the maximal 
\textbf{G3c}-derivation $\psi$ of the \seHs or it is an element of the sets $\NA (S)$ and $\NB (S)$. For a literal $L$ of $N(S)\cup\overline{N(S)}$, we can define $\lambda x,y.L$ even though $L$ is variable-free. Hence, $L$ is an element of $\NA (S)$ or $\NB (S)$. By Theorem \ref{the:Partial completeness}, $\mathcal{C}\in Sol(\mathcal{N})$. 
\end{proof}

Given a solution which is a subset of the naive starting 
set, we define a new solution that is a subset of the starting set for $\Gcal^*$-unifiable sequents. 

\begin{lemma}
\label{lem:N to U}
Assume that $Sol(\mathcal{N})$ contain a balanced solution for a given \mbox{\seHs\!\!,} 
for a maximal \textbf{G3c}-derivation $\psi$ of its reduced representation $R$, and for 
the naive starting set $\mathcal{N}$. 
Let $\Gcal$ be the corresponding schematic $\Pi_2$-grammar. 
Then $Sol(\mathcal{U})\neq\emptyset$ for the starting set for $\Gcal^*$-unifiable 
sequents $\mathcal{U}$. 
\end{lemma}

\begin{proof}
Let $\mathcal{C}\in Sol(\mathcal{N})$ be a balanced solution. We choose an arbitrary literal $L$ of $\mathcal{C}$ that is not an element of any set of literals in $\mathcal{U}$. If there are none, all literals of $\mathcal{C}$ occur in $\mathcal{U}$. Since we consider in $\mathcal{U}$ all possible sets with a finite number of literals, $\mathcal{C}$ is an element of $\mathcal{U}$, $Sol(\mathcal{U})\neq\emptyset ,$ and we are done. Otherwise we distinguish between two cases 
\begin{align*}
L\in &\; \bigcup\limits_{S\in\text{DNTA}(\psi )}\NA (S) 
\text{ and }\{ L\}\notin\mathcal{U} &\text{(I)} \\ 
L\in &\; \bigcup\limits_{S\in\text{DNTA}(\psi )}\NB (S) 
\text{ and }\{ L\}\notin\mathcal{U}. & \text{(II)}
\end{align*}
First we consider (I). In this case there is a leaf $S\in 
\text{DNTA}(\psi )$ and there is a $j\in\{ 1,\ldots ,p\}$ such that 
\begin{align*}
(\lambda xy.L)\alpha t_j\in A(S)\cup N(S). 
\end{align*}
By $\{ L\}\!\notin\mathcal{U}$ there is no leaf $S'$ such that $Q\!\in B(S')
\cup N(S')$ where $((\lambda xy.L)\alpha t_j,Q)$ is $\Gcal^*$-unifiable 
with the $\mathcal{G}^*$-unifiable literal $L$. If $C \in \mathcal{C}$ and $C=\{ L\}$ is a unit clause then the sequent 
\begin{align*}
((\lambda xy.L)r_1\beta_1,\ldots ,(\lambda xy.L)r_m\beta_m) \circ S &\; 
\end{align*}
is not a tautology and $\mathcal{C}$ is not a solution. Thus if $C \in \mathcal{C}$ and $C$ contains $L$ it cannot be a unit clause. So we define the new clause $C' = C \setminus \{L\}$ and we know that $C'$ is not empty. A maximal \textbf{G3c}-derivation of the sequent $ J\circ (\;\seq \mathcal{C}')$\footnote{For a set of clauses $\mathcal{C}$ and a sequent $J$, we abbreviate $ J\circ (\;\seq \DNF (\mathcal{C}))$ with $ J\circ (\;\seq \mathcal{C})$} where $\mathcal{C}' = (\mathcal{C} \setminus \{C\}) \cup \{C'\}$ and $J$ is an arbitrary element of DNTA$(\psi )$ contains only axioms that also appear in $J\circ (\; \seq \mathcal{C}).$ Hence, the new sequent is a tautology, too. 

Now we consider the sequent $(\mathcal{C'}\seq )\circ J$ for an arbitrary $J\in{DNTA}(\psi )$. If it were not a tautology there would be a leaf $S'\in\text{DNTA}(\psi )$ and an $i\in\mathbb{N}_m$ such that 
\begin{align*}
\overline{(\lambda xy.L)r_i\beta_i}\in B(S')\cup N(S'). 
\end{align*}
(Note that we need here, that the given solution is a balanced 
solution. Otherwise we would have to consider the case that $(\lambda xy.L)r_i\beta_i$ 
appears in $\mathcal{C}'$, too). 
But then there exists the $\Gcal^*$-unifiable pair 
\begin{align*}
((\lambda xy.L)\alpha t_j, 
\overline{(\lambda xy.L)r_i\beta_i}) 
\end{align*}
and $L \in \mathcal{U}$ contradicting our assumption; we conclude that $(\mathcal{C'}\seq )\circ J$ is a tautology.

With this procedure we can erase all literals of $\mathcal{C}$ that are elements of 
\begin{align*}
\bigcup\limits_{S\in\text{DNTA}(\psi )}\NA (S) 
\end{align*}
and do not appear in a clause of $\mathcal{U}$. 

\smallskip

Now let us consider (II). In this case there is a leaf $S\in\text{DNTA}(\psi )$ and there is a $j\in\mathbb{N}_m$ such that $(\lambda xy.L)r_j\beta_j\in\overline{B(S)}\cup\overline{N(S)}$. Given that $\{ L\}\notin\mathcal{U}$ there is no leaf $S'$ such that $Q\in A(S')\cup N(S')$ where $(Q,\overline{(\lambda xy.L)r_j\beta_j} )$ is $\mathcal{G}^*$-unifiable with the $\mathcal{G}^*$-unifiable literal $L$. Let $C$ be a clause containing $L$. Assume $C$ is the only clause then $\mathcal{C}$ is not a solution because $S\circ (\seq \mathcal{C})$ contains the branch
\begin{align*}
S\circ ( \;\seq (\lambda xy.L)\alpha t_1,\ldots ,(\lambda xy.L)\alpha t_p)
\end{align*}
which is not a tautology. Therefore $\mathcal{C}$ does not only contain the clause $C$ and we can define $\mathcal{C}' = \mathcal{C} \setminus \{C\}$; moreover we know that $\mathcal{C} \neq \emptyset$. Since $\mathcal{C}$ contains more than one clause $\mathcal{C}'$ is not empty. A maximal \textbf{G3c}-derivation of the sequent $ J\circ (\mathcal{C}'\seq )$ where $J$ is an arbitrary element of DNTA$(\psi )$ contains only axioms that also appear in $J\circ (\mathcal{C}\seq ).$ Hence, the new sequent is a tautology, too. 

Now we consider the sequent $J\circ (\;\seq \mathcal{C}')$ for an arbitrary $J\in\text{DNTA}(\psi )$. If it were not a tautology there would 
be a leaf $S'\in\text{DNTA}(\psi )$ and an $i\in\mathbb{N}_p$ such that 
$(\lambda xy.L)\alpha t_i\in 
A(S')\cup N(S')$. But then there exists the $\Gcal^*$-unifiable pair 
\begin{align*}
((\lambda xy.L)\alpha t_j, 
\overline{(\lambda xy.L)r_i\beta_i}). 
\end{align*}
So we obtain $L \in \mathcal{U}$ contradicting our assumption; again we conclude that   $J\circ (\; \seq \mathcal{C'})$ is a tautology. 

With this procedure we can erase all literals of $\mathcal{C}$ that are elements of 
\begin{align*}
\bigcup\limits_{S\in\text{DNTA}(\psi )}\NB (S) 
\end{align*}
and do not appear in a clause of $\mathcal{U}$.

\smallskip

By an exhaustive application of these two methods we 
get a solution that is a subset of $\mathcal{U}$. 
\end{proof}

\begin{proof}[Proof of Theorem \ref{the:balanced solution implies G*-unifiable solution}]
The proof can be obtained by combining Corollary \ref{cor:naive solution} and Lemma \ref{lem:N to U}. 
\end{proof}

\section{Generalizing the Cut Formulas}
\label{sec.tuples}
In the previous sections we considered (for the sake of simplicity)  only cut formulas of the form $\forall x\exists y.A(x,y)$ for single variables $x,y$. The purpose of this section is to generalize the approach to the construction of cut formulas of the form $\forall \vec{x}\exists \vec{y}.A(\vec{x},\vec{y})$ for variable tuples $\vec{x},\vec{y}$. Most definitions and proofs remain almost unchanged by replacing terms by tuples of terms. We indicate the changes in the most important definitions  and theorems and reformulate the crucial definitions of the paper.

We denote by $\lengthOfTuple{\tuples{t}}$ the arity of $\tuples{t}$. For simplicity we try to keep the arity implicit, i.e.\ we avoid to use to many indexes. Hence, by describing a substitution with $[\tuples{x}\backslash\tuples{t}]$, we assume that $\tuples{x}$ is a ordered set of variables of the same arity as $\tuples{t}$. The substitution is then given by $[x_1\backslash t_1],\ldots ,[x_{\lengthOfTuple{\tuples{x}}}\backslash t_{\lengthOfTuple{\tuples{t}}}]$. To extend the notion of grammars, we have to allow production rules to handle tuples. A production rule of the form $\tuples{\alpha }\to\tuples{t}$ applied to a term $s$ is the replacement of the non-terminals $\alpha_1,\ldots ,\alpha_{\lengthOfTuple{\tuples{\alpha}}}$ according to $[\tuples{\alpha}\backslash\tuples{t}]$, i.e.\ we substitute $t_i$ for $\alpha_i$ at a designated position. We say the tuples of non-terminals $\tuples{\alpha},\tuples{\beta}$ (not necessarily of the same length) are ordered with respect to $<$, i.e.\ $\tuples{\alpha}<\tuples{\beta}$, if for all production rules $\tuples{\alpha}\to\tuples{t}$ the terms of the tuples $\tuples{t}$ do not contain any non-terminal of $\tuples{\beta}$. The extended definition of schematic $\Pi_2$-grammars reads as follows:
\begin{definition}[Schematic $\Pi_2$-grammar with tuples]\label{def.schem-pi2-grammar tuples}
Let $\Gcal = \langle\tau ,N,\Sigma ,Pr\rangle$ be an acyclic tree grammar and $N = 
\{\tau ,\tuples{\alpha} ,\tuples{\beta}_1,\ldots ,\tuples{\beta}_m\}$. Let the variables (tuples of variables) be ordered according to 
\begin{align*}
\tuples{\beta}_1<\ldots  <\tuples{\beta}_m<\tuples{\alpha}<\tau 
\end{align*} 
and $\lengthOfTuple{\tuples{\beta}_i}=\lengthOfTuple{\tuples{\beta}_j}$ for $i,j\in\mathbb{N}_m$.
We call $\Gcal$ a {\em schematic $\Pi_2$-grammar} if the production rules are of the following form: 
\begin{align*}
\tau \to &\; s_1|\ldots |s_c|w_1|\ldots |w_d\ \text{ with }V(s_i)\subseteq V(\tuples{\alpha}) 
\text{ for }1\le i\le c
\text{ and }\\
&\; V(w_j)\subseteq V(\tuples{\beta}_1)\cup\ldots\cup V(\tuples{\beta}_m)\text{ for } 1\le j\le d \\ 
\tuples{\alpha} \to &\;  \tuples{r}_1|\ldots |\tuples{r}_m \\
&\; \text{ with }V(\tuples{r}_j)\subseteq V(\tuples{\beta}_1)\cup\ldots\cup V(\tuples{\beta}_j)\text{ for }2\le j\le m\text{ and }V(\tuples{r}_1)=\emptyset \\
\tuples{\beta}_j \to &\; \tuples{t}_1\tuples{r}_j|\ldots |\tuples{t}_p\tuples{r}_j \text{ for }1\le j\le m.
\end{align*}
By $\tuples{t}\tuples{s}$ we denote the tuple of terms $(\tuples{t}|_1[\tuples{\alpha}\backslash\tuples{s}],\ldots,\tuples{t}|_n[\tuples{\alpha}\backslash\tuples{s}])$ where $n=\lengthOfTuple{\tuples{t}}$ and $\tuples{t}|_k$ for $1\le k\le n$ being a term (possibly) containing variables of $\tuples{\alpha}$. We call $m$ the \mbox{\emph{\AllMul}\!\!,} $p$ the \mbox{\emph{\ExMul}}and denote $\lengthOfTuple{\tuples{\alpha}}$ by $q_\forall$ and $\lengthOfTuple{\tuples{\beta}_1}$ by $q_\exists$.
\end{definition}
\begin{example}
\label{exa:schematifPi2grammarTuples}
Let $c,d,e$ be constants, $f,g,h$ unary functions, $h_F$ a function with arity six, $h_G$ a function with arity four, and $\tuples{\alpha}=(\alpha_1,\alpha_2),\tuples{\beta}=(\beta_1,\beta_2),\tuples{\gamma}=(\gamma_1,\gamma_2)$. The grammar $\mathcal{G}=\langle \tau ,N,\Sigma ,\Pr\rangle$ with $\tau$ being the designated starting symbol, $N=\{\tau ,\tuples{\alpha},\tuples{\beta},\tuples{\gamma}\}$, $\lengthOfTuple{\tuples{\alpha}}=\lengthOfTuple{\tuples{\beta}}=\lengthOfTuple{\tuples{\gamma}}=2$, and $\Pr =$
\begin{align*}
\{\; \tau \to &\; h_F(\alpha_1,\alpha_1,\alpha_2,\alpha_2,\alpha_2,\alpha_2)|h_G(\beta_1,\beta_2,\gamma_1,\gamma_2) \\
\tuples{\alpha}\to &\; (c,d)|(e,e) \\
\tuples{\gamma}\to &\; (fe,ge)|(fe,he) \\
\tuples{\beta}\to &\; (fc,fd)|(fc,hd) \;\;\}
\end{align*}
with 
\begin{align*}
\tuples{\beta}<\tuples{\gamma}<\tuples{\alpha}<\tau
\end{align*}
is a schematic $\Pi_2$-grammar with tuples. The language consists of the words
\begin{align*}
h_F(c,c,d,d,d,d),h_F(e,e,e,e,e,e) \\
h_G(fc,fd,fe,ge),h_G(fc,fd,fe,he) \\
h_G(fc,hd,fe,ge),h_G(fc,hd,fe,he).
\end{align*}
A corresponding extended Herbrand sequent with tuples can be extracted from the proof 
\begin{center}
\begin{fCenter}
\Axiom$ \fCenter\vdots $
\UnaryInf$\Gamma \;\fCenter\seq \Delta ,\forall x_1,x_2\exists y_1,y_2.C$
\Axiom$ \fCenter\vdots $
\UnaryInf$\forall x_1,x_2\exists y_1,y_2.C,\Gamma \;\fCenter\seq \Delta$
\LeftLabel{\textit{Cut}}
\BinaryInf$\Gamma \;\fCenter\seq \Delta$
\DisplayProof
\end{fCenter}
\end{center}
where
\begin{align*}
\Gamma := &\; \forall x,y.P(x,fx)\lor Q(y,gy)\lor Q(y,hy) \\
\Delta := &\; \exists u,v,w,x.(P(c,u)\lor Q(d,v))\land (P(e,w)\lor Q(e,x)) \\
C := &\; P(x_1,y_1)\lor Q(x_2,y_2),
\end{align*}
and
\begin{align*}
&P(\alpha_1,f\alpha_1)\lor Q(\alpha_2,g\alpha_2), P(\alpha_1,f\alpha_1)\lor Q(\alpha_2,h\alpha_2), \\
&P(c,\beta_1)\lor Q(d,\beta_2),P(e,\gamma_1)\lor Q(e,\gamma_2) 
\end{align*}
are all quantifier-free instantiations of the cut formula in the proof,
\begin{align*}
&P(\alpha_1,f\alpha_1)\lor Q(\alpha_2,g\alpha_2)\lor Q(\alpha_2,h\alpha_2), \\
&(P(c,\beta_1)\lor Q(d,\beta_2))\land (P(e,\gamma_1)\lor Q(e,\gamma_2)) 
\end{align*}
are all quantifier-free instantiations of the context $\Gamma\seq\Delta$ in the proof.
\end{example}
Then the schematic extended Herbrand sequent for $\Pi_2$-cuts with blocks of quantifiers $S(X)$ has to be solved, i.e.\ we have to find a substitution $C$ for $X$ that makes $S(C)$ a tautology.
\begin{definition}[Schematic extended Herbrand sequent with blocks of quantifiers] 
\label{def:Schematic extended Herbrand sequent tuples}
Let $S\colon$ $\forall\tuples{x}F\seq\exists\tuples{y}G$ be a provable sequent and $H_s(S)$ be a Herbrand term set of $S$. Let  $\Gcal\colon \langle\tau ,N,\Sigma ,\Pr\rangle$ be a schematic $\Pi_2$-grammar as in Definition \ref{def.schem-pi2-grammar tuples} with the non-terminals $N=\{\tau ,\tuples{\alpha} ,\tuples{\beta}_1, 
\ldots ,\tuples{\beta}_m\}$, the order 
\begin{align*}
\tuples{\beta}_1<\ldots  <\tuples{\beta}_m<\tuples{\alpha}<\tau .
\end{align*} 
and the production rules 
\begin{align*}
\tau \to &\; h_F\tuples{u}_1|\ldots |h_F\tuples{u}_c\ |  \  h_G\tuples{v}_1|\ldots |h_G\tuples{v}_d \text{ with }V(\tuples{u}_i)\subseteq\tuples{\alpha}\text{ for }1\le i\le c \\ 
&\; \text{ and }V(\tuples{v}_j)\subseteq\tuples{\beta}_1\cup\ldots\cup\tuples{\beta}_m\text{ for }1\le j\le d \\ 
\tuples{\alpha} \to &\;  \tuples{r}_1|\ldots |\tuples{r}_m \text{ with }V(\tuples{r}_j)\subseteq\tuples{\beta}_1\cup\ldots\cup\tuples{\beta}_j\text{ for }2\le j\le m\text{ and }V(\tuples{r}_1)=\emptyset \\
\tuples{\beta}_j \to &\; \tuples{t}_1\tuples{r}_j|\ldots |\tuples{t}_p\tuples{r}_j \text{ for }1\le j\le m.
\end{align*}
Let $L(\Gcal)$ be the language of $\Gcal$ generated only by rigid derivations with respect to all non-terminals, and 
$H_s(S) \subseteq L(\Gcal)$. Let $U_1:=\{\tuples{u}_1,\ldots ,\tuples{u}_c\}$ and $U_2:=\{\tuples{v}_1,\ldots ,\tuples{v}_d\}$. Then 
we call the sequent 
\begin{align*}
S(X)\colon F[\tuples{x}\backslash U_1],\bigvee\limits_{i=1}^pX\tuples{\alpha }\tuples{t}_i\to \bigwedge\limits_{j=1}^mX\tuples{r}_j\tuples{\beta}_j\seq 
G[\tuples{y}\backslash U_2],
\end{align*}
where $X$ is a $(q_\forall +q_\exists )$-place predicate variable, a {\em schematic extended Herbrand 
sequent} corresponding to $\Gcal$ and $S$ (in the following abbreviated by {\em SEHS}). 
Furthermore, we call $F[\tuples{x}\backslash U_1]\seq G[\tuples{y}\backslash U_2]$ the 
{\em reduced representation} of $S(X)$. 
\end{definition}
\begin{example}
\label{exa:schematicExtendedTuples}
Let $\mathcal{G}$ be as in Example \ref{exa:schematifPi2grammarTuples}. Then we can define the schematic extended Herbrand sequent with blocks of quantifiers 
\begin{align*}
&P(\alpha_1,f\alpha_1)\lor Q(\alpha_2,g\alpha_2)\lor Q(\alpha_2,h\alpha_2), \\
&(X\alpha_1\alpha_2f\alpha_1g\alpha_2\lor X\alpha_1\alpha_2f\alpha_1h\alpha_2) \to (Xcd\beta_1\beta_2\land Xee\gamma_1\gamma_2) \\
& \seq (P(c,\beta_1)\lor Q(d,\beta_2))\land (P(e,\gamma_1)\lor Q(e,\gamma_2)) .
\end{align*}
The corresponding end sequent is $\Gamma\seq\Delta$ where $\Gamma$ and $\Delta$ is as in Example \ref{exa:schematifPi2grammarTuples}.
\end{example}
\begin{definition}
\label{def:solution of an sehs tuples}
Let $S$ be a provable sequent, $\mathcal{G}$ a schematic $\Pi_2$-grammar with the non-terminals $\{\tau ,\tuples{\alpha} ,\tuples{\beta}_1,\ldots ,\tuples{\beta}_m\}$ as in Definition \ref{def.schem-pi2-grammar tuples}, and $S(X)$ the corresponding \seHs\!\!. Let $S(X)[X \setminus \lambda \tuples{x} \tuples{y}.A]$ be a tautology where $A$ may not contain any variable in $\tuples{\alpha}$ or $\tuples{\beta}_j$ with $j\in\mathbb{N}_m$. Then we call $A$ a \emph{solution} of the \seHs\!\! $S(X)$.
\end{definition}
\begin{example}
A solution of the \seHs of Example \ref{exa:schematicExtendedTuples} is $P(x_1,y_1)\lor Q(x_2,y_2)$.
\end{example}
Note that this is a generalization of the previous sections. For $q_\forall =1$ and $q_\exists =1$, the generalization tallies with the already described method. Also the rest of the procedure has only to be adjusted to operate with tuples of variables. The \emph{starting set for $\Pi_2$-cuts with blocks of quantifiers} now contains the designated variables $x_1,\ldots ,x_k$ and $y_1,\ldots ,y_l$ where $k=q_\forall ,l=q_\exists $ and may not contain variables of $\tuples{\alpha}$ or $\tuples{\beta}_j$ with $j\in\mathbb{N}_m$ (compare to Definition \ref{def:starting set}). The set of \emph{possible sets of clauses with tuples of variables}, the set of \emph{refined allowed clauses with tuples of variables}, and the set of \emph{solution candidates for $\Pi_2$-cuts with blocks of quantifiers} can be defined accordingly. 

The main theorem for the characterization generalizes to the case of blocks of quantifiers.

\begin{theorem}\label{the.main tuples}
Let 
\begin{align*}
& F[\tuples{x}\backslash U_1],\bigvee\limits_{i=1}^pX\tuples{\alpha} \tuples{t}_i\to \bigwedge\limits_{j=1}^mX\tuples{r}_j\tuples{\beta}_j\seq 
G[\tuples{y}\backslash U_2]
\end{align*}
be an \seHs \!\! corresponding to a Herbrand sequent of a cut-free proof of $S$ and a grammar $\mathcal{G}$ covering the Herbrand term set of $S$. Let $Sol(\mathcal{A}) \neq\emptyset$ be the set of solution candidates for $\Pi_2$-cuts with blocks of quantifiers for a given starting set for $\Pi_2$-cuts with blocks of quantifiers $\mathcal{A}$, and $\Ccal \in Sol(\mathcal{A})$. Let $E = \DNF(\Ccal)$ be the formula in DNF corresponding to $\Ccal$ and $V(E) \subseteq \{\tuples{x},\tuples{y}\}$ where $|\tuples{x}|=q_\forall $ and $|\tuples{y}|=q_\exists$.  Then there exists a proof of $S$ with one cut and the cut formula $\forall \tuples{x} \exists \tuples{y}.E$ 
\end{theorem}
Furthermore, we can easily adjust the $\Gcal^*$unification method by replacing the production rules of Definition \ref{def:type-0-grammar} with 
\begin{align*}
\text{Pr}^*= & \Upsilon_1\cup\Upsilon_2\cup\Upsilon_3 \text{ with} \\
\Upsilon_1= & \{ q\; |\; q\text{ is a }F\text{-production or a } 
G\text{-production}\} \\ 
\Upsilon_2= & \{\tuples{\alpha}\to \tuples{x},\tuples{r}_1\to \tuples{x},\ldots ,\tuples{r}_m\to \tuples{x}\} ,\text{ and} \\ 
\Upsilon_3= & \{ \tuples{t}_1(\tuples{\alpha} )\to \tuples{y},\ldots ,\tuples{t}_p(\tuples{\alpha} )\to \tuples{y},\tuples{\beta}_1\to \tuples{y},\ldots ,\tuples{\beta}_m\to \tuples{y}\} . 
\end{align*}
Then we can prove the non-emptiness of the set of solution candidates for $\Pi_2$-cuts with blocks of quantifiers by assuming the existence of a balanced solution as in the simplified case.
\begin{theorem}
\label{the:balanced solution implies G*-unifiable solution with tuples}
Let $S$ be $\forall\tuples{x}.F\seq\exists\tuples{y}.G$, $\mathcal{G}$ be a schematic 
$\Pi_2$-grammar, and 
\begin{align*}
& S(X):=F[\tuples{x}\backslash U_1],\bigvee\limits_{i=1}^pX\tuples{\alpha} \tuples{t}_i\to \bigwedge\limits_{j=1}^mX\tuples{r}_j\tuples{\beta}_j\seq G[\tuples{y}\backslash U_2]
\end{align*}
be an \seHs for $S$ and $\mathcal{G}$. Assume that $S(X)$ has a balanced solution $\mathcal{C}$. Then the set of solution candidates for $\Pi_2$-cuts with blocks of quantifiers $Sol(\mathcal{U})$ is not empty where $\mathcal{U}$ is the starting set for the $\Gcal^*$-unifiable sequent $R$ (where we consider production rules for tuples). 
\end{theorem}

\section{Proof Compression}
\label{sec:app example}
In Section~\ref{sec.Gunify} we have defined a method to find balanced 
solutions for \seHs \!\!.  Here we demonstrate their potential of proof compression 
via $\Pi_2$ cuts.  Again we consider the example from Section~\ref{sec:Example}. 
\\
At the beginning we prove each sequent of the sequence by constructing a Herbrand 
sequent. Afterwards we measure the complexity in three different ways. We either count  
the number of weak quantifier inferences (\emph{quantifier complexity}), the number of 
inferences (\emph{logical complexity}), or the number of symbols (\emph{symbol 
complexity}). We know that for instance a compression in terms of weak quantifier 
inferences can be easily achieved by increasing the logical complexity or the symbol 
complexity of the cut-formula. By measuring all of them we ensure that the compression 
we achieve is largely independent of the measurement. 
\\
In the end we will see that by the method of $G^*$-unifiability we find for all 
sequents of the sequence $S_n$ defined in Section~\ref{sec:Example} proofs $\psi_n$ 
with the cut-formula $\forall x\exists y.P(x,fy)$ which are polynomially  bounded in 
$n$. We also show that {\em all} sequences of cut-free proofs of $S_n$ grow 
exponentially in $n$, which yields an exponential compression of proof complexity.
\\
We already defined the quantifier complexity (see Definition~\ref{def.quantifiercomp}). 
\begin{definition}[Logical complexity]
Let $\pi$ be a given {\LK} proof. If $\pi$ is of the form 
\begin{center}
\begin{fCenter}
\Axiom$ \;\fCenter $
\LeftLabel{Axiom}
\UnaryInf$\Delta \;\fCenter\seq \Gamma$
\DisplayProof
\end{fCenter}
\end{center}
then the \emph{logical complexity} $|\pi |_l$ is defined to be $0$. If $\pi$ is of 
the form 
\begin{center}
\begin{fCenter}
\Axiom$ \;\fCenter \pi_l$
\Axiom$ \;\fCenter \pi_r$
\LeftLabel{Binary rule}
\BinaryInf$\Delta \;\fCenter\seq \Gamma$
\DisplayProof
\end{fCenter}
\end{center}
with an arbitrary binary rule of {\LK} subproofs $\pi_l$ and $\pi_r$ then 
$|\pi |_l:=|\pi_l|_l+|\pi_r|_l+1$. If $\pi$ is of the form 
\begin{center}
\begin{fCenter}
\Axiom$ \;\fCenter \pi '$
\LeftLabel{Unary rule}
\UnaryInf$\Delta \;\fCenter\seq \Gamma$
\DisplayProof
\end{fCenter}
\end{center}
with the subproof $\pi '$ and an arbitrary unary rule of {\LK} then 
$|\pi |_l:=|\pi '|_l+1$.
\end{definition}
The symbol complexity counts the number of symbols in each sequent of the proof and the 
number of rules that connect these sequents with each other. Therefore, it can be 
defined with the help of the logical complexity which represents the number of 
{\LK}-rules. 
\begin{definition}[Symbol complexity]
Let $\pi$ be a given {\LK} proof and $\Sigma$ the corresponding signature. Let 
$S_1,\ldots ,S_n$ be the sequents occurring in $\pi$. The \emph{symbol complexity} 
$|S_i|_s$ of a sequent $S_i$ for $i\in\mathbb{N}_n$ is equal to the number of 
occurrences of the symbols of the set 
$\Sigma\cup\{$`$\lor$' , `$\land$' , `$\to$' , `$\neg$' , `$\exists$' , `$\forall 
$' , `$,$' , `$\seq$'$\}$ and of variables occurring in $S_i$. The symbol complexity 
$|\pi |_s$ of the proof is defined as 
\begin{align*}
|\pi |_s:=|\pi |_l + \sum\limits_{i\in\mathbb{N}_n}S_i.
\end{align*}
\end{definition}
It is easy to see that the different measurements follow an order. While the quantifier 
complexity is the most coarse one, the symbol complexity is the finest. 
\begin{proposition}
\label{pro:orderCompl}
Let $\pi$ be a given {\LK} proof. Then the following inequalities hold:
\begin{align*}
|\pi |_q\le |\pi |_l\le |\pi |_s . 
\end{align*}
\end{proposition}
\begin{proof}
The claim trivially holds.
\end{proof}
Before we start to compute the different complexities of our example, we adjust the form of the end-sequents $A_n,B,C_n\seq D$. In the presented method we require a sequent of the form $\forall\tuples{x}F\seq\exists \tuples{y}G$. But as already mentioned we can transform each sequent into this format. Let $A_n',B',C_n',$ and $D'$ be the quantifier free part of $A_n,B,C_n,$ and $D$ (we rename the variables) 
\begin{align*}
A_n':= &\; P(x_1,f_1x_1)\lor\ldots\lor P(x_1,f_nx_1), \\
B':= &\; P(x_2,x_3)\to P(x_2,fx_3), \\
C_n':= &\; P(y_1,fy_2)\land P(fy_2,fy_3)\land\ldots\land P(fy_{n-1},fy_n) 
\to P(y_1,gy_n),\text{ and} \\
D':= &\; P(y_{n+1},gy_{n+2}). \\
\end{align*}
Furthermore, let $\tuples{x}=(x_1,x_2,x_3)$ be the tuple of the $3$ variables occurring in $A_n'\land B'$, and $\tuples{y}=(y_1,\ldots ,y_{n+2})$ be 
tuples  of the $n+2$ variables occurring in $\overline{C_n'}\lor D'$;  let $\overline{C_n'}$ be the negation of $C_n'$
\begin{align*}
\overline{C_n'}:= &\; P(y_1,fy_2)\land P(fy_2,fy_3)\land\ldots\land 
P(fy_{n-1},fy_n) \land\neg P(y_1,gy_n).
\end{align*}
Then we can define 
equivalent sequents 
\begin{align*}
S_n':= &\; \forall \tuples{x}.A_n'\land B'\seq \exists\tuples{y}.
\overline{C_n'}\lor D' .
\end{align*}
From now on $S_n'$ will always refer to the rewritten sequence of sequents that is in 
the correct form for the presented cut-introduction method. $S_n$ will refer to the 
original version.

\subsection{Minimal Cut-Free Proofs}
\label{sub:Minimal cut-free proofs}

In this section we consider cut-free proofs of $S_n$ for a fixed natural number $n$. 
For convenience we will compute lower bounds on the different complexities of 
minimal proofs of $S_n'$ in terms of the respective complexity measurement instead 
of computing the exact complexity. Moreover, we will show that minimal proofs of 
$S_n$ always have a smaller complexity than minimal proofs of $S_n'$ no matter 
which complexity measurement we choose. 
\begin{lemma}
\label{lem:SntoSn'}
Let $\pi$ be a minimal proof of the sequent $S_n$ of the example of Section \ref{sec:Example} in terms of quantifier, logical, or symbol complexity and $\pi'$ a minimal proof of the sequent $S_n'$ in prenex normal form in terms of quantifier, logical, and symbol complexity, respectively then 
\begin{align*}
|\pi |_\Diamond\le |\pi '|_\Diamond
\end{align*}
where $\Diamond\in\{ q,l,s\}$. 
\end{lemma}
\begin{proof}[Proof sketch]
Each minimal proof of $S_n$ can be transformed into a minimal proof of $S_n'$. 
This transformation will at most add inferences and, therefore, the respective 
complexity can only increase. 
\end{proof}
Before we compute the complexities of minimal proofs of $S_n$ we have to show 
some properties of a potential minimal proof. In a first step we show that in a minimal 
proof (with respect to an arbitrary complexity measurement) all atoms that appear in an 
instantiation of $A_n,B,C_n,$ or $D$ are active in an axiom. 
\begin{lemma}
\label{lem:allACT}
Let $\pi$ be a minimal proof in terms of quantifier, logical, or symbol complexity of the sequent $S_n$ of the example of Section \ref{sec:Example} and 
\begin{align*}
& P(a,f_1a)\lor\ldots\lor P(a,f_na), \\ 
& P(b_1,b_2)\to P(b_1,fb_2), \\ 
& \bigr( P(c_1,fc_2)\land P(fc_2,fc_3)\land\ldots\land P(fc_{n-1},fc_n)\bigr)\to 
P(c_1,gc_n), \text{ and} \\ 
& P(d_1,gd_2)
\end{align*}
be instantiations of $A_n,B,C_n,$ and $D$ for some proof-specific terms $a,b_1,b_2,$ $c_1,\ldots ,c_n,d_1,$ and $d_2$. Then there are axioms for each atom  
\begin{align*}
& P(a,f_1a),\ldots ,P(a,f_na),P(b_1,b_2),P(b_1,fb_2), \\ 
& P(c_1,fc_2),P(fc_2,fc_3),\ldots ,P(fc_{n-1},fc_n),P(c_1,gc_n), \text{ and} \\ 
& P(d_1,gd_2)
\end{align*}
in which the respective atom is active. 
\end{lemma}
\begin{proof}
The proof works for all four formulas in a similar way. We will only consider the formula 
\begin{align*}
A_n^{\downarrow a}:= &\; P(a,f_1a)\lor\ldots\lor P(a,f_na).
\end{align*}
\\
Assume there is an $i\in\mathbb{N}_n$ such that $P(a,f_ia)$ is not active in any  
axiom. Then we can order $\pi$ such that the $\lor\colon$l-rules that apply to 
$A_n^{\downarrow a}$ 
are the rules in the new minimal proof $\pi '$ that appear at the top of the 
corresponding proof tree. Let $S:=A_n^{\downarrow a},\Delta\seq\Gamma$ be 
a sequent in which $A_n^{\downarrow a}$ appears. The provability implies that also 
$P(a,f_ia), 
\Delta\seq\Gamma$ is a tautological axiom. Hence, $\Delta\seq\Gamma$ is already 
tautological and we can drop all the $\lor\colon$l-rules applied to 
$A_n^{\downarrow a}$ (and 
even the instantiation rules). Thus, there is a proof with smaller quantifier, 
logical, and symbol complexity which contradicts the assumption that $\pi$ was already 
minimal in these terms. Hence, there is no such instantiation. 
\end{proof}
\begin{remark}
This is not a general property of minimal proofs. Consider, for instance, the proof 
\begin{center}
\begin{fCenter}
\Axiom$ \fCenter $
\RightLabel{\textit{Ax}}
\UnaryInf$\forall x.P(x)\land Q(x),P(c), Q(c) \;\fCenter\seq P(c)$
\RightLabel{$\land\colon$l}
\UnaryInf$\forall x.P(x)\land Q(x),P(c)\land Q(c) \;\fCenter\seq P(c)$
\RightLabel{\alll}
\UnaryInf$\forall x.P(x)\land Q(x) \;\fCenter\seq P(c)$
\DisplayProof
\end{fCenter}
\end{center}
of the sequent $\forall x.P(x)\land Q(x)\seq P(c)$. The proof is minimal, but the atom $Q(c)$ is not active.
\end{remark}
The next property guarantees that $A_n,B,C_n,$ and $D$ have to be instantiated at least 
once. 
\begin{lemma}
\label{lem:allINS}
Let $\pi$ be a proof of the sequent $S_n$ of the example of Section \ref{sec:Example} then the formulas 
\begin{align*}
& P(a,f_1a)\lor\ldots\lor P(a,f_na), \\ 
& P(b_1,b_2)\to P(b_1,fb_2), \text{ and} \\ 
& \bigr( P(c_1,fc_2)\land P(fc_2,fc_3)\land\ldots\land P(fc_{n-1},fc_n)\bigr)\to 
P(c_1,gc_n) 
\end{align*}
with some proof-specific terms $a,b_1,b_2,c_1,\ldots ,c_n$ appear on the left side of 
some sequents in $\pi$ and the formula 
\begin{align*}
& P(d_1,gd_2)
\end{align*}
with proof-specific terms $d_1,d_2$ appears on the right side of some sequent in $\pi$. 
\end{lemma}
\begin{proof}
First of all at least one formula has to be instantiated. Otherwise, there cannot 
be a valid proof. By showing that an instantiation of an arbitrary formula enforces all 
other formulas to be instantiated at least once we will complete the proof. This can 
easily be seen by Lemma \ref{lem:allACT} and the facts that all potential atoms of $A_n$ 
can only build valid axioms with potential atoms of $B$ ($P(a,f_ia),\Delta\seq\Gamma 
,P(b_1,b_2)$ with $a=b_1$ and $f_ia=b_2$), all potential atoms of $B$ has to build 
axioms with $A_n$ and $C_n$, and so on. In the end we have to 
instantiate $A_n,B,C_n,$ and $D$. 
\end{proof}
Now we can describe sets of instantiations that belong to a minimal proof of $S_n$. 
We will not write down the whole proof because of its large size. But by proving 
the minimality of this instantiations we will implicitly give a sketch of the proof 
and show its validity. 
\begin{theorem}
\label{the:minInstFree}
Let $n$ be a fixed natural number and $S_n=A_n,B,C_n\seq D$ be given. Then the sets 
\begin{align*}
\mathbf{A}_n^1:=&\; \{ c\} 
,\;\mathbf{A}_n^2:= \{ fh_1c\; |\; h_1\in\{ f_1,\ldots ,f_n\}\} , \\
\mathbf{A}_n^i:=&\; \{ fh_{i-1}\ldots fh_1c\; |\; h_1,\ldots ,h_{i-1}\in 
\{ f_1,\ldots ,f_n\}\}\text{ for }i\in\{ 3,\ldots ,n-1\} , \\
\mathbf{A}'_n:=&\; \bigcup\limits_{i=1}^{n-1}\mathbf{A}_n^i ,\; 
\mathbf{B}':= \{ (t,f_it)\; |\; t\in\mathbf{A}_n'\textbf{ and }i\in\{ 1,\ldots ,n\}\} , 
\\ 
\mathbf{C}'_n:=&\; \{ (t_1,\ldots ,t_{n+2})\; |\; t_1=c\textbf{ and } 
t_2=h_1t_1\textbf{ and } 
t_3=h_2ft_2\textbf{ and }\ldots \\ 
&\; \ldots \textbf{ and }t_n=h_{n-1}ft_{n-1} 
\textbf{ and }t_{n+1}=t_1\textbf{ and }t_{n+2}=t_n\textbf{ and } \\ 
&\; h_1,\ldots ,h_{n-1}\in\{ f_1,\ldots ,f_n\}\} , \\
\mathbf{D}':=&\; \{ (c,t)\; |\; t=t_n\textbf{ and }\exists t_1,\ldots ,t_{n-1}. 
(t_1,\ldots ,t_n)\in\mathbf{C}'_n\} .
\end{align*}
are instantiations of the formulas $A_n,B,C_n,$ and $D$ such that the corresponding 
fully instantiated sequent $S_n^{\downarrow}$ is tautological and there is a minimal 
(in terms of quantifier, logical, or symbol complexity) proof $\pi$ of $S_n$ with the 
midsequent $S_n^{\downarrow} $. 
\end{theorem}
\begin{proof}
By Lemma \ref{lem:allINS} we can assume an instantiation $(t_1,\ldots ,t_n)$ of $C_n$. 
Let $c:=t_1$. Given that atomic subformulas of an instantiated formula in a minimal 
proof have to be active (see Lemma \ref{lem:allACT}) we know that $P(c,ft_2)$ of 
\begin{align*}
\bigl( P(c,ft_2)\land P(ft_2,ft_3)\land\ldots\land P(ft_{n-1},t_n)\bigr)\to P(t_1,gt_n)
\end{align*}
has to be active in an axiom. In an axiom $P(c,ft_2)$ appears on the right side of the 
sequent and, hence, the only formula that can become $P(c,ft_2)$ on the left side of 
the sequent is $P(b_1,fb_2)$ of 
\begin{align*}
B=P(b_1,b_2)\to P(b_1,fb_2).
\end{align*}
Then $b_1$ has to be equal to $c$ and $b_2$ has to be equal to $t_2$. By applying Lemma \ref{lem:allACT} again we have to find the counterpart of $P(b_1,b_2)=P(c,t_2)$. 
Hence, there has to be an instantiation of $A_n$, i.e.\ 
\begin{align*}
P(c,f_1c)\lor\ldots\lor P(c,f_nc).
\end{align*}
Given that this is the only possibility we can conclude that there have to be 
instantiations of $B$ and $C_n$ where $t_2$ is equal to $f_1c,\ldots ,f_{n-1}c,$ and 
$f_nc$. 
\\
So far we described $\mathbf{A}_n^1$, the parts of $\mathbf{B}'$ where $\mathbf{A}_n'$ 
is replaced with $\mathbf{A}_n^1$, the first two elements of the tuples in 
$\mathbf{C}_n'$, and the first element of the tuples in $\mathbf{D}'$. With the 
second elements $f_1c,\ldots ,f_nc$ of the tuples in $\mathbf{C}_n'$ we have to go 
through the 
same procedure as we did with $c$. That is, we will get new instantiations of 
$A_n$, i.e.\ $\mathbf{A}_n^2$, a new part of $\mathbf{B}'$ and the third 
elements of tuples in $\mathbf{C}_n'$. After $n$ applications of this procedure 
we would have constructed the sets of the theorem such that each atom has exactly one 
necessary counterpart, i.e.\ all atoms appear as an active formula in an axiom 
and we cannot drop a single atom without making the proof invalid. Hence, the 
instantiation correspond to a minimal proof of $S_n$ in terms of quantifier 
complexity. Given that all proofs contain at least as many instantiations as 
the given one there also has to be a corresponding minimal proof in terms of logical 
and symbol complexity. 
\end{proof}
Now we can compute the quantifier complexity of $S_n$. Let $\mathbf{A}_n',\mathbf{B}', 
\mathbf{C}_n',\mathbf{D}',$ and $S_n^{\downarrow}$ be defined as in Theorem \ref{the:minInstFree} then 
\begin{align*}
|S_n^{\downarrow}|_q=(|\mathbf{A}_n'|)+(|\mathbf{B}'|+\sum\limits_{i=1}^{n}n^{i-1})+
(|\mathbf{C}_n'|+\sum\limits_{i=1}^{n}n^{i-1})+(|\mathbf{D}'|+1). 
\end{align*}
The additional instantiations, besides the ones covered by $|\mathbf{A}_n'|+
|\mathbf{B}'|+|\mathbf{C}_n'|+|\mathbf{D}'|$, derive from the number of variables 
in each formula $A_n,B,C_n,$ and $D$. Given that 
\begin{align*}
|\mathbf{A}_n'|=&\; \sum\limits_{i=1}^{n}n^{i-1} , \\
|\mathbf{B}'|=&\; n\cdot |\mathbf{A}_n'|=n\cdot 
\sum\limits_{i=1}^{n}n^{i-1}, \\ 
|\mathbf{C}_n'|=&\; n^{n-1} ,\text{ and} \\
|\mathbf{D}'|=&\; |\mathbf{C}_n'|=n^{n-1} 
\end{align*}
the quantifier complexity sums up to 
\begin{align*}
|S_n^{\downarrow}|_q=&\; n^n+6\cdot n^{n-1}+4\cdot n^{n-2}+\ldots +4\cdot n +5 >n^n 
\end{align*}
for $n\ge 3$ and 
\begin{align*}
|S_n^{\downarrow}|_q=&\; n^n+6\cdot n^{n-1}+5 >n^n 
\end{align*}
for $n=2$. 
\\
By Lemma \ref{lem:SntoSn'} we can give a lower bound for the quantifier complexity of $S_n'$. Moreover, the quantifier complexity is a lower bound for the logical complexity and the symbol complexity (see Proposition \ref{pro:orderCompl}). To summarize: the various complexities of minimal proofs of $S_n'$ are bigger than $n^n$.

\subsection{A Proof Scheme with a $\Pi$-2-cut}

After computing the complexity of a minimal cut-free proof, we want to generate a cut formula by the presented method and analyse the complexity of the corresponding proof with cut. We consider the scheme of schematic $\Pi_2$-grammars $\Gcal_n$. $\Gcal_n$ is defined by the starting symbol $\tau$, the non-terminals $\tau,\beta_1,\ldots ,\beta_{n-1},\alpha$, and the production rules 
\begin{align*}
\tau\to &\; h_{F_n}(\alpha ,\alpha ,f_1\alpha ) \; |\;\ldots\; |\; 
h_{F_n}(\alpha ,\alpha ,f_n\alpha ) \; |\; h_{G_n}(c,\beta_1,\ldots 
\beta_{n-1},c,\beta_{n-1}) 
, \\ 
\alpha\to &\; f\beta_{n-1}\; |\;\ldots\; |\; f\beta_1\; |\; c, \\ 
\beta_{n-1}\to &\; f_1f\beta_{n-2}\; |\;\ldots\; |\; f_nf\beta_{n-2}\; , \\ 
&\; \vdots \\ 
\beta_2\to &\; f_1f\beta_1\; |\;\ldots\; |\; f_nf\beta_1, \text{ and } \\
\beta_1\to &\; f_1c\; |\;\ldots\; |\; f_nc 
\end{align*}
where $h_{F_n}$ and $h_{G_n}$ are function symbols that correspond to the 
$\lambda$-terms $\gamma_n=\lambda \tuples{x}.A_n'\land B'$ and $\delta_n=\lambda 
\tuples{y}.C_n'\lor D'$. Note that the language $L(\Gcal_n)$ of $\Gcal_n$ covers the 
Herbrand term set that can be derived from the instantiations of Section \ref{sub:Minimal cut-free proofs}. The leaves of a maximal \textbf{G3c}-derivation $\psi_n$ of 
the reduced representation 
\begin{align*}
(\lambda\tuples{x}.A'_n\land B')\alpha\alpha (f_1\alpha ),\ldots , 
(\lambda\tuples{x}.A'_n\land B')\alpha\alpha (f_n\alpha ) \\ 
\;\;\;\;\;\;\;\;\;\;\;\;\;\;\; 
\seq (\lambda\tuples{y}.C'_n\lor D')c\beta_1\ldots \beta_{n-1}c\beta_{n-1} 
\end{align*}
can be represented in the normal form of Definition \ref{def:dual sequent operator}, 
i.e.\ 
\begin{align*}
\{ &\; P(\alpha ,h\alpha ), \\ 
&\;\{\neg P(\alpha ,f_i\alpha )\; |\; i\in I_1\} , 
\{ P(\alpha ,ff_i\alpha )\; |\; i\in I\backslash I_1\} , \\ 
&\;\{\neg P(c,f\beta_1)\; |\; \text{if }j=1\} , \\ 
&\;\{\neg P(f\beta_1,f\beta_2)\; |\; \text{if }j=2\} ,\ldots 
,\{\neg P(f\beta_{n-2},f\beta_{n-1})\; |\; \text{if }j=n-1\}, \\ 
&\;\{P(c,g\beta_{n-1})\; |\; \text{if }j=n\}, \\ 
&\;\neg P(c,g\beta_{n-1}) \seq \\
& |\; h\in \{ f_1,\ldots ,f_n\},I=\mathbb{N}_n,I_1\subseteq I,j\in I\;\} 
\end{align*}
and the non-tautological leaves are 
\begin{align*}
\text{DNTA}(\psi_n )= &\; \{\; P(\alpha ,f_i\alpha ), \\ 
&\;\;\;\; \{\neg P(\alpha ,f_l\alpha )\; |\; l\in I_1\}, 
P(\alpha ,ff_i\alpha ),\{ P(\alpha ,ff_k\alpha )\; |\; k\in I\backslash I_1\}, \\ 
&\;\;\;\; \neg P(f\beta_{j-1},f\beta_j), \\ 
&\;\;\;\; \neg P(c,g\beta_{n-1})\seq \\ 
&\;\;\; |\; i\in\mathbb{N}_n ,j\in\mathbb{N}_{n-1} , 
I_1\subseteq I=\mathbb{N}_n \backslash i,f\beta_0:=c \;\;\} . 
\end{align*}
In each leaf there is a least one conjunct of $A_n$ (first line of DNTA$(\psi )$). 
Hence, if we branch $B$ with the corresponding term of the chosen disjunct only the 
branch containing the succedent of this $B$ is not a tautology (second line of 
DNTA$(\psi )$). Given that each leaf contains the instantiation of $D$ (fourth 
line of DNTA$(\psi )$) we have to look at the branch containing the antecedent of 
the instantiation of $C_n$. Otherwise the leaf is a tautology. The antecedent is a 
conjunction that moves to the right of the sequent after branching $C_n$ and 
,therefore, we have to pick an arbitrary conjunct (third line of DNTA$(\psi )$). 
\\
Now we want to construct the maximal set of $G^*$-unified terms for each leaf. The only 
interactive literals are $(P(\alpha ,ff_i\alpha ),\neg P(f\beta_{j-1},f\beta_j))$ with 
$i\in\mathbb{N}_n$ and $j\in\mathbb{N}_{n-1}$. The maximal set of 
$G^*$-unified terms is $\{ P(x,fy)\}$ accordingly. The starting set of $G^*$-unifiable 
sequents is then given by 
\begin{align*}
& \mathcal{U}:=\{ U\; |\;U\subseteq\{ P(x,fy)\}\} =\{\{ P(x,fy)\} ,\emptyset\} .
\end{align*}
The empty set can be ignored. 
The set of possible sets of clauses is then $Cl(\{\{ P(x,fy)\}\} )= 
\{\{ P(x,fy)\}\}$ and the set of solution candidates is 
\begin{align*}
Sol(\{\{ P(x,fy)\}\} )= \{\{ P(x,fy)\}\} 
\end{align*} 
which is independent from $n$. Let 
\begin{align*}
[P(\alpha ,ff_i\alpha )]_{i=1}^n:=P(\alpha ,ff_1\alpha ),\ldots 
,P(\alpha ,ff_n\alpha ). 
\end{align*}
Then we find the correct cut-formula and the proof $\pi_n$ with cut can be sketched 
by 
\begin{center}
\begin{fCenter}
\Axiom$ \;\fCenter \vdots $
\UnaryInf$\Gamma^l,\mathcal{A}_n',\mathcal{B}_n' 
\;\fCenter\seq \Delta_{n+1}^l ,[P(\alpha ,ff_i\alpha )]_{i=1}^n$
\LeftLabel{\alll}
\UnaryInf$ \;\fCenter \vdots $
\LeftLabel{\alll}
\UnaryInf$\forall \tuples{x}.A_n'\!\land\! B' 
\;\fCenter\seq \Delta_{n+1}^l ,[P(\alpha ,ff_i\alpha )]_{i=1}^n$
\LeftLabel{\exr}
\UnaryInf$ \fCenter \vdots $
\LeftLabel{\exr}
\UnaryInf$\forall \tuples{x}.A_n'\!\land\! B' \;\fCenter\seq \Delta_{1}^l ,\exists y.P(\alpha ,fy)$
\LeftLabel{\allr}
\UnaryInf$\forall \tuples{x}.A_n'\!\land\! B' \;\fCenter\seq \exists\tuples{y}.\overline{C_n'}\!\lor\! D',\forall x\exists y.P(x,fy)$
\Axiom$ \fCenter \pi_n' $
\LeftLabel{\textit{Cut}}
\BinaryInf$\forall \tuples{x}.A_n'\!\land\! B' \;\fCenter\seq 
\exists\tuples{y}.\overline{C_n'}\!\lor\! D'$
\DisplayProof
\end{fCenter}
\end{center}
where $\pi_n'=$
\begin{center}
\begin{fCenter}
\Axiom$ \;\fCenter \vdots $
\UnaryInf$P(c,f\beta_1),\ldots ,P(f\beta_{n-2},f\beta_{n-1}),\Gamma_{2n-2}^l 
\;\fCenter\seq \Delta^r ,\overline{\mathcal{C}_n'},\mathcal{D}_n$
\LeftLabel{\exr}
\UnaryInf$ \;\fCenter \vdots $
\LeftLabel{\exr}
\UnaryInf$P(c,f\beta_1),\ldots ,P(f\beta_{n-2},f\beta_{n-1}),\Gamma_{2n-2}^l \;\fCenter\seq 
\exists\tuples{y}.\overline{C_n'}\!\lor\! D'$
\LeftLabel{\exl}
\UnaryInf$ \fCenter \vdots $
\LeftLabel{\alll}
\UnaryInf$P(c,f\beta_1),\Gamma_{2}^l \;\fCenter\seq 
\exists\tuples{y}.\overline{C_n'}\!\lor\! D'$
\LeftLabel{\exl}
\UnaryInf$\exists y.P(c,fy),\forall x\exists y.P(x,fy),\Gamma_{1}^l \;\fCenter\seq 
\exists\tuples{y}.\overline{C_n'}\!\lor\! D'$
\LeftLabel{\alll}
\UnaryInf$\forall x\exists y.P(x,fy),\forall \tuples{x}.A_n'\!\land\! B' \;\fCenter\seq 
\exists\tuples{y}.\overline{C_n'}\!\lor\! D'$
\DisplayProof
\end{fCenter}
\end{center}
and $\mathcal{A}_n',\mathcal{B}_n',\overline{\mathcal{C}_n'},$ and $\mathcal{D}_n$ are 
instantiations of $\forall \tuples{x}.A_n'\land B'$ and $\exists\tuples{y}. 
\overline{C_n'}\lor D'$. 
To check the correctness of the proofs we look at the instantiations and the leaves 
of a maximal \textbf{G3c}-derivation. Let $\beta_1,\ldots ,\beta_{n-1},$ and $\alpha$ 
be the derivations of the cut-formula then 
\begin{align*}
\mathbf{A}':= &\; \{\alpha\} ,\; 
\mathbf{B}_n':= \{ (\alpha,f_i\alpha )\; |\; i\in\{ 1,\ldots ,n\}\} ,\\
\mathbf{C}_n':= &\; \{ (c,\beta_1,\ldots ,\beta_{n-1})\} ,\;
\mathbf{D}_n':= \{ (c,\beta_{n-1})\} 
\end{align*}
are the sets of instantiations. We can define the substituted formulas as follows
\begin{align*}
\mathcal{A}_n':= &\; \{ (\lambda x.A_n')t\; |\; t\in\mathbf{A}'\} ,\; 
\mathcal{B}_n':= \{ (\lambda \tuples{x}.B')\tuples{t}\; |\; \tuples{t}\in\mathbf{B}_n'\} ,\; \\ 
\mathcal{C}_n':= &\; \{ (\lambda \tuples{x}.C_n')\tuples{t}\; |\; \tuples{t}\in\mathbf{C}_n'\} 
,\;\overline{\mathcal{C}_n'}:= \{ P\; |\; \overline{P}\in\mathcal{C}_n'\} ,\; 
\mathcal{D}_n':= \{ (\lambda \tuples{x}.D')\tuples{t}\; |\; \tuples{t}\in\mathbf{D}_n'\} .
\end{align*}
The leaves of the left branch $\mathbb{L}_l$ in literal normal form (see Definition 
\ref{def:dual sequent operator}) are 
\begin{align*}
\{ &\; P(\alpha ,h\alpha ), 
\{\neg P(\alpha ,f_i\alpha )\; |\; i\in I_1\} , 
\{ P(\alpha ,ff_i\alpha )\; |\; i\in I\backslash I_1\}, \\ 
&\; \{\neg P(\alpha ,ff_i\alpha )\; |\; i\in I\}  \seq\;\; |\; 
h\in \{ f_1,\ldots ,f_n\},I=\mathbb{N}_n,I_1\subseteq I \;\} 
\end{align*}
and the leaves of the right branch $\mathbb{L}_r$ in literal normal form are 
\begin{align*}
\{ &\; \{\neg P(c,f\beta_1)\; |\; \text{if }j=1\} , \\ 
&\; \{\neg P(f\beta_1,f\beta_2)\; |\; \text{if }j=2\} ,\ldots 
,\{\neg P(f\beta_{n-2},f\beta_{n-1})\; |\; \text{if }j=n-1\}, \\ 
&\;\{P(c,g\beta_{n-1})\; |\; \text{if }j=n\}, \\
&\;\neg P(c,g\beta_{n-1}), \\ 
&\;\{ P(f\beta_{j-1} ,f\beta_i)\; |\; i\in \mathbb{N}_{n-1}\text{ and } 
f\beta_0:=c\} \seq \\
& |\; h\in \{ f_1,\ldots ,f_n\},I=\mathbb{N}_n,I_1\subseteq I,j\in I\;\} .
\end{align*}
Let us assume a leaf $\mathcal{L}$ of the set $\mathbb{L}_l$. It contains an atom 
$P(\alpha ,f_k\alpha )$ for a given $k\in\mathbb{N}_n$. If $k\in I_1$ then 
$\mathcal{L}$ contains also $\neg P(\alpha ,f_k\alpha )$ and is therefore a tautology. 
Let us assume $k\in I\backslash I_1$. Then $P(\alpha ,ff_k\alpha )$ is an element of 
$\mathcal{L}$. But each leaf contains the set $\{\neg P(\alpha ,ff_i\alpha )\; 
|\; i\in I\}$, i.e.\ $\mathcal{L}$ contains also $\neg P(\alpha ,ff_k\alpha )$ and is 
a tautology.

Let us assume a leaf $\mathcal{L}$ of the set $\mathbb{L}_r$. Then it contains 
the set $\{ P(r_i,f\beta_i)\; |\; i\in\mathbb{N}_{n-1}\}$. If 
$j\in\mathbb{N}_{n-1}$ we get the dual of an 
element of $\{ P(r_i,f\beta_i)\; |\; i\in\mathbb{N}_{n-1}\}$. If $j=n$ the leaf 
contains $P(c,g\beta_{n-1})$ and $\neg P(c,g\beta_{n-1})$. Hence, all leaves in 
$\mathbb{L}_r$ are tautologies and thus, the proof scheme is a correct. 

\smallskip

Now we compute the quantifier complexity of the proof sequence. Let 
$|\mathcal{A}_n\land \mathcal{B}_n|$ denote the number of instantiations of 
$\forall \tuples{x}.A_n'\land B'$ and $|\mathcal{C}_n\lor\mathcal{D}_n|$ denote the number 
of instantiations of $\exists\tuples{y}.\overline{C_n'}\lor D'$.
The number of instantiations of the end-sequent is 
\begin{align*}
|\mathcal{A}_n\land\mathcal{B}_n|= &\; n+2 ,\; 
|\mathcal{C}_n\lor\mathcal{D}_n|= n+2 ,\; 
|\mathcal{A}_n\land\mathcal{B}_n|+|\mathcal{C}_n\lor\mathcal{D}_n|= 2\cdot n+4  
\end{align*}
and the number of instantiations of the cut-formula is $2\cdot n-1$, and therefore the 
quantifier complexity of the proof is 
\begin{align*}
|\pi_n|_q=|\mathcal{A}_n\land\mathcal{B}_n|+|\mathcal{C}_n\lor\mathcal{D}_n|+ 
(2\cdot n-1)= &\; 4\cdot n+3 \in O(n) . 
\end{align*}
The logical complexity can easily be verified by counting. Hence, the number of 
inferences is 
\begin{align*}
|\pi_n|_l= (n\cdot n^{1})\cdot n\cdot n+n+1+1=n^4+n+2 . 
\end{align*}
To give an upper bound on the symbol complexity we have to compute the maximal symbol 
complexity of the sequents appearing in the proofs. This depends heavily on the used 
sequent calculus and the order of the proofs. Therefore, we will assume a polynomial 
function $\mathcal{P}(\cdot )$ that maps from natural numbers to natural numbers such 
that the maximal size of each sequent in the proofs is smaller than $\mathcal{P}(n)$. 
The interested reader is invited to prove the existence of such a function. Given 
$\mathcal{P}$ we can define the upper polynomial bound  
\begin{align*}
|\pi_n|_s\le 2\cdot\mathcal{P}(n)\cdot |\pi_n|_q+|\pi_n|_q.
\end{align*}

\smallskip

While the complexity in terms of logical inferences, in terms of weak quantifier 
inferences, or in terms of symbol complexity is bigger than $n^n$ 
for the cut-free proofs the introduction of the 
$\Pi_2$-cut decreases the complexity by an exponential factor.

\section{Implementation and Experiments}\label{sec.experiments}

We already saw in the previous section that the starting set for the $\Gcal^*$-unified literals in the example of Section \ref{sec:Example} consists only of a single literal. For the purpose of measuring the size of this set and testing the applicability, we implemented the method in the \textbf{GAPT} framework~\cite{Ebner2016System} (\textbf{G}eneral \textbf{A}rchitecture for \textbf{P}roof \textbf{T}heory). A prototype of the algorithm is released with version $2.5$ of \textbf{GAPT} which is available at \href{https://www.logic.at/gapt}{https://www.logic.at/gapt}. Until the up-to-date version of the algorithm is contained in a new \textbf{GAPT} release, it can be found in the branch \href{https://github.com/gapt/gapt/tree/pi2-cut-intro-clean-up}{pi2-cut-intro-clean-up} of the \emph{Git repository} \href{https://github.com/gapt/gapt}{https://github.com/gapt/gapt}. The results displayed in this section are based on this new version. Note that due to the early stage of the implementation, we do not have many examples. Furthermore, schematic $\Pi_2$-grammars have to be computed by hand. All cuts introduced by the algorithm contain only cut formulas with at most one variable per quantifier. Therefore, we only give an intuition of the method's applicability by comparing the number of non-tautological leaves (not counting duplications and supersets of leaves) with the number of $\Gcal^*$-unified literals, and the runtime for computing the cut-formula and the proof with cut:
\begin{center}
\begin{tabular}{c|c|c}
\textbf{Limited number of} & $\Gcal^*$\textbf{-unified} & \textbf{Runtime} \\
\textbf{non-tautological leaves} & \textbf{literals} & \\
\hline
 2 & 2 & 884ms \\
 16 & 4 & 951ms \\
 44 & 4 & 1s 94ms \\
 86 & 3 & 1s 953ms \\
 308 & 6 & 2s 292ms \\
 1386 & 6 & 24s 652ms
\end{tabular}
\end{center}
This tests give just an impression of the algorithmic behaviour, but it shows that the number of $\Gcal^*$-unified literals stays small. The runtime for constructing the cut formula for the given set of $\Gcal^*$-unified literals is by far the largest part of the whole runtime. This is due to a naive implementation of the characterization (see Definitions \ref{def:Set of possible sets of clauses} and \ref{def:solution set}). Note that the characterization computes a set of solutions while only one is needed. Here, improvements have to be found. The construction of $\Gcal^*$-unified literals, however, is fast and shows that the method is feasible in practice.

\section{Conclusion}

In this paper we extended the current range of algorithmic cut introduction from 
$\Pi_1$-cuts to $\Pi_2$-cuts. While any $\Pi_1$-grammar specifying the set of Herbrand 
instances of a cut-free proof yields a solution of the corresponding 
$\Pi_1$-cut-introduction problem (the canonical solution) this does not hold 
for schematic $\Pi_2$-grammars and $\Pi_2$-cuts; in Section~\ref{sec.cutintro} we have 
presented a schematic $\Pi_2$-grammar specifying a set of Herbrand instances which is 
not solvable in the sense that the corresponding schematic extended Herbrand sequent 
(representing the cut-introduction problem) does not have a solution. As for the 
$\Pi_2$ case canonical solutions do not exist in general we have chosen a different 
approach to compute $\Pi_2$-cuts corresponding to given schematic $\Pi_2$-grammars by 
characterizing the validity of cut formulas. 
Given a starting set (a set of sets of literals with two designated free 
variables $x,y$ of the intended $\Pi_2$-cut formula $\forall x\exists y.A$) we have 
developed a method to decide whether this 
starting set contains a logical equivalent version of such a formula $A$. However, 
the general problem to decide whether such a 
starting set exists at all remains unsolved.  But in case {\em balanced solutions} exist 
appropriate starting sets can be defined. However, the straightforward method to 
construct naive starting sets for balanced problems is computationally  
inefficient. To improve the resulting cut-introduction method we developed a 
unification method which yields much smaller sufficient starting sets for the 
computation of the cut-formulas. We have shown how to generalize our cut-introduction methods to
cut-formulas with blocks of quantifiers. Finally we have shown that our method of introducing 
(single) $\Pi_2$-cuts is capable of achieving an exponential proof compression: there 
exists a sequence of sequents having only cut-free proofs of at least exponential 
size for which the method based on $G^*$-unification efficiently generates a sequence of 
proofs of polynomial size with $\Pi_2$ cuts.

\smallskip

Concerning future work, we plan to improve the implementation of the $G^*$-unification method developed in this paper and to test it on proof data bases. This would benefit from an algorithmic construction of schematic $\Pi_2$-grammars, one of the problems on our research plan for the near future. For practical applications the implementation should be enriched by the use of effective heuristics, especially the construction of a cut formula for a set of $\mathcal{G}^*$-unified literals. There are also several open theoretical questions: is it possible to construct starting sets whenever there is a solution to the cut-introduction problem and to decide whether the $\Pi_2$-cut-introduction problem is solvable at all?  A positive answer would yield a decision procedure for the $\Pi_2$ cut-introduction problem, and (in case of solvability) a complete method to construct proofs with $\Pi_2$-cuts. So far our method can only deal with a single cut formula of the form $\forall\tuples{x} \exists\tuples{y}.A(\tuples{x},\tuples{y})$. An extension of the method to the introduction of several $\Pi_2$-cuts promises the same compression as can be obtained by a single $\Pi_3$-cut, i.e.\ a {\em super-exponential} one. Finally, a method to introduce $\Pi_n$-cuts could be capable of a nonelementary proof compression and would represent a long term goal of this research.

\bibliographystyle{apalike}
\bibliography{bibPi2CutIntroduction}

\begin{comment} 
\bibliographystyle{apalike}
\bibliography{/home/logic/Dokumente/Zwischenergebnisse/References/bib001}
%\bibliography{/home/LOGIC/lettmann/Dokumente/Zwischenresultate/References/bib001} 

\end{comment}

\end{document}